\newif\ifFullversion
\newif\ifShortversion

\Fullversiontrue % for eprint only. Not for conference version

\ifFullversion
\Shortversionfalse
\else
\Shortversiontrue
\fi
\PassOptionsToPackage{usenames,dvipsnames}{xcolor}

\ifFullversion
\documentclass[version=preprint]{iacrcc}
\else
\documentclass[conference]{IEEEtran}
\usepackage{hyperref}
\fi
\ifFullversion
\else
\usepackage{amsmath,amssymb,amsfonts,amsthm}
\newtheorem{definition}{Definition}
\newtheorem{theorem}{Theorem}
\newtheorem{remark}{Remark}

\usepackage[noabbrev,capitalise]{cleveref}
\fi

\usepackage{algorithm}
\usepackage{algpseudocode}
\usepackage{graphicx}
\usepackage{textcomp}
\usepackage{comment}
\usepackage{placeins}
\usepackage{mdframed}
\usepackage{subcaption}

\usepackage{tikz, tikz-qtree}
\usetikzlibrary{arrows.meta, positioning, calc, bending}

\usepackage[edges]{forest}

\usepackage{soul}

\usepackage[most]{tcolorbox}

\newtcbtheorem{assumption}{\bfseries Similarity Assumption}{enhanced,drop shadow={black!50!white},
  coltitle=black,
  top=0.3in,
  attach boxed title to top left=
  {xshift=1.5em,yshift=-\tcboxedtitleheight/2},
  boxed title style={size=small,colback=pink}
}{assumption}

\newtcbtheorem{application}{\bfseries Application}{enhanced,drop shadow={black!50!white},
  coltitle=black,
  top=0.3in,
  attach boxed title to top left=
  {xshift=1.5em,yshift=-\tcboxedtitleheight/2},
  boxed title style={size=small,colback=pink}
}{assumption}

\usepackage{comment}
\usepackage{enumitem}

\def\BibTeX{{\rm B\kern-.05em{\sc i\kern-.025em b}\kern-.08em
    T\kern-.1667em\lower.7ex\hbox{E}\kern-.125emX}}

\begin{document}

\title[plaintext={Towards Verifiable AI with Lightweight Cryptographic Proofs of Inference}]{Towards Verifiable AI with Lightweight Cryptographic Proofs of Inference$^\star$}
%\title{Towards Verifiable AI with Lightweight Proofs of Inference from Model Separation}

\ifFullversion
\addauthor[inst=1,
           %orcid={0000-1111-2222-3333},
           email={panchuri@offchainlabs.com}]{Pranay Anchuri}
\addauthor[inst={1,2},
           %orcid={1111-2222-3333-4444},
           email={binarywhalesinternaryseas@gmail.com}]{Matteo Campanelli}
\addauthor[inst=3,
           %orcid={1111-2222-3333-4444},
           email={pcesaretti@gradcenter.cuny.edu}]{Paul Cesaretti}
\addauthor[inst={1,3,4},
           %orcid={1111-2222-3333-4444},
           email={rosario@ccny.cuny.edu}]{Rosario Gennaro}
\addauthor[inst={3,4},
           %orcid={1111-2222-3333-4444},
           email={tjois@ccny.cuny.edu}]{Tushar M. Jois}
\addauthor[inst=4,
           %orcid={1111-2222-3333-4444},
           email={hkayman@ccny.cuny.edu}]{Hasan S. Kayman}
\addauthor[inst=3,
           %orcid={1111-2222-3333-4444},
           email={tozdemir@ccny.cuny.edu}]{Tugce Ozdemir}

\addaffiliation{Offchain Labs}
\addaffiliation{University of Tartu, Estonia}
\addaffiliation{CUNY Graduate Center}
\addaffiliation{City College of New York}

\authorrunning{Anchuri, Campanelli, Cesaretti, Gennaro, Jois, Kayman, Ozdemir}

\else
\author{\IEEEauthorblockN{Anonymous Authors}}

% \and
% \and
% \and
% \and
% \and

\fi

\maketitle

\renewcommand{\thefootnote}{$\star$}
\footnotetext{This work has been accepted for publication at the IEEE Conference on Secure and Trustworthy Machine Learning (SaTML). The final version will be available on IEEE Xplore.}
\renewcommand{\thefootnote}{\arabic{footnote}}

\newcommand{\needcite}{{\color{red}[{\bf Citation needed.}]\PackageWarning{colornotes}{Citation needed}}}
\newcommand{\colornote}[3]{{\color{#2}{[{\bf #1}: #3]}}\PackageWarning{colornotes}{Outstanding note from #1}}

% Paragraph headers
\newcommand{\parhead}[1]{\smallskip\noindent\textbf{#1}}
\newcommand{\subparhead}[1]{\smallskip\noindent\textit{#1}}

% Define your own commands!
\definecolor{lavender}{RGB}{125,85,199}
\newcommand{\tushar}[1]{{\colornote{Tushar}{lavender}{#1}}}

\definecolor{blue}{RGB}{0,0,225}
\newcommand{\paul}[1]{{\colornote{Paul}{orange}{#1}}}

\definecolor{yellow}{RGB}{150,150,0}
\newcommand{\kayman}[1]{{\colornote{Kayman}{yellow}{#1}}}
\newcommand{\kaymans}[1]{{{\sout{#1}}}}

\definecolor{green}{RGB}{34,139,34}
\newcommand{\tugce}[1]{{\colornote{Tugce}{green}{#1}}}

\definecolor{randomcolor}{RGB}{115,155,207}
\newcommand{\pranay}[1]{{\colornote{Pranay}{randomcolor}{#1}}}

\newcommand{\matteo}[1]{\colornote{Matteo}{Maroon}{#1}}

\newcommand{\act}{\ensuremath{\mathsf{Act}}}

\newcommand{\weight}{\ensuremath{\mathsf{Weight}}}

\newcommand{\bias}{\ensuremath{\mathsf{Bias}}}

\newcommand{\ar}{\ensuremath{\mathsf{Arch}}}

\newcommand{\stoch}{\ensuremath{\mathsf{SGD}}}

\newcommand{\test}{\ensuremath{\mathsf{PathTest}}}

\newcommand{\commit}{\ensuremath{\mathsf{Commit}}}

\newcommand{\open}{\ensuremath{\mathsf{Open}}}

\newcommand{\verify}{\ensuremath{\mathsf{Verify}}}

\newcommand{\layer}{\mathcal{L}}

\begin{abstract}
When large AI models are deployed as cloud-based services, clients have no guarantee that responses are correct or were produced by the intended model. Rerunning inference locally is infeasible for large models, and existing cryptographic proof systems---while providing strong correctness guarantees---introduce prohibitive prover overhead (e.g., hundreds of seconds per query for billion-parameter models). We present a verification framework and protocol that replaces full cryptographic proofs with a lightweight, sampling-based approach grounded in statistical properties of neural networks. We formalize the conditions under which trace separation between functionally dissimilar models can be leveraged to argue the security of verifiable inference protocols. The prover commits to the execution trace of inference via Merkle-tree-based vector commitments and opens only a small number of entries along randomly sampled paths from output to input. This yields a protocol that trades soundness for efficiency, a tradeoff well-suited to auditing, large-scale deployment settings where repeated queries amplify detection probability, and scenarios with rationally incentivized provers who face penalties upon detection. Our approach reduces proving times by several orders of magnitude compared to state-of-the-art cryptographic proof systems, going from the order of minutes to the order of milliseconds, with moderately larger proofs. Experiments on ResNet-18 classifiers and Llama-2-7B confirm that common architectures exhibit the statistical properties our protocol requires, and that natural adversarial strategies (gradient-descent reconstruction, inverse transforms, logit swapping) fail to produce traces that evade detection. We additionally present a protocol in the refereed delegation model, where two competing servers enable correct output identification in a logarithmic number of rounds.
\end{abstract}

\ifFullversion
\else
\fi

\newcommand{\qry}{\mathsf{qry}}
\newcommand{\outfn}{\mathsf{out}}
\newcommand{\model}{\mathcal{M}}
\newcommand{\Smodel}{\mathcal{S}_\mathsf{model}} %Domain for the model
\newcommand{\Sqry}{\mathcal{S}_\mathsf{query}}
\newcommand{\Sout}{\mathcal{S}_\mathsf{out}} 
\newcommand{\Dmodel}{\mathcal{D}_\mathsf{model}} %Distribution for the model
\newcommand{\Dqry}{\mathcal{D}_\mathsf{query}}

\newcommand{\Eval}{\mathsf{EvalTrace}}
\newcommand{\EvStr}{\mathsf{trc}} % String that describes an evaluation
\newcommand{\trace}{\EvStr}
\newcommand{\htrace}{\mathsf{htrc}}
\newcommand{\advtrace}{\widetilde{\EvStr}}
\newcommand{\yields}[1]{\rightsquigarrow{#1}}

\newcommand{\advmodel}{\widetilde{\model}}
\newcommand{\advDmodel}{\widetilde{\mathcal{D}}_\mathsf{model}}

\newcommand{\adv}{\mathcal{A}}

\newcommand{\NN}{\mathbb{N}}
\newcommand{\RR}{\mathbb{R}}

\newcommand{\sndeps}{\varepsilon}
\newcommand{\teps}{\varepsilon_\mathsf{tst}}

\newcommand{\bit}{\{ 0,1 \}}
\newcommand{\bits}{\bit^*}

\newcommand{\multilinePr}[2]{\Pr\left[\begin{aligned}#1\end{aligned} \;:\; \begin{aligned}#2\end{aligned}\right]}
\newcommand{\multilinePrSingle}[2]{\Pr_{#2}\left[\begin{aligned}#1\end{aligned}\right]}

\newcommand{\genparams}{\mathsf{GenParams}}
\newcommand{\pp}{\mathsf{pp}}
\newcommand{\prove}{\mathsf{Prove}}
\newcommand{\testTrace}{\mathsf{TestTrace}}
\newcommand{\modelcommit}{\mathsf{CommitToModel}}

\newcommand{\gentest}{\mathsf{Test}}

\newcommand{\sample}{\leftarrow_{\$}}

% Delta separation function
\newcommand{\Dfn}{\Delta}

\newcommand{\dsepT}{\delta^{\mathsf{trace}}_\mathsf{sep}} % for traces
\newcommand{\dsepO}{\delta^{\mathsf{out}}_\mathsf{sep}} % for outputs
\newcommand{\esep}{\varepsilon_\mathsf{sep}}

\newcommand{\pivc}{\Pi_{\mathsf{VecCom}}}
\newcommand{\piIdeal}{\Pi_\mathsf{ideal}}

\newcommand{\negl}{\mathsf{negl}}

\newcommand{\cm}{\mathsf{cm}}
\newcommand{\cmM}{\cm_\model}

\newcommand{\hash}{\mathsf{H}}

\newcommand{\Spos}{\mathcal{S}^{\mathsf{(pos)}}}
\newcommand{\SposP}{\bar{\mathcal{S}}^{\mathsf{(pos)}}}

\newcommand{\chalSpace}{\mathcal{C}}

\newcommand{\far}{\not\sim}

\newcommand{\paranoindent}[1]{\smallskip \noindent\textit{#1}}

\newcommand{\alphab}{\Sigma}

\newcommand{\nM}{{n_\mathsf{model}}}
\newcommand{\nQ}{{n_\qry}}
\newcommand{\nTrc}{{n_\mathsf{trc}}}
\newcommand{\idxsOut}{{\mathsf{Idxs}_\textit{out}}}

\section{Introduction}
\label{sec:intro}

The landscape of modern computing is dominated by deep neural networks which have seen tremendous success in a wide range of problem domains from biology and protein folding \cite{jumper2021highly}, to image classification and object arrangement in complex scenes with computer vision \cite{krizhevsky2012imagenet, girshick2014rich} culminating in complex general-purpose tasks via large language models \cite{geminiteam2025geminifamilyhighlycapable, deepseekai2025, openai2024gpt4technicalreport}. Due to their ubiquitous use as an external service via cloud-based platforms or APIs, there is urgency in developing protocols and methods to help safeguard this interface. In practice a client of these models will only receive an answer from an untrusted, possibly corrupted, server who could provide incorrect responses. A naive solution to this problem would be to allow the user to have complete access to the underlying model by downloading it from a verified site and executing the query on the user's home device. This is clearly infeasible as the required computing resources for these networks are too large for the average user's systems. 
The question, therefore, is:
\begin{center}
\emph{Can we efficiently verify the computation of a large AI model performed by an untrusted server with less resources than what is needed to actually run them?} 
\end{center}

%The general problem of \emph{Verifiable Outsourced Computation} has been extensively studied in the Theoretical Computer Science and Cryptography research communities for the last forty years, starting from the introduction of the concept of Interactive Proofs~\cite{GMR,Babai85}. 

Most proposals to build verifiable AI systems use Cryptographic Proofs, in particular \emph{Succinct Non-interactive Arguments of Knowledge} (or SNARKs)~\cite{ITCS:BCCT12,GGPR,parno2016pinocchio,groth2016size} (see, e.g.,~\cite{wang2022ezdps, weng2021mystique, kang2022scaling, feng2021zen,lee2024vcnn, liu2021zkcnn, weng2021mystique,qu2025zkgpt, sun2024zkllm, singh2024enhancing}). SNARKs allow the client to verify the correctness of an arbitrary computation from a short proof, achieving remarkably efficient verification (e.g., constant-size proofs and verification time, regardless of the size of the computation being proven). 

The challenge with using SNARKs, however, is that they incur a large overhead for proof generation, typically $1000\times$ or larger~\cite{Tha22a}. The reason is that they introduce expensive cryptographic operations for \emph{every} elementary computation step (for example, a gate in a circuit representation). In practice this leads to unacceptable delays in computing proofs: consider zkLLM~\cite{CCS:SunLiZha24}, where producing a proof of correct inference for a large language model with up to $13$B parameters takes around $15$ minutes. This is particularly problematic since verifiable inferences must be done at almost zero overhead, as predictions are computed many times and they are usually very fast to produce, so the proof computation has to be comparably fast\footnote{
This may not necessarily hold for a scheme for verifiable \textit{training}: training is a much more computationally intensive process, performed infrequently, and can tolerate slower proof production times. In our work we exclusively focus on proofs of inference, from an already correctly trained model.}. So, while work is ongoing to reduce the burden of using cryptographic approaches, full proofs on machine learning models using them still appear to be infeasible.

\parhead{Further motivation.} We briefly point out that efficient proofs for inference can be used to solve many other AI related security questions beyond outsourced computation. One of particular interest is \emph{model transparency}: the ability to prove that an AI system---one that has been ``certified'' as having been trained correctly over certain ``legal'' data---is indeed the one being used in the field to compute predictions. Another is \emph{fairwashing} in which a provider uses a biased model for users but substitutes an unbiased one when under audit~\cite{pmlr-v97-aivodji19a}. See also \cref{sec:motivation-other-model-snd}.

\subsection{Our Contributions}

We present a verification framework and protocols for \emph{Verifiable Inference} of large AI models. {\bf We take a lightweight approach to cryptographic proofs of inference and as a result are several orders of magnitude faster than previous protocols (e.g., milliseconds rather than minutes)}. Our technical approach has not been previously used in this setting and is of independent interest.

\smallskip
\noindent
{\bf A formal framework for verifiable inference from trace separation:}
One of our main conceptual contributions is a theoretical framework that can be used by protocol designers to leverage statistical properties of computation traces to argue the security of verifiable inference protocols. Our starting observation, grounded in results in the AI literature on representational similarity, is that two models that produce sufficiently different outputs must also differ in their internal activation traces. We formalize this as a \emph{trace separation} property and show that, combined with an appropriate testing mechanism, it is sufficient to guarantee security (\cref{sec:framework}). Our framework is directly applicable to a relatively limited~---though still of practical interest~---security notion (which we call \emph{other-model soundness}), under which the client can detect whether the server used a model whose behavior differs substantially from the intended one.

\smallskip
\noindent
{\bf A concrete protocol and its experimental evaluation:}
We instantiate our framework with a concrete scheme that detects cheating via a small number of checks on random parts of the entire computation (in contrast to Cryptographic Proofs that check the correctness of every step yielding a proof of exact, rather than approximate, correctness). Key advantages of our construction are both simplicity and efficiency: the prover simply commits to the execution trace of the computation and then opens a small number of entries\footnote{One consequence of this ``cleartext opening'' of values is that our protocol is not \emph{zero-knowledge} and reveals information about the underlying model, still preserving our goal of enforcing correctness. If model confidentiality is an additional goal, in \cref{sec:concl} we discuss possible ways to add ZK to our protocol, that should still perform much better than current state of the art protocols.}. This enables our improvement in proving time, of several orders of magnitude, compared to prior solutions.

We extensively validate our protocol's assumptions and performance through a series of experiments. We show that common neural network architectures~-- classifiers and LLMs~-- have the statistical properties necessary to enable the testing of the correctness of the computation. We also investigate to which extent our techniques provide security in a stronger adversarial setting where cheating may follow \textit{arbitrary} strategies, rather than merely substituting a different model, and show that our test approach is resistant to natural attacks of this kind. Additionally, we provide a performance comparison on commodity hardware and highlight our significant efficiency gains over prior work. An intuitive explanation of our protocol is presented in \cref{sec:intuition}, with the full protocol given in \cref{sec:constr}.

\smallskip
\noindent
{\bf A protocol in the refereed model:} As an additional contribution, we present a protocol in the \emph{refereed model}~\cite{canetti2013refereed}, which assumes that two servers compute the inference and we are guaranteed that at least one does so honestly. If the servers make competing different claims, we present an efficient protocol to determine which one is correct which requires only a logarithmic (in the size of the model) number of steps. This can be found in \cref{sec:ref}.

\subsection{Related Work}
\label{sec: background}

Having discussed cryptographic proof approaches in the Introduction, we focus here on alternative methods.
The recent work in~\cite{svip} introduced SVIP, a verifiable inference protocol that works by training a secret-dependent proxy task. A trusted intermediary generates and distributes secrets to users. However, this approach for verifiable inference has notable overheads: i) it requires model-specific training that can be resource-intensive. ii) The protocol's security relies on active refresh of secrets to defend against secret recovery attacks, in which the provider attempts to steal the secrets by posing as a user and exploiting the proxy task. In contrast, our protocol provides a more general verification method. It operates directly between a prover and a verifier and does not depend on an intermediary to manage secrets. Furthermore, our protocol directly verifies the integrity of the inference process without requiring proxy tasks. This allows our protocol to naturally apply to a wide class of general deep neural networks beyond just LLMs.

AI Watermarking (see~\cite{AIwaterm} for a survey) embeds certain signals into the output of an AI model to identify that content as AI generated, through later detection. These techniques solve a different problem than the one we address: they verify that \emph{some} AI model was used, while our goal is to verify that a specific model was computed correctly.

Our main protocol can be seen as an example of \emph{proof of proximity} introduced in~\cite{RVW13}. Most of the literature for proofs of proximity is focused on increasing the efficiency of the Verifier (e.g., by allowing proof verification without even reading the entire input, but accessing only a sublinear portion of it). We, instead, focus on reducing the prover overhead.

Our protocol in the refereed model follows the basic bisection blueprint first proposed in~\cite{canetti2013refereed}. This paradigm has been successfully adopted by so-called optimistic rollup protocols to scale executions on the Ethereum blockchain (e.g., Arbitrum \cite{kalodner2018arbitrum}, Optimism \cite{armstrong2021ethereum}).
The application of the refereed model to verify the outputs of large-scale AI systems was initially proposed in~\cite{opt-ai}. While that work correctly identifies the inherent suitability of the refereed computation model for this domain---given the sublinear verification complexity and the often-extensive computational traces of large AI models---it does not provide empirical results or an experimental validation of this hypothesis, which we do.

\subsection{Outline}
The rest of the paper is organized as follows.
\cref{sec:intuition} provides an intuitive technical overview of our approach, including the key ideas behind our construction and the security notions we consider.
\cref{sec:prelim} introduces preliminaries and notation.
In \cref{sec:framework} we present a formal framework for verifiable inference, including our security definitions, a theoretical treatment linking machine learning properties to soundness, and a generic compiler from idealized to cryptographic constructions using vector commitments.
\cref{sec:constr} describes our specific construction and discusses its security properties.
\cref{sec:om-exps} and \cref{sec:fs-exps} present our experimental evaluation.
\cref{sec:perf} reports on our performance evaluation and a comparison with SNARK-based approaches to verifiable inference.
We conclude in \cref{sec:concl}.

%\section{An Informal Description of Our Approach}
\section{An Intuitive Technical Overview}
\label{sec:intuition}

A neural network is a directed acyclic graph (DAG), where nodes represent neurons and edges are the weighted connections between them. The network is configured by a training process that yields a vector of weights  $\left( w_{ij} \right)_{ij}$ assigned to the directed edge connecting neuron $i$ to neuron $j$.
The state of each neuron is represented by a scalar value known as its activation. During evaluation, the input data is encoded as the activation values of the network's source nodes. The activation of a subsequent neuron $j$ is determined by a function of the weights and the activation values of its parent neurons. The standard formulation for this process is:

$$ a_j = \phi\Big(\sum_{i\in G_j} w_{ij}a_i \Big)$$

where $a_j$ is the activation of neuron $j$, $G_j$ is the set of indices of its parent neurons, $w_{ij}$ is the weight on the edge from parent neuron $i$ to neuron $j$, and $\phi$ is a non-linear activation function.\footnote{Formulations often include a bias term $b_i$ added to $w_{ij}a_i$, but for ease of vectorization we consider $b_i$ an additional weight with an activation of 1. Also, neural networks can use different $\phi$, but we fix it here for simplicity.} The final output of the neural network is represented by the activation values of the terminal neurons in the DAG.
We refer to the list $\left( a_1, a_2, \dots \right)$ of activation values as the \textit{execution trace} of the model (or simply as trace).

\paragraph{Our methodology---simpler tests from similarity.} As stated earlier, our goal is to design an \textit{efficient} verification scheme for neural network inference, even for large AI models. In order to avoid the massive prover overhead of cryptographic proofs, we need to completely depart from their usual approach, which encodes the correctness of each step of the computation in an algebraic ``circuit-checking'' engine~\cite{nitulescu2020zk}.

Our starting point was the body of work in the AI research literature that has developed metrics to compare neural network representations (see~\cite{Klabunde_2025} for a comprehensive survey). The idea of comparing internal representations originated in neuroscience, where pairwise dissimilarity responses of brain activity across various stimuli are used to build structured comparisons \cite{kriegeskorte2008rsa}, and this technique was later translated into AI to compare internal representations of deep neural network structures~\cite{Kornblith2019SimilarityRevisited}. One of the motivations of this research is \emph{Explainable AI}, i.e., the desire to understand what happens during the execution of a neural network that allows it to ``understand'' or ``generate knowledge''.

Of particular interest to us were two notions of neural network similarity: \emph{functional similarity} and \emph{representational similarity}. Functional similarity measures the difference in the probability distributions of the outputs, while representational similarity focuses on quantifying similarity in the distributions of the internal activations. Importantly, as Klabunde et al. note~\cite{Klabunde_2025}, \emph{``When functional similarity measures indicate dissimilarity, representations must be dissimilar at some layer.''} In other words, two models that produce sufficiently different output distributions must be different also in the distributions of their internal activations.

This suggests an intuitive way to test the correctness of a claimed output. Assume a client asks an untrusted server to compute a model $\model_1$ on input $\qry$ and the server responds with $y$. If the server is lying by using a substantially different AI model $\model_2$ the client can hopefully detect this by testing the two internal representations on a small number of samples. While exploring the literature on similarity measures, however, we could not find a similarity notion which would be suitable for this task (i.e., which guarantees that if the two internal distributions are sufficiently far apart then we can detect it from a few samples).  We therefore developed our own approach.

\begin{figure}[t]
	\centering
	\begin{tikzpicture}[font=\small, >=Latex, scale=1.15, transform shape]
		
		% frame 
		\draw[blue!15, line width=0.8pt, rounded corners=3pt] (-0.3,1.10) rectangle (4.3,-1.0);
		
		% coordinates
		\node[circle, draw, fill=blue!8, inner sep=2pt] (N00) at (0,0.8) {};
		\node[circle, draw, fill=blue!8, inner sep=2pt] (N01) at (0,0.4) {};
		\node[circle, draw, fill=blue!8, inner sep=2pt] (N02) at (0,0) {};
		\node[circle, draw, fill=blue!8, inner sep=2pt] (N03) at (0,-0.4) {};
		\node[circle, draw, fill=blue!8, inner sep=2pt] (N04) at (0,-0.8) {};
		
		% L1
		\node[circle, draw, fill=blue!8, inner sep=2pt] (N10) at (1,0.8) {};
		\node[circle, draw, fill=blue!8, inner sep=2pt] (N11) at (1,0.4) {};
		\node[circle, draw, fill=blue!8, inner sep=2pt] (N12) at (1,0) {};
		\node[circle, draw, fill=blue!8, inner sep=2pt] (N13) at (1,-0.4) {};
		\node[circle, draw, fill=blue!8, inner sep=2pt] (N14) at (1,-0.8) {};
		
		% L2
		\node[circle, draw, fill=blue!8, inner sep=2pt] (N20) at (2,0.8) {};
		\node[circle, draw, fill=blue!8, inner sep=2pt] (N21) at (2,0.4) {};
		\node[circle, draw, fill=blue!8, inner sep=2pt] (N22) at (2,0) {};
		\node[circle, draw, fill=blue!8, inner sep=2pt] (N23) at (2,-0.4) {};
		\node[circle, draw, fill=blue!8, inner sep=2pt] (N24) at (2,-0.8) {};
		
		% L3
		\node[circle, draw, fill=blue!8, inner sep=2pt] (N30) at (3,0.8) {};
		\node[circle, draw, fill=blue!8, inner sep=2pt] (N31) at (3,0.4) {};
		\node[circle, draw, fill=blue!8, inner sep=2pt] (N32) at (3,0) {};
		\node[circle, draw, fill=blue!8, inner sep=2pt] (N33) at (3,-0.4) {};
		\node[circle, draw, fill=blue!8, inner sep=2pt] (N34) at (3,-0.8) {};
		
		% L4
		\node[circle, draw, fill=blue!8, inner sep=2pt] (N40) at (4,0.8) {};
		\node[circle, draw, fill=blue!8, inner sep=2pt] (N41) at (4,0.4) {};
		\node[circle, draw, fill=blue!8, inner sep=2pt] (N42) at (4,0) {};
		\node[circle, draw, fill=blue!8, inner sep=2pt] (N43) at (4,-0.4) {};
		\node[circle, draw, fill=blue!8, inner sep=2pt] (N44) at (4,-0.8) {};
		
		% gray connections 
		\foreach \i in {0,...,4} {
			\foreach \j in {0,...,4} {
				\draw[gray!30, line width=0.3pt] (N0\i) -- (N1\j);
				\draw[gray!30, line width=0.3pt] (N1\i) -- (N2\j);
				\draw[gray!30, line width=0.3pt] (N2\i) -- (N3\j);
				\draw[gray!30, line width=0.3pt] (N3\i) -- (N4\j);
			}
		}
		
		% Blue connections to N12
		\draw[blue!60, line width=1pt] (N00) -- (N12);
		\draw[blue!60, line width=1pt] (N01) -- (N12);
		\draw[blue!60, line width=1pt] (N02) -- (N12);
		\draw[blue!60, line width=1pt] (N03) -- (N12);
		\draw[blue!60, line width=1pt] (N04) -- (N12);
		
		% Blue connections to N20
		\draw[blue!60, line width=1pt] (N10) -- (N20);
		\draw[blue!60, line width=1pt] (N11) -- (N20);
		\draw[blue!60, line width=1pt] (N12) -- (N20);
		\draw[blue!60, line width=1pt] (N13) -- (N20);
		\draw[blue!60, line width=1pt] (N14) -- (N20);
		
		% Blue connections to N33
		\draw[blue!60, line width=1pt] (N20) -- (N33);
		\draw[blue!60, line width=1pt] (N21) -- (N33);
		\draw[blue!60, line width=1pt] (N22) -- (N33);
		\draw[blue!60, line width=1pt] (N23) -- (N33);
		\draw[blue!60, line width=1pt] (N24) -- (N33);
		
		% Blue connections to N41
		\draw[blue!60, line width=1pt] (N30) -- (N41);
		\draw[blue!60, line width=1pt] (N31) -- (N41);
		\draw[blue!60, line width=1pt] (N32) -- (N41);
		\draw[blue!60, line width=1pt] (N33) -- (N41);
		\draw[blue!60, line width=1pt] (N34) -- (N41);
		
		% Red path
		\draw[red, line width=1.5pt] (N01) -- (N12);
		\draw[red, line width=1.5pt] (N12) -- (N20);
		\draw[red, line width=1.5pt, dashed] (N20) -- (N33);
		\draw[red, line width=1.5pt] (N33) -- (N41);
		
		% Highlighted nodes 
		\node[circle, draw=red, line width=1.1pt, minimum size=15pt, inner sep=0.5pt, fill=white, font=\fontsize{7}{8}\selectfont] at (N01) {$\qry$};
		\node[circle, draw=red, line width=1.1pt, minimum size=15pt, inner sep=0.5pt, fill=white, font=\fontsize{7}{8}\selectfont] at (N12) {$a_1$};
		\node[circle, draw=red, line width=1.1pt, minimum size=15pt, inner sep=0.5pt, fill=white, font=\fontsize{7}{8}\selectfont] at (N20) {$a_2$};
		\node[circle, draw=red, line width=1.2pt, minimum size=15pt, inner sep=.01pt, fill=white] at (N33) {$\scriptstyle a_{L\!-\!1}$};
		\node[circle, draw=red, line width=1.1pt, minimum size=15pt, inner sep=0.5pt, fill=white, font=\fontsize{7}{8}\selectfont] at (N41) {$y$};
		
		% Legend
		\node[anchor=north west, text width=2.6cm, font=\scriptsize, align=left, draw=blue!30, line width=0.6pt, rounded corners=2pt, fill=white, inner sep=3pt] at (4.55,0.8) {
			\textbf{Legend:}\\
			\textcolor{red}{---} Random verification\\
			\phantom{---} path\\
			\textcolor{blue!60}{---} All incoming\\
			\phantom{---} connections\\
			\textcolor{gray!40}{---} Other connections % a tad darker for visibility
		};
	\end{tikzpicture}
    \caption{\fussy Path-based verification protocol. The prover commits to the full output vector for  query $\qry$. Then, the verifier samples one path, starting from a single output neuron, picks one random neuron per layer (red path) and checks each selected neuron's activation using all incoming connections (blue) from parent neurons in the previous layer. At each step, the new sampled neuron is picked among the parents.  Counterfeit outputs  lead to detectable inconsistencies with high probability.}
	% \caption{Path-based verification protocol. The prover commits to output $y_A$ for challenge query $\qry$ \matteo{XXX: The output $y_A$ should be a vector (if it's the whole output) but is visualized as a single value}. Then, the verifier starts from the last layer and picks one random neuron per layer (red path) and checks each selected neuron's activation using all incoming connections (blue) from parent neurons in the previous layer. At each step, the new sampled neuron is picked among the parents. Any counterfeit output will lead to detectable inconsistencies.}
	\label{fig:path-test}
\end{figure}
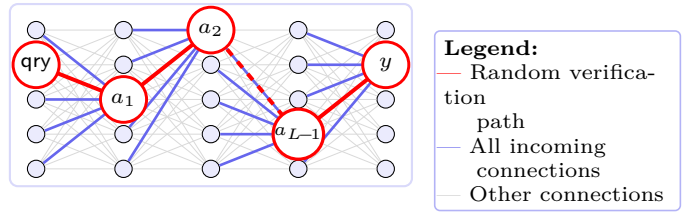

\paragraph{Towards a construction.} Our approach begins with a simple intuition: if two models $\model_1$ and $\model_2$ produce outputs that diverge substantially under the same input distribution, then their internal computations must also differ in many places. In particular, for any specific input, many activations in the trace of $\model_1$ cannot be reproduced correctly using the weights of \emph{some other} model $\model_2$ whose output behavior diverges significantly from that of $\model_1$ (and $\model_2$ need not be known to the client). This observation leads to a simple test for detecting whether the internal activations are inconsistent—and, in turn, whether the server’s claimed output is correct\footnote{
Assume here for simplicity that the two models have the same architecture; we will see later that this assumption is not needed---see footnote~\ref{foot:notsame}. \label{foot:same-arch}}.

Before describing the test, we recall that in a verifiable inference protocol, the client validates a claimed output using two inputs to the protocol: \textit{(i)} a ``ground truth’’ containing the weights and parameters of $\model_1$, and \textit{(ii)} the messages sent by the server\footnote{
This is actually true in {\em any} Verifiable Computation protocol, where Prover and Verifier must agree on the ``ground truth'' of what function is being computed. Sometimes, like in our case, the description of the function is too large for the Verifier to handle and we provide the ground truth via some cryptographic tool, as we will see soon.}. 
The second input may, of course, be adversarial (for clarity, we sometimes denote potentially malicious values with a tilde).
For a simpler treatment, we first describe our solution in an \textit{idealized} setting:
%where we will ignore certain issues\ In this model, 
we assume that at the beginning of the interaction the server gives the client oracle access to a long message, i.e., the client has random access to this message and can query only the parts it needs to remain efficient. The client has analogous random-access capability for its ground truth\footnote{This idealized presentation is standard in the cryptographic literature~\cite{kilian,marlin,znarks} and has the advantage of providing a modular treatment separating cryptographic concerns from information-theoretic ones.}. The concrete construction that keeps the Verifier efficient without this idealized model is described in \cref{sec:final-constr}.

\paragraph{A strawman and why it fails.} Before describing our approach, it is instructive to first consider a simpler test. Suppose the client samples a random node $a_j$ in the trace, reads the corresponding weights from the ground truth and the activation values of the nodes that ``feed'' into $a_j$, and checks that they are consistent:
\smallskip
\noindent
\underline{$\sf RandTestStrawman$}:
\begin{itemize}[topsep=0pt,itemsep=0pt]
    \item Query the activation value ${a}_{j}$ of a randomly chosen neuron.
	\item query the correct weights $w_{ij}$  and activation values of the nodes $a_{i}$ for $i \in G_{j}$
	\item \fussy Accept iff
	$a_{j} \!=\!\ \phi\Big(\sum_{i\in G_{j}} w_{ij} a_i \Big)$.
\end{itemize}

Our experiments, however, showed that discrepancies in the activations typically manifest only in ``late'' layers. This makes $\sf RandTestStrawman$ susceptible to false positives, as a random check in an early layer will likely pass even if the final output is different. Furthermore, while internal activations can be manipulated by an adversary, the input layer remains an immutable anchor.

The issues with $\sf RandTestStrawman$ suggest a different approach, which verifies the correctness of an entire \emph{path} from input to output. By tracing a path from an output node to the input, we ensure the output is backed by a chain of consistent computation.

\paragraph{The construction.}
%\matteo{Let us make sure it is clear that we are not actually sending the activation values}
The protocol starts with the server sending oracles to the alleged activation values of the execution $\left( \tilde{a}_1, \tilde{a}_2, \dots \right)$.
%Let $\layer_1, \layer_2, \dots, \layer_L$ be the layers in the network.
The client first samples a random node index $j$ from the output layer (the rightmost layer in \cref{fig:path-test}). Then, it reads  the following activation values from the trace sent by the server: $\tilde{a}_j$ together with all the activation values of $j$'s parents (the nodes that feed into $j$); that is, the client also reads $\left( \tilde{a}_i\right)_{i \in G_j}$.
From the ground truth string, the client will then read the weights $\left( {w}_{i j}\right)_{i \in G_j}$, those relating node  $j$ to its parents  in model $\model_1$. The client performs the following ``local'' check and rejects if this does \textit{not} hold:
\begin{equation} 
	\tilde{a}_j = \phi\Big(\sum_{i\in G_j} w_{ij} \tilde{a}_i \Big)
	\tag{$\star$}
	\label{eqn:loc-check}
	\end{equation}
The client will then sample a random node index $j'$ from $G_j$ and proceeds as above (reading $\tilde{a}_{j'}$ and the nodes in $G_{j'}$, and performing the check in $\left( \star \right)$), until it reaches the input layer. 
At that point, if all the checks along the path passed, the client accepts the claimed value $y$ as the valid output of $\model_1$ on $\qry$.

If the server is providing a trace from some different model $\model_2$, we hypothesized that a randomly sampled path will fail this test with high probability. We validated our hypothesis by running extensive experiments which confirmed our intuition that this is indeed a good test:  when comparing models with very different behaviors, we never experienced our test failing, therefore we believe that the probability of failure is negligible (see Sections \ref{sec:om-exps} and \ref{sec:llm-sep} for details). This provides us with a lightweight but robust protocol template where it is sufficient to query and validate only a fraction of the network's nodes.
%handful of nodes to check with  high confidence that the result is correct.
To cheat without being detected by this approach,
the adversary would have to construct a trace with the correct input, the desired incorrect output, and internal activations such that many paths evaluate correctly according to the weights of the real model. We conjectured this to be hard. Indeed, our experimental analysis in Section~\ref{sec:fs-exps} shows that two natural strategies (gradient descent and inverse transformations) fail to forge a trace that fools our test, even when the adversary is free to construct the trace arbitrarily (strong other-model soundness, \cref{sec:strong-om-snd}). A separate swap attack (\cref{swap_attack}), which perturbs the true model's output without involving any substitute model, tests a threat closer to full soundness and also fails.

\paragraph{Stronger security notions.} The discussion so far focused on a relatively weaker security notion (which we dub \emph{other-model soundness}) where we only consider cheating servers who run a different model $\model_2$.
%and commit to the real execution trace of $\model_2$.
This seemingly restrictive notion is of practical and theoretical interest on its own. From an application standpoint, for example, whenever inference is offered as a cloud-based service, providers may have incentives in providing outputs from cheaper models, rather than the one advertised~\cite{substitution2025} (the ground truth). We refer the reader to \cref{sec:motivation-other-model-snd} for more on this and other example scenarios. More broadly, other-model soundness captures settings where ad-hoc cheating strategies are practically infeasible or too expensive, making the use of an existing alternative model the most natural attack vector.
% Other-model soundness has another appealing feature: as we show in \cref{sec:omsound}, it lends itself very nicely to a \textit{formal} treatment (in addition to the \textit{experimental} treatment in \cref{sec:om-exps}). In particular, in that section we provide mathematical foundations  for other-model soundness; we do this by showing how general properties, related to similarity between models, can act as (approximate) \textit{sufficient}  conditions for a protocol to satisfy other-model soundness\footnote{The stronger notions of soundness we consider in this paper (and which we are about to discuss in the text) was not as easy to capture  immediately through a foundational lens. at least directly. We don't exclude, however, that other-model soundness might in the future provide \textit{the} starting point for  a theoretical treatment of stronger security notions. This may be the case, for instance, if one could prove that any successful adversary with high success probability acting arbitrarily, must essentially have produced a response that is related to the \textit{response of some other model}. This, at least in principle, could be the case for our technical approach, which requires the adversary to commit to an alleged execution trace of a model. We leave such theoretical analysis---exploring how and when the ``weaker'' soundness notion implies the stronger one---as an  interesting problem for future work.}.

Other-model soundness has another appealing feature: as we show in \cref{sec:omsound}, it lends itself very nicely to a \textit{formal} treatment (in addition to the experimental treatment in \cref{sec:om-exps}). As we show in \cref{thm:trace-sep-yields-other-model-snd}, two 
conditions are sufficient to guarantee security: (1) \emph{trace separation}, meaning that two models with sufficiently dissimilar behavior have very different activation traces when run on the same input; (2) a ``good'' test that separates the two traces by sampling only a small number of locations of the trace and the ground truth. Whether these two conditions apply can, in principle, be assessed both theoretically and experimentally. 

%Through these mathematical foundations, we show in \cref{thm:trace-sep-yields-other-model-snd} some sufficient conditions that guarantee security. Our results show that we can obtain other-model soundness under two conditions, respectively on the ``testing mechanism'' we are using and on the setting: the testing mechanism should be, at its essence,  a ``reasonable sampling mechanism for trace deviations'' (first condition) and  the deviations in outputs we are interested in ruling out are produced by models whose traces are ``sufficiently dissimilar'' from the ground truth (second condition). Whether these  two conditions  apply can, in principle, be assessed both theoretically and experimentally (especially the first, which is mostly combinatorial in nature).

% The text in this commented footnote is now in the appendix
%\footnote{This may be the case in settings where the threat model is along the following lines. Consider a model provider with a model whose training process has been audited (for example, it has been guaranteed that it has been trained only on data it had the right to use). However, the model provider may try to serve responses to users through a \textit{different} model (one, for example, trained on additional data, including some that were not permitted for copyright or other legal reasons). Our notion of other-model soundness squarely captures this scenario.}

In addition to other-model soundness, we also consider a natural and stronger security notion of broader applicability, which we refer to as \textit{full soundness}.
Under this definition a server may act arbitrarily (may be ``fully malicious'') and thus try more sophisticated strategies than running a different model.
It could, in particular, provide an arbitrary output (not computed using any model) and commit to an arbitrary trace in $C_{\trace}$.
At least intuitively, full soundness is harder to achieve than other-model soundness. In fact, it escapes our formal treatment because it is not immediately clear, analytically, whether trace separation (which deals with traces from two different models) can have any implications for an arbitrary attack (\cref{sec:motivation-other-model-snd} discusses this point as well).

We point out, however, that the analysis for other-model soundness trivially extends to full soundness for the case of ``classification'' models, i.e., models with a binary output (e.g., images of cats vs images of dogs) since the adversary can only ``flip'' the output.
Between other-model soundness and full soundness lies \emph{strong other-model soundness} (\cref{sec:strong-om-snd}), where the adversary runs a substitute model but may forge the trace arbitrarily. Our gradient descent and inverse transform attacks (Section~\ref{sec:fs-exps}) empirically test this intermediate notion, while the swap attack (\cref{swap_attack}) that perturbs the true model's output without any substitute model is a threat closer to full soundness.
Things get more complicated in the case of Large Language Models where the adversary can tinker with the output in more sophisticated ways, which make the detection probability of our ``random path'' test very low. We leave achieving full soundness for LLMs using our approach as an open problem.
Formal definitions for the above security notions in the idealized model can be found in \cref{sec:secarg} and \cref{sec:strong-om-snd}; their cryptographic counterparts are in \cref{sec:formal-treatment-verif-inf-protocol}.

\paragraph{From idealized to concrete constructions.} The above assumes for simplicity that the client has been given a direct random access to both the trace claimed by the server and  the weights of the model. In our concrete scheme this is not the case; we implement this mechanism through cryptographic commitments. At the end of a training phase, a commitment $C_{\model}$ to the weights of the correct model ($\model_1$ above) is published in a public repository. This serves as the ``ground truth'' of what the correct model is.
At the end of the protocol execution, the server publishes a commitment $C_{\trace}$ to the execution trace. At verification time the server provides the client with openings of the relevant values from the two commitments $C_{\model}$ and $C_{\trace}$. Here we will use \textit{vector} commitments~\cite{MT,VC}, which guarantee that the server is bound to unique values and the opening is ``succinct'' (i.e., much smaller than the size of the network), yielding the desired efficiency for both client and prover. A general description of this compilation approach from idealized to cryptographic constructions, together with its security proof, can be found in \cref{sec:compiler}.

\paragraph{Practical settings for approximate soundness.}

In many scenarios, even a relatively low soundness probability per query is acceptable.
One reason is that inferences are produced at massive scale: if an adversary is routinely substituting models, the adversary will on average be caught after a moderate number of queries. Similarly, our approach is well-suited to auditing scenarios, where a verifier can issue many test queries against a model to which the provider has published a commitment.

Another reason is that providers may face severe penalties when caught cheating; a rational agent will not accept the risk even if per-query detection probability is low.
This is especially natural in usage-based pricing settings (e.g., per-token billing for LLM APIs), where economic incentives to cut costs by serving a cheaper model are directly at play.
More broadly, one could deploy our approach in settings where the user pays only if the result passes a verification check.
This connects to the framework of \emph{rational proofs}~\cite{STOC:AzaMic12,CamGen15,CamGen17, CiC:CamGanGen24}, in which a prover is incentivized to behave honestly through a reward mechanism tied to verification rather than through cryptographic soundness alone.

\section{Preliminaries}
\label{sec:prelim}

We now provide some preliminaries and notation that will be useful in describing our construction.
Let $(U, Y)$ be the input set and output set for a given inference, respectively. 
Consider a trained model $\model$ that runs on input queries $\qry \in U$ and output $y=\model(\qry) \in Y$. A client (which we will also call a \emph{Verifier}) will ask a potentially untrusted server (a \emph{Prover}) to run the model $\model$ on a particular query $\qry$. Our goal is to prevent a malicious Prover (an \emph{Adversary}) $\cal A$ to claim that $\hat{y} \neq y=\model(\qry)$ is the correct output of $\model(\qry)$. We assume that an adversary has access to the correct model $\model$ (its architecture and weights). 

We will assume that models and queries are sampled from a distribution (see also \cref{appx:why-dist}). 
When executed on a query $\qry$, a model $\model$ can produce an evaluation trace $\trace$ (e.g., the list of all the activation values). We denote this mechanism through the syntax $\trace := \Eval\left( \model, \qry \right)$.
We assume that every $\trace$ is associated with a unique output, intuitively the output $\model(\qry)$ of the model and query generating the trace. 
For this purpose we posit the existence of a function $\outfn$ with the property that
$\outfn\!\left( \Eval\left( \model, \qry \right)\right)\!=\!\model(\qry)$\footnote{In natural cases $\outfn$ simply selects the appropriate neurons from the trace. See also \cref{sec:basics-models}.}. We require that both $\Eval$ and $\outfn$ be efficiently computable. 
We stress that the output function always returns \textit{some value} even if the trace if not fully well-formed or internally coherent: the output function is only trying to model the fact that every trace will implicitly be associated with some \textit{claim} on the output of that trace. The output of the trace could for example be the output of the last layer of neurons, which is well-defined by the mere fact that this is an acceptable ``trace''; clearly, this  output is well-defined even if the prior layers are not in agreement with this output (in this case, somewhere in the trace there of course must be a contradiction). See also the formal treatment in the appendix (\cref{sec:ver-inf-def} and \cref{rem:on-outputs-of-traces}).

If $F$ is an algorithm and $s$ is a string we use the standard convention $F^s$ to denote the fact that $F$ has oracle access (that is, random access) to the string $s$. Since $F$ is not constrained to reading all of $s$, it can run in time \textit{sublinear} in $|s|$.

We will occasionally talk about a ``distance function'' $\Dfn$ and apply it in the same statement to both outputs and traces. In these situations we are implicitly assuming ``a pair of distance functions'', respectively applied to outputs and traces. Talking about a single distance function allows us to use fewer subscripts without losing any essential points.

\section{A Formal Framework for Verifiable Inference}
\label{sec:framework}

In this section we present a formal framework for verifiable inference. Our approach borrows techniques from the \emph{property testing} literature \cite{goldreich2017introduction}, by leveraging the notion of similarity we discussed in \cref{sec:intuition}. We assume that two models that produce very different outputs (they are \emph{functionally dissimilar}) have ``very different'' traces (they are \emph{representionally dissimilar}). We then show that this difference can be detected by a test that requires access only to a very small number of locations in the trace.

We first introduce our security definitions in an idealized model (\cref{sec:secarg}), then provide a theoretical framework linking trace separation to other-model soundness (\cref{sec:omsound}). We also present the cryptographic counterparts of these definitions (\cref{sec:formal-treatment-verif-inf-protocol}) and a generic compiler from idealized schemes to cryptographic ones using vector commitments (\cref{sec:compiler}).

\subsection{Idealized Verifiable Inference Schemes}
\label{sec:secarg}

%\matteo{This is not what we are doing here. Also, we should introduce the testtrace algorithm first.} We now prove that our protocol satisfies an idealized definition of secure inference\footnote{We defer the cryptographic definition of secure inference to \cref{sec:formal-treatment-verif-inf-protocol}.} based on the properties of the $\testTrace$ algorithm.

We now present our framework for verifiable inference, which we will use to argue the security of our scheme.

\begin{definition}[Inference Correctness Testing]
	\label{def:idealized}
Given a model $\model$ a \emph{correctness tester} for $\model$ is an efficient algorithm $\testTrace$, with oracle access to some of its inputs, satisfying correctness and one of the two soundness definitions below (the former implies the latter).
\end{definition}

\paragraph{Correctness.}
The scheme satisfies correctness if for any security parameter $\lambda \in \NN$, model $\model$, query  $\qry$, we have:
\[ \testTrace^{\model, \trace}\left( 1^\lambda, \qry \right) = 1 \]
%where  $y := \model(\qry) \text{ and }$
where $\trace := \Eval\left( \model, \qry \right)$. Correctness requires that the testing algorithm always accept the honestly generated trace.

\paragraph{Full soundness.} Let $\sndeps: \NN \to \RR$ be a function and let $\left( \Dmodel, \Dqry \right)$ be a pair of  model and query distributions and $\Delta$ a notion of distance for probability distributions. We say that the scheme satisfies full $(\delta,\sndeps)$-soundness w.r.t to  $\left( \Dmodel, \Dqry, \Delta \right)$ if for any  adversary $\adv$ and for any security parameter $\lambda \in \NN$, we have:
\begin{align*}
	& \multilinePr{
& \Delta[\model(\qry),\outfn(\advtrace)]>\delta \; \wedge \\
	& \testTrace^{\model, \advtrace}\left( 1^\lambda, \qry \right) = 1
}{
	& \model \sample \Dmodel\\
	& \qry \sample \Dqry\\
	&  \advtrace  \gets \adv\left( 1^\lambda, \model, \qry \right)
} \leq \sndeps(\lambda)
\end{align*}

Informally, (full) soundness states that, except with $\sndeps$ probability, no adversary can fool the test into accepting when providing a trace with an output far from the correct one.

\paragraph{Other-model soundness.} Let $\sndeps: \NN \to \RR$ be a function and let $\left( \Dmodel, \Dqry,\Delta \right)$ be as above and $\advDmodel$ be a model distribution (potentially distinct from $\Dmodel$). We say that the scheme satisfies $(\delta,\sndeps)$-other-model-soundness w.r.t to  $\left( \Dmodel, \Dqry, \Delta  \right)$ and $\advDmodel$ if for any security parameter $\lambda \in \NN$, we have:
\begin{align*}
& \multilinePr{
	& \Delta[\model(\qry),\advmodel(\qry)]>\delta \;  \wedge \!\!\!\\
	& \testTrace^{\model, \advtrace}\left( 1^\lambda, \qry \right)\!\!\!\\
	&  = 1
}{
	& \!\!\model \sample \Dmodel\\
	& \!\!	\advmodel \sample \advDmodel\\
	& \!\!\qry \sample \Dqry\\
	& \!\!\advtrace := \Eval\!\left( \advmodel, \qry \right)
} \leq \sndeps(\lambda)
\end{align*}

Informally, other-model soundness states that, except with $\sndeps$ probability, the tester will be able to distinguish a false output whenever an adversary tries to use a different model (from a specific class) to evaluate $\qry$.

In the next paragraphs, we will provide a theoretical framework to argue other-model soundness from first principles (based only on properties of the models involved and on the combinatorial properties of the testing mechanism).
%\matteo{[Check/edit next para. Can we really establish full soundness? What are we referring to?]}
We will also later discuss how our experiments (and others from future work) are consistent with strong-other model soundness holding for the path-testing approach we propose in this work (as well as potentially other testing mechanisms with similar features).

%\paragraph{Query Complexity} The query complexity of $\testTrace$ is the number of locations it accesses in $\trace$ and $\cal M$.
\begin{comment}

\begin{theorem}
	\label{thm:security{def:verifiable-proper}}
Let $\Dmodel, \advDmodel$ be model distributions and $\Dqry$ be a query distribution. If $\testTrace$ is a correctness tester for $\model \gets \Dmodel$ and the Prover used a secure Vector Commitment scheme, then the Protocol in Figure~\ref{prot:our} is a secure argument system for verifiable inference (Def.~\ref{def:verifiable-proper}), which inherits the same type of soundness as $\testTrace$.
\end{theorem}
\begin{proof}[Proof (sketch)]
	The proof essentially follows from the properties of $\testTrace$ and the secure binding of the vector commitment scheme.
\end{proof}

If the query complexity (number of locations accessed) of $\testTrace$ is $k$, then the proof size of our protocol is $k o(n)$ to account for the opening proofs of the vector commitments (which are sublinear in $n$ the size of the model).
\end{comment}

\subsection{Other-Model Soundness From Testing Trace Separation}
\label{sec:omsound}

We now formalize a general version of the assumptions related to similarity  used in the AI literature. We will then formally show how these assumptions can allow us build an efficient $\testTrace$ algorithm that detects if a different model is used to compute the wrong output.

To argue that our protocol has other-model soundness, we first formalize the \textit{separability assumption} that underlies our approach:
when sampling a model $\model$ and query $\qry$ from a certain distribution, any trace that significantly disagrees with the output $\model(\qry)$ will be ``far enough'' from the actual activation trace of $\model$ with high probability.

\begin{definition}[Trace separation]
	\label{def:sep-prop}
	Let $\Dmodel, \advDmodel$ be model distributions and $\Dqry$ be a query distribution. We say that they satisfy the
	$\left( \dsepO, \dsepT, \esep  \right)$-separation property with respect to a distance function $\Dfn$:
\begin{align*}
& 	\multilinePr{
	&	\Dfn\!\left( \! \model(\qry), \advmodel(\qry) \! \right) \!\geq\!\dsepO \!\!\\
	& \,  \implies \Dfn\left(\trace, \advtrace \right)	\geq \dsepT
}{
	& \!\model \sample \Dmodel\\
	& 	\!\advmodel \sample \advDmodel\\
	& \!\qry \sample \Dqry\\
	& \!\trace := \Eval\left( \model, \qry \right)\!\\
	& \!\advtrace := \Eval\left( \advmodel, \qry \right)\!
} \geq 1 - \esep
\end{align*}
\end{definition}

It is not enough that the traces follow distinct distributions; we must also have a good way to test it.

\begin{definition}[Good separation test]
	\label{def:good-test}
	Let $\teps : \NN \to \RR$ be a function. We say that a testing mechanism $\gentest$ is $\teps$-good with respect to a distance function $\Dfn$ and a threshold $\delta$ if   for any parameter $\lambda \in \NN$, model $\model$, query $\qry$,   $\advtrace$ such that:
\begin{equation}
	\tag{$\dagger$} \gentest^{\model, \trace}\left( 1^\lambda, \qry \right) = 1
	\label{eq:adm}
	\end{equation}
	%where  $y := \model(\qry) \text{ and }$
\begin{equation}
	\tag{$\ddagger$}  \mathsf{Pr}[
\Dfn\left(\trace, \advtrace \right)	\geq \delta   \wedge  \gentest^{\model,\advtrace}\left(1^\lambda, \qry \right) = 1
	] \leq \teps(\lambda)
	\end{equation}
where $\trace := \Eval\left( \model, \qry \right)$.

\end{definition}

A good test will enable inference correctness testing as in  \cref{def:idealized}: condition $(\dagger)$ corresponds to correctness, while $(\ddagger)$ captures  other-model soundness (with testing error $\teps(\cdot)$).

\begin{theorem}
	\label{thm:trace-sep-yields-other-model-snd}
	Let $\Dmodel, \advDmodel$ be model distributions and $\Dqry$ be a query distribution.
	Then the following two conditions are sufficient (both should hold)
	for the existence of an inference correctness tester (\cref{def:idealized}) with other-model-soundness:
	\begin{itemize}[topsep=4pt]
		\item they satisfy the $\left( \dsepO, \dsepT, \esep  \right)$-separation property (\cref{def:sep-prop}) with respect to a distance function $\Dfn$;
		\item there exists a good testing mechanism with respect to $\Dfn$ and threshold $\dsepT$ (\cref{def:good-test}).
	\end{itemize}
	The resulting verifiable inference scheme has other-model soundness $\sndeps(\lambda) \leq \teps(\lambda) + \esep$.
\end{theorem}
\begin{proof}[Proof (sketch)]
	We can essentially employ the testing mechanism $\gentest$ as a verifiable inference scheme. Correctness follows immediately from condition $(\dagger)$.
	Now, in the other-model soundness experiment in \cref{def:idealized}, consider an adversary providing a model whose output is far from the honest one, that is such that $\Dfn \left(  \model(\qry), \advmodel(\qry)  \right) $. By invoking trace separation we know that the honest trace $\trace$ and $\advtrace$ will differ by at least $\dsepT$ (according to distance function $\Dfn$) with probability $(1-\esep)$. Invoking condition $(\ddagger)$, we can then conclude that the testing algorithm will reject with overall  probability at least $(1-\esep)\cdot(1-\teps(\lambda))$. This provides the claimed soundness bound.
\end{proof}

In \cref{appx:methods} we describe how to empirically estimate $\dsepO, \dsepT, \esep, \text{ and }\teps$ that can be used in \cref{thm:trace-sep-yields-other-model-snd}.

\subsection{Secure Verifiable Inference Protocol}
\label{sec:formal-treatment-verif-inf-protocol}
We now define what  a secure verifiable inference protocol in the cryptographic sense is. This notion can be seen as the ``real-world'' equivalent of our idealized notion in \cref{def:idealized} (we show how an idealized scheme can be generically compiled into cryptographic ones in \cref{sec:compiler}). Note, firstly, that our definition relaxes the notion of soundness to accept outputs that may only be ``approximately'' correct, i.e., from an output distribution not too far from the one produced by the correct model. Also, for simplicity we present a definition for 3-round protocols, where the Prover sends the first and last messages, and the second round message is a random challenge (the structure of our instantiated protocol). It is not hard to extend this to an arbitrary interactive protocol between Prover and Verifier running on a commitment to the model $\model$.

\begin{definition}
	\label{def:verifiable-proper}
	A protocol for verifiable inference on a committed model is a tuple of algorithms
	$\left( \genparams, \modelcommit, \prove, \verify \right)$ with the following syntax:
	\begin{itemize}
		\item $\genparams(1^\lambda) \to \pp$: outputs security parameters;
		\item $\modelcommit(\pp, \model) \to \cmM$: outputs a commitment, or digest, to a model (deterministically).
		\item $\prove_1(\pp, \model, \qry) \to \left( \pi_1, \mathsf{state} \right)$: outputs a first proof message $\pi_1$ and a state for the next round, on input parameters, model and query.
		\item $\prove_2(\mathsf{state}, \rho) \to \pi_2$: outputs the last proof message $\pi_2$ on input the state and a random challenge $\rho$ (notice that without loss of generality the state can contain $\pp, \model$ and $\qry$).
		\item $\verify(\pp, \cmM, \qry, y, \left(  \pi_1, \pi_2 \right), \rho )$: checks that the output of the model committed in $\cmM$ on query $\qry$ is $y$, on additional inputs the security parameters, the challenge $\rho$ and the proof messages $\pi_1,\pi_2$.
	\end{itemize}
	These algorithms should satisfy binding, correctness and one of the soundness definitions described below. In these definitions we denote by $\chalSpace$ the space from which the random challenge is sampled.
\end{definition}

\paranoindent{Binding.}
For our notion to be meaningful, we require that a digest to a model is binding. Formally, for any efficient adversary $\adv$, for every $\lambda$:
\ifShortversion
\begin{align*}
&\multilinePr{
	&\model \neq \model' \ \wedge \\
	& \cm = \cm'
}{&\!\pp \gets \genparams(1^\lambda)\!\\
& \!\left( \model, \model' \right) \gets \adv \left( \pp \right)\\
&  \cm\! \gets\! \modelcommit\!\left( \pp, \model \right)\\
&  \cm' \! \gets\!\modelcommit\!\left( \pp, \model' \right)
} \leq \negl(\lambda)
\end{align*}
\else
\begin{align*}
\multilinePr{
	\model \neq \model' \ \wedge
	 \cm = \cm'
}{ \pp \gets \genparams(1^\lambda)\\
\left( \model, \model' \right) \gets \adv \left( \pp \right)\\
  \cm\! \gets\! \modelcommit\!\left( \pp, \model \right)\\
  \cm' \gets\modelcommit\left( \pp, \model' \right)\\
} \leq \negl(\lambda)
\end{align*}
\fi

\paranoindent{Correctness.}
The scheme satisfies correctness if for any security parameter $\lambda \in \NN$, model $\model$, query  $\qry$, we have:
\begin{align*}
	&\multilinePr{
		&\verify(\pp, \cm, \qry, y,\!\\
		 & \qquad \left(  \pi_1, \pi_2 \right), \rho ) \!=\! 1\!
	}{&\!\pp \gets \genparams(1^\lambda)\!\\
		&  \!\cm\! \gets\! \modelcommit\!\left( \pp, \model \right)\\
		&  \!\!\left(\! \pi_1, \mathsf{state}\! \right) \gets \prove_1(\pp, \model, \qry)\!\\
		& \!\rho \sample \chalSpace\\
		& \!\pi_2 \gets \prove_2(\mathsf{state}, \rho)\\
		& \!y := \model\left( \qry \right)
	} = 1
\end{align*}

\paranoindent{Full soundness.} Let $\sndeps: \NN \to \RR$ be a function and let $\left( \Dmodel, \Dqry \right)$ be a pair of model and query distributions. Let $\Delta$ be a distance measure for distributions. We say that the scheme satisfies full $(\delta,\sndeps)$-soundness w.r.t to  $\left( \Dmodel, \Dqry, \Delta \right)$ if for any efficient adversary $\adv = \left( \adv_1, \adv_2 \right)$ and for any security parameter $\lambda \in \NN$, we have:
\begin{align*}
	& \multilinePr{
		%	& \model(\qry) \neq \outfn(\advtrace)\  \wedge \\
		& \Delta[\model(\qry),\tilde{y}]>\delta \  \wedge \\
	&\ \verify(\pp, \cm, \qry, \tilde{y},\!\\
	& \ \qquad \left(  \pi_1, \pi_2 \right), \rho ) \!=\! 1\!
	}{
		& \!\model \sample \Dmodel\\
		&\!\pp \gets \genparams(1^\lambda)\!\\
		&  \!\cm\! \gets\! \modelcommit\!\left( \pp, \model \right)\!\\
		& \!\qry \sample \Dqry\\
		&  \!\!\left( \tilde{y}, \pi_1, \mathsf{state} \right)  \gets \adv\left( \pp, \model, \qry \right)\\
		& \!\rho \sample \chalSpace\\
		& \!\pi_2 \gets \adv_2(\mathsf{state}, \rho)
	} \leq \sndeps(\lambda)
\end{align*}

\paranoindent{Other-model soundness.} Let $\sndeps: \NN \to \RR$ be a function and let $\left( \Dmodel, \Dqry, \Delta \right)$ be as above and $\advDmodel$ be a model distribution (potentially distinct from $\Dmodel$). We say that the scheme satisfies $\sndeps$-other-model-soundness w.r.t to  $\left( \Dmodel, \Dqry, \Delta\right)$ and $\advDmodel$ if for any security parameter $\lambda \in \NN$, we have:
\begin{align*}
	& \multilinePr{
		%	& \model(\qry) \neq \outfn(\advtrace)\  \wedge \\
		& \Delta[\model(\qry),\tilde{y}]>\delta \  \wedge\! \\
		&\ \verify(\pp, \cm, \qry, \tilde{y},\!\!\\
		& \ \qquad \left(  \pi_1, \pi_2 \right), \rho ) \!=\! 1\!
	}{
		& \!\model \sample \Dmodel\\
		&\!\pp \gets \genparams(1^\lambda)\!\\
		&  \!\cm\! \gets\! \modelcommit\!\left( \pp, \model \right)\!\\
		& 	\!\advmodel \sample \advDmodel\\
		& \!\qry \sample \Dqry\\
		& \!\tilde{y} := \advmodel\left( \qry \right)\\
		&  \!\!\left(\! \pi_1, \mathsf{state}\! \right)\! \gets\!\prove_1(\pp, \advmodel, \qry)\!\\
	& \!\rho \sample \chalSpace\\
	& \!\pi_2 \gets \prove_2(\mathsf{state}, \rho)
	} \leq \sndeps(\lambda)
\end{align*}

\begin{remark}
\label{rem:asymmetry-defs}
In contrast to \cref{def:idealized}, \cref{def:verifiable-proper} above has an explicit variable $y$ referring to the claimed output by the prover. This a mere cosmetic departure from \cref{def:idealized}, however---in both notions the prover/adversary is claiming a well defined output. In the case of our idealized notion, this output corresponds to the value of $\outfn(\advtrace)$. Adopting $\outfn$ kept the idealized setting slightly simpler and without redundant variables (the  trace already contains output values; see also \cref{rem:on-outputs-of-traces}). In \cref{def:verifiable-proper}, we need an explicit $y$ because we don't have fewer assumptions on the output of the adversary (we don't assume it is a trace anymore in the sense of \cref{def:model}).
\end{remark}

\subsection{Compiling Idealized Schemes into Cryptographic Ones}
\label{sec:compiler}

Now we move away from the idealized model and show how to implement our protocol in practice, by constructing a secure way to give the Verifier random access to the parameters of the correct model, and to the trace of the computation of the adversary.
In particular, we present a compiler that transforms an idealized scheme (in the sense of \cref{def:idealized})  into its cryptographic counterpart (\cref{def:verifiable-proper}) through general vector commitments~\cite{PKC:CatFio13}.
 This is  approach is standard in the cryptographic literature.

A \emph{vector commitment (VC)} is a cryptographic primitive that allows a prover to commit to a tuple
\( v = (v_1, v_2, \dots, v_n) \) in such a way that the resulting commitment is
\emph{binding} (the prover cannot later open it to a different sequence) and
\emph{positionally binding} (each committed value is bound to its index in the vector).
Later, the prover can efficiently produce a short proof---called an \emph{opening}---that reveals the value at a specific
position \( i \) and convinces a verifier that this value is indeed the one that was committed at that position,
without revealing the other entries in the vector.\footnote{We provide a more formal definition in \cref{apdx:vec}.}
A popular and foundational construction for vector commitments is the Merkle tree \cite{merkle1989certified}, which we adopt in our scheme (in an optimized variant).

\fussy
Let 	 $\pivc$ be a vector commitment scheme and $\piIdeal$ be an idealized verifiable inference scheme. Let $( \nM, \nQ, \nTrc,$ $\Eval, \idxsOut  )$ be a family of models (as defined in \cref{def:model}). The final scheme would work as follows:
\begin{itemize}
	\item \textit{Parameter generation:} runs the parameter generation algorithm of $\pivc$ (if the algorithm requires a length parameter, use the maximum expected length of any trace).
	\item \textit{Model commitments: } Run the commitment algorithm of $\pivc$ on a string representation of $\model$.
	\item \textit{Prover: }
	\begin{enumerate}
		\item Generate $\trace := \Eval \left( \model, \qry \right)$ and use $\pivc$ to commit to it. Let $\cm_\trace$ be the resulting commitment.
		\item Send $\pi_1 := \cm_\trace$; receive random challenge $\rho$. %Generate random coins $\rho \gets \hash(\pp || \cmM || \qry || y || \cm_\trace)$ using a cryptographic hash function $\hash$.
		\item Use string $\rho$ as randomness for running the idealized  testing algorithm $$\piIdeal.\testTrace^{\model, \trace}\left( 1^\lambda, \qry \right);$$
		\item Let $\Spos_\model$ and $\Spos_\trace$ be the sets of indices that the execution $\testTrace$ queries for the string $\model$ and $\trace$ respectively. Let  $\SposP_\trace := \Spos_\trace \cup \idxsOut$. Use the opening algorithm of $\pivc$ to open each of the values
		in $\left(  \model[j] \right)_{j \in \Spos_\model}$ (resp. in $\left(  \trace[j] \right)_{j \in \SposP_\trace}$) against commitment $\cmM$ (resp. $\cm_\trace$). Let $\mathsf{openings}$ denote the concatenation of all these opening proofs as well as the related values from $\model$ and $\trace$.
		\item Return $\pi_2 := \mathsf{openings}$
	\end{enumerate}
	\item \textit{Verifier}:
	\begin{enumerate}
		%			\item Generate random coins $\rho \gets \hash(\pp || \cmM || \qry || y || \cm_\trace)$ through the hash function $\hash$.
		\item Use string $\rho$ as randomness for running the idealized  testing algorithm $$\piIdeal.\testTrace^{\model, \trace}\left( 1^\lambda, \qry \right)$$ in  the following way: whenever there is an access to $\model$ or $\trace$, use the corresponding value in $\mathsf{openings}$ for it; if that value has not been provided by the prover reject.  If $\testTrace$ does not reject, go to the next step; otherwise, reject.
		\item Check all the opening proofs in $\mathsf{openings}$ and accept if and only if they are all valid and $y$ is consistent with the trace output found in $\mathsf{openings}$.
	\end{enumerate}
\end{itemize}

\bigskip
The following theorem  links the security of idealized verifiable inference schemes (\cref{def:idealized}) and of vector commitments (\cref{apdx:vec}) to that of cryptographic verifiable inference schemes (\cref{def:verifiable-proper}).

\begin{theorem}
	\label{thm:comp}
	Let $\Dmodel, \advDmodel$ be model distributions and $\Dqry$ be a query distribution.
	Let $\pivc$ be a vector commitment scheme and $\piIdeal$ be an idealized verifiable inference scheme, then:
	\begin{itemize}[topsep=4pt]
		\item if $\piIdeal$ is fully $\sndeps$-sound with respect to$( \Dmodel,$ $\Dqry, \Delta )$ then the construction above  is a verifiable inference scheme that is fully $\left(\sndeps(\cdot)+\negl(\cdot) \right)$-sound with respect to  $( \Dmodel,$ $\Dqry, \Delta )$.
		\item if $\piIdeal$ satisfies $\sndeps$-other-model-soundness with respect to $\left( \Dmodel, \Dqry, \Delta\right)$ and $\advDmodel$, then the construction above is a  verifiable inference scheme that is $\left(\sndeps(\cdot)+\negl(\cdot) \right)$-other-model-sound with respect to  $\left( \Dmodel, \Dqry, \Delta\right)$ and $\advDmodel$.
		\item if $\piIdeal$ satisfies $\sndeps$-strong-other-model-soundness with respect to $\left( \Dmodel, \Dqry, \Delta\right)$ and $\advDmodel$, then the construction above is a verifiable inference scheme that is $\left(\sndeps(\cdot)+\negl(\cdot) \right)$-strong-other-model-sound with respect to $\left( \Dmodel, \Dqry, \Delta\right)$ and $\advDmodel$.
	\end{itemize}
\end{theorem}
\begin{proof}[Proof (sketch)]
	\fussy
	This compilation approach is standard and we do not propose a full proof here. See~\cite{kilian} or~\cite{cryptoeprint:2016/116}.
	The completeness and soundess of the protocol above follow immediately. Binding follows directly from that of the vector commitment.
\end{proof}

\section{Our Main Construction}
\label{sec:constr}
\label{sec:sec-discussion}
\label{sec:final-constr}
\label{sec:full-snd}

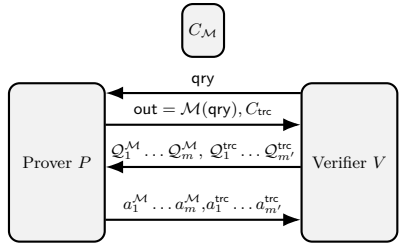
\begin{figure}[t]
\centering
% Top row: (a) on left, (b) on right - both starting at same top level
\hspace{0.65cm}
%\vspace{0pt} % Remove any top spacing
\begin{tikzpicture}[font=\small, >=Latex, node distance=3.5cm, scale=0.7, transform shape, baseline=(current bounding box.north)]
% Prover on left, Verifier on right - positioned to align with a^(0) of figure c
\node[draw, rounded corners, thick, fill=gray!10, minimum width=1.8cm, minimum height=3cm] (P) at (0,0) {Prover $P$ };

\node[draw, rounded corners, thick, fill=gray!10, minimum width=0.5cm, minimum height=1cm] (S) at (2.75,2.5) {$C_{\model}$};

\node[draw, rounded corners, thick, fill=gray!10, minimum width=1.8cm, minimum height=3cm] (V) at (5.5,0) {Verifier $V$};

% Arrows from top to bottom (1-4) with proper vertical separation
% \draw[->, thick] ($(P.east)+(0,0.2)$) -- node[below, font=\small]{1. } ($(V.west)+(0,0.2)$);

% \advmodel(\qry)
\draw[->, thick] ($(V.north west)+(0,-0.2)$) -- node[above, font=\small]{$\qry$} ($(P.north east)+(0,-0.2)$);

\draw[->, thick] ($(P.north east)+(0,-0.8)$) -- node[above, font=\small]{$\outfn = \model ( \qry ), C_{\trace}$ } ($(V.north west)+(0,-0.8)$);

%\draw[->, thick] ($(P.north east)+(0,-0.8)$) -- node[below, font=\small]{$C_{\trace}$ } ($(V.north west)+(0,-0.8)$);

\draw[->, thick] ($(V.north west)+(0,-1.6)$) -- node[above, font=\small]{$\mathcal{Q}^\model_{1} \dots \mathcal{Q}^\model_{m}$, $\mathcal{Q}^\trace_{1} \dots \mathcal{Q}^\trace_{m^\prime}$} ($(P.north east)+(0,-1.6)$);

%\draw[->, thick] ($(V.north west)+(0,-1.6)$) -- node[below, font=\small]{$\mathcal{Q}^\trace_{1} \dots \mathcal{Q}^\trace_{m^\prime}$} ($(P.north east)+(0,-1.6)$);

\draw[->, thick] ($(P.north east)+(0,-2.6)$) -- node[above, font=\small]{$a_1^\model \dots a_m^\model$,$a_1^\trace \dots a_{m^\prime}^\trace$} ($(V.north west)+(0,-2.6)$);

%\draw[->, thick] ($(P.north east)+(0,-2.6)$) -- node[below, font=\small]{$a_1^\trace \dots a_{m^\prime}^\trace$} ($(V.north west)+(0,-2.6)$);

% \draw[->, thick] ($(P.east)+(0,-0.15)$) -- node[below, font=\small]{3. Reveal path $S$} ($(V.west)+(0,-0.15)$);

% \draw[->, thick] ($(V.south west)+(0,0.1)$) -- node[below, font=\small]{4. Send activations $\Pi$} ($(P.south east)+(0,0.1)$);
\end{tikzpicture}
%\phantomcaption
\caption{Interaction flow in our protocol. The Prover computes $\trace=\Eval(\model,\qry)$ and commits to it via a vector commitment $C_{\trace}={\sf VC}(\trace)$. Verifier runs {$\testTrace$} where the queries of the test algorithm to $\model,\trace$ are answered by the Prover through the vector commitments $C_{\model},C_{\trace}$. See also the alternative protocol description in \cref{sec:constr}.}
\label{prot:our}
\end{figure}

In this section we present  our  construction and discuss its security.
%, instantiating the formal framework of \cref{sec:framework}.
%We propose $\sf RandPathTest$ as our testing mechanism (\cref{sec:sec-discussion}), present our concrete cryptographic protocol compiled from the idealized construction using vector commitments (\cref{sec:final-constr}), and discuss soundness bounds (\cref{sec:full-snd}).

%In addition to our main protocol, we create a version of our protocol in the refereed model of computation~\cite{canetti2013refereed}, which can be found in \cref{sec:ref}.

%\subsection{Security}

%We propose $\sf RandPathTest$ as the testing mechanism to instantiate our framework.
\subsection{The Construction}

Before describing our construction, we briefly recall the relevant structure of neural networks (introduced in \cref{sec:intuition}) and some other items useful to understand how the construction works.
\begin{itemize}
	\item A neural network is organized into $d$ layers, numbered $1$ (input) through $d$ (output). Each neuron $j$ in layer $\ell$ receives inputs from a set $G_j$ of parent neurons in layer $\ell-1$, and its activation is computed as $a_j = \phi\!\big(\sum_{i \in G_j} w_{ij} a_i\big)$, where $w_{ij}$ are the edge weights and $\phi$ is a non-linear activation function. 
\item We consider the  \emph{execution trace} $\trace$ of an inference as the  vector of all activation values $(a_1, a_2, \ldots)$, one per neuron, organized so that each entry is associated with a specific layer and has a well-defined set of parents.

\item The \emph{architecture} of the network---i.e., the graph structure specifying which neurons connect to which, the number of layers, and the layer each neuron belongs to---is public and known to both the Prover and the Verifier. The committed values are the \emph{weights} $(w_{ij})_{ij}$ and the \emph{activations} $(a_1, a_2, \ldots)$; the wiring itself is not secret. This means that when either party refers to a parent set $G_j$ or samples a random neuron in a given layer, both sides agree on what that means.
\end{itemize}

In Figure~\ref{fig:randpathtest} we describe $\sf RandPathTest$, the idealized version of our protocol. As sketched in \cref{sec:intuition}, the test samples a random path from an output node back to the input layer and checks, at each node along the path, that the claimed activation is locally consistent with the correct model's weights and the claimed activations of the node's parents.

In \cref{fig:final-protocol} we describe our final cryptographic protocol
(the reader can also find a high-level perspective of the protocol flow in Figure~\ref{prot:our}).
As shown in \cref{fig:final-protocol}, access to the correct weights is provided via a VC $C_\model$ which we assume has been computed correctly. However, note, the commitment to the trace $C_\trace$ is generated by the adversary, so we have no guarantee that it is indeed the actual trace of the computation of the expected model.
This construction can be seen as an application of a specific instance of the generic compiler using vector commitments described in \cref{sec:compiler}\footnote{\label{foot:notsame} This test also enforces that whatever the Prover committed to in $C_{\trace}$  has to be a trace on a model with the same architecture as $\model$, therefore we do not need to assume that is the case---see also Footnote~\ref{foot:same-arch}.}.

\begin{figure}[ht]
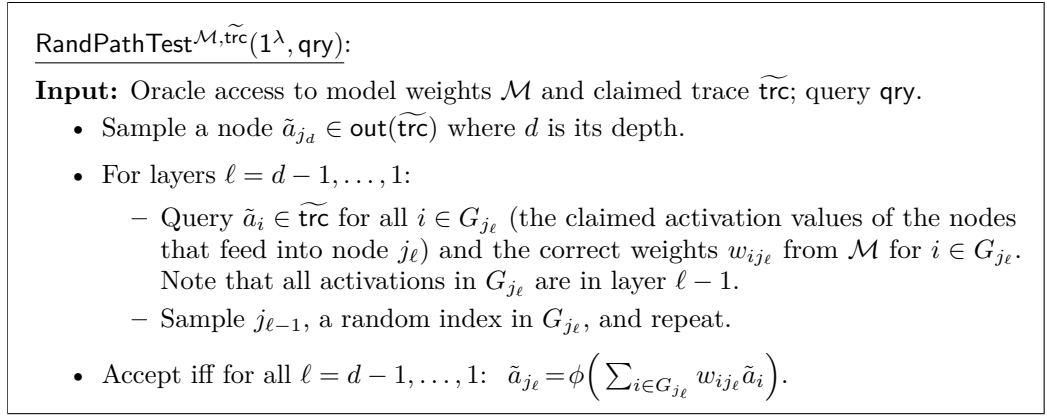

\begin{mdframed}
\smallskip
\noindent
\underline{$\sf RandPathTest^{\model, \advtrace}(1^\lambda, \qry)$}:
\smallskip

\noindent\textbf{Input:} Oracle access to model weights $\model$ and claimed trace $\advtrace$; query $\qry$.
\begin{itemize}[topsep=2pt,itemsep=2pt]
	\item Sample a node $\tilde{a}_{j_d} \in \outfn(\advtrace)$ where $d$ is its depth.
	\item For layers $\ell=d-1,\ldots,1$:
	\begin{itemize}[topsep=0pt,itemsep=0pt]
		\item Query $\tilde{a}_i \in \advtrace$ for all $i \in G_{j_{\ell}}$ (the claimed activation values of the nodes that feed into node $j_{\ell}$) and the correct weights $w_{ij_{\ell}}$ from $\model$ for $i \in G_{j_{\ell}}$. Note that all activations in $G_{j_\ell}$ are in layer $\ell-1$.
		\item Sample $j_{\ell-1}$, a random index in $G_{j_\ell}$, and repeat.
	\end{itemize}
	\item \fussy Accept iff for all $\ell=d-1,\ldots,1$: \
	$\tilde{a}_{j_{\ell}}\!=\!\phi\Big(\sum_{i\in G_{j_\ell}} w_{ij_{\ell}} \tilde{a}_i \Big)$.
\end{itemize}
\end{mdframed}
\caption{The $\sf RandPathTest$ algorithm (our protocol in the idealized model). The test samples a random path from an output node back to the input layer and checks local consistency at each node along the path using the correct model's weights.}
\label{fig:randpathtest}
\end{figure}

%\subsection{From the Idealized to the Concrete Cryptographic Scheme Using Vector Commitments}

%In fact, the {\sf RandTestStrawman} procedure described earlier in the idealized model will fail in this case, since the adversary can simply commit to the correct trace, except for the incorrect output values (which means that all the internal nodes will trivially satisfy the test).

%A main contribution of our paper is to propose {\sf RandPathTest} as a ``good test'' for trace separation, which also works in the stronger adversarial model where the claimed computation trace is built arbitrarily.

\begin{figure}[ht]
\begin{mdframed}
\smallskip
\noindent
\underline{The Final Protocol}:
\smallskip
\begin{itemize}[topsep=2pt,itemsep=4pt]
		\item \textbf{Setup (commitment to the model).} A digest of the model $\model$ is provided to the client via a vector commitment $C_{\model}$ to the weights $(w_{ij})_{ij}$ of the neural network (where $w_{ij}$ is the weight on the edge from neuron $i$ to neuron $j$). We assume $C_{\model}$ to be correctly computed at the end of the training phase. The Verifier holds $C_{\model}$; the Prover holds the decommitting information for $C_{\model}$.
		\item \textbf{Proving stage.} On input query $\qry$:
		\begin{enumerate}[topsep=0pt,itemsep=2pt]
			\item The Prover computes $\trace := \Eval(\model, \qry)$, the execution trace consisting of all activation values $(a_1, a_2, \ldots)$ produced during inference. It sends the Verifier the claimed output $\outfn(\trace)$ together with a vector commitment $C_{\trace}$ to the entire trace.
			\item The Verifier sends a random challenge $\rho$.
			\item Using $\rho$ as randomness, the Prover determines a random path from an output node back to the input layer as done in \cref{fig:randpathtest}.
 This identifies a set of positions in $\model$ (weights along the path) and in $\trace$ (activations along the path and their parents). During this stage, the Prover opens the corresponding entries from $C_\model$ and $C_\trace$ and sends the values and opening proofs to the Verifier.
		\end{enumerate}
		\item \textbf{Verification stage.} The Verifier:
		\begin{enumerate}[topsep=0pt,itemsep=2pt]
			\item Checks that all vector commitment openings are valid with respect to $C_\model$ and $C_\trace$; rejects if any opening fails.
			\item Using randomness $\rho$, reconstructs the same path as the Prover. For each node $j_\ell$ on the path, verifies the \emph{local consistency check}: that the claimed activation $\tilde{a}_{j_\ell}$ equals $\phi\!\Big(\sum_{i \in G_{j_\ell}} w_{ij_\ell} \tilde{a}_i\Big)$, where the weights $w_{ij_\ell}$ come from $C_\model$ and the parent activations $\tilde{a}_i$ come from $C_\trace$.
			\item Accepts if and only if all local consistency checks pass.
		\end{enumerate}
\end{itemize}
\end{mdframed}
\caption{Our concrete cryptographic protocol for verifiable inference. The protocol instantiates the generic compiler of \cref{sec:compiler} with $\sf RandPathTest$ (Figure~\ref{fig:randpathtest}) as the testing mechanism. We instantiate  vector commitments with a  Merkle-tree-based approach.}
\label{fig:final-protocol}
\end{figure}

\subsection{Security}
We now discuss security properties of our construction. Since by \cref{thm:comp}, the compiled cryptographic construction inherits the guarantees of the idealized protocol, in the following discussion we simply focus on the idealized $\sf RandPathTest$.

\paragraph{Other-model soundness.}
Recall that under other-model soundness, the adversary is restricted to running a different model $\widetilde{\model}$ and committing to its honest execution trace. Intuitively, if $\model$ and $\widetilde{\model}$ produce sufficiently different outputs, their activation traces must also diverge at many internal nodes, making the discrepancy detectable by sampling a random path.

To establish this formally, one may apply our theoretical framework from \cref{sec:framework}. By \cref{thm:trace-sep-yields-other-model-snd}, two conditions suffice: (1)~trace separation (\cref{def:sep-prop}), meaning that models with dissimilar outputs produce traces that are far apart; and (2)~the existence of a good separation test (\cref{def:good-test}), meaning that this distance is detectable from a small number of sampled positions. If both conditions hold, the resulting scheme achieves other-model soundness with error $\varepsilon(\lambda) \leq \varepsilon_{\mathrm{tst}}(\lambda) + \varepsilon_{\mathrm{sep}}$. 

We validate both conditions empirically. Our model separation experiments (\cref{sec:om-exps}) confirm that common architectures---ResNet-18 classifiers and Llama-2-7B---exhibit trace separation: models with different training or fine-tuning produce measurably different activation patterns, with separation values bounded well away from zero across all tested paths. The same experiments confirm that $\sf RandPathTest$ successfully distinguishes these traces, acting as a good separation test. We therefore conclude that our protocol satisfies other-model soundness for the model families and input distributions tested.

\paragraph{Strong other-model soundness.}
Our adversarial experiments in \cref{sec:fs-exps} provide evidence that $\sf RandPathTest$ resists natural attack strategies that go beyond the other-model soundness threat model. Under strong other-model soundness (defined formally in \cref{sec:strong-om-snd}), the adversary still uses a substitute model $\widetilde{\model}$ to produce the output but is free to construct the committed trace arbitrarily, rather than submitting the honest trace of $\widetilde{\model}$. This captures a natural threat scenario: an adversary who serves a cheaper or dealigned model and then attempts to forge a trace that passes verification against the committed weights of $\model$.

We evaluated three distinct attack strategies---gradient-descent reconstruction, inverse transforms, and logit swapping---each designed to produce a trace that is consistent with the correct model's weights yet yields an incorrect output. Across millions of evaluated paths, none succeeded: all produced traces with separation values bounded well away from zero at every layer. The fundamental obstacle is that fooling $\sf RandPathTest$ requires constructing a trace that is simultaneously consistent along many paths from output to input, a constraint that tightens combinatorially with network depth and width. We stress that these results demonstrate resistance to specific attack strategies  in the context of specific tasks, not a guarantee against all possible adversaries in any setting; stronger attacks may exist.

The swap attack (\cref{swap_attack}) tests a threat that is closer to full soundness: the adversary takes the true model's output, perturbs it by swapping the highest- and lowest-probability tokens, and then optimizes intermediate activations to match. Even in this setting, 100\% of attack attempts produced measurable activation deviations (minimum separation value of $0.031$), indicating that $\sf RandPathTest$ detects tampering even when no explicit substitute model is involved.

\paragraph{Full soundness.}
Achieving full soundness---where the adversary may produce an arbitrary output and construct the trace in any way it chooses---remains in part an open problem. A fundamental limitation of any single-path sampling strategy is that detection probability cannot exceed $1/N$, where $N$ is the maximum layer width of the network. To see why, observe that an adversary can modify a single node in a layer of size $N$, and a single random path will select that node with probability exactly $1/N$. While this bound does not appear to be a practical limitation under (strong) other-model soundness---where our experiments show detection across all tested paths---it may become relevant when the adversary is free to choose arbitrary outputs and concentrate modifications in as few nodes as possible.

A natural mitigation is to sample multiple independent paths, which increases the detection probability at the cost of larger proofs and more communication.  One can also consider adaptive sampling strategies that prioritize layers or nodes where activations are statistically more sensitive to tampering (we have some partial results from this approach and we are currently developing them further). Understanding the precise tradeoffs between the number of sampled paths, proof size, and the resulting soundness guarantees is an important direction for future work.

\section{Model Separation Experiments}
\label{sec:om-exps}

To show that our protocol satisfies other-model soundness (\cref{sec:omsound}), we must show that the underlying assumptions of separability hold. In this section, we detail our efforts to empirically evaluate this assumption through a series of experiments on common machine learning models.

\subsection{Experiment Design}
\label{sec:om-exps:design}

To validate our assumptions, our experiments must answer two questions:

\begin{enumerate}[topsep=0pt,itemsep=0pt]
    \item Can models with identical architectures but different training data, or models trained on the same data where one is fine-tuned, be distinguished based on their evaluation trace (i.e., the list of all activation values)?
    \item Is our path testing protocol $\mathsf{RandPathTest}$ (\cref{sec:final-constr}) able to distinguish between two models with different evaluation traces?
\end{enumerate}

The first question relates to the notion of trace separation, whereas the second question helps us understand if our protocol is a good separation test. 

In our experiments, to calculate the differences between model traces, we computed the activation differences at each layer. For a verified model $\model$ and an unverified model $\advmodel$, and for all  $1 \leq i \leq L$, this difference, which we call the \textbf{separation value}, is computed as:%\footnote{Implementation includes bias terms, omitted here for clarity.}
\begin{equation}\label{eqn:mformula}
    \lvert\phi((a_\model^{i-1})^{\intercal} W_{\model}^{i})) - \phi((a_{\advmodel}^{i-1})^{\intercal} W_{\model}^{i})\lvert 
\end{equation}

To answer the first question, we select a model $\model$ and a similar model $\advmodel$, and compute \cref{eqn:mformula} for all layers in the network. If we can see a marked difference between the two models then we can claim trace separation. 

%\kayman{ we used same-size datasets for training, we can remove this part.Additionally, we consider the case where $M_A$ is trained on a larger subset set of data from the domain set in order to test whether modified or augmented datasets could affect the internal representations of the final models to see if it shows up as a detectable and measurable divergence in their internal activation patterns. The hypothesis being that since additional data yields higher prediction accuracy this should show up in the weights and hence in the trace separation in the two models. }
To answer the second question, instead of looking at the entire model, we focus  on a single, specific path through the network. So, we replace the activations of $\model$ along a random path $P$ with those from $\advmodel$, and then recompute \cref{eqn:mformula}. 

%\tushar{What is this referring to?} \kayman{someone did this but this experiment is under appendix actually} \tushar{Please integrate with swap attack text in \cref{sec:fs-exps:ad:grad}} We take a random path $P$ from the original model and swap its internal values (activations) with the corresponding ones from the new model. We then recalculate the model's output. This method helps us amplify subtle differences between the two models. Even if the models are identical in every other part, a small change in just one path would be enough to show that they are different. 

In the remainder of this section, we discuss our model separation experiments on two specific architectures: a pre-trained ResNet-18 classifier model and the Llama-2-7B LLM.

\subsection{Model Separation in Classifiers}
\label{sec:resnet-sep}
\subsubsection{Setup}

\begin{figure}[t]
    \centering
    \includegraphics[width=0.8\linewidth]{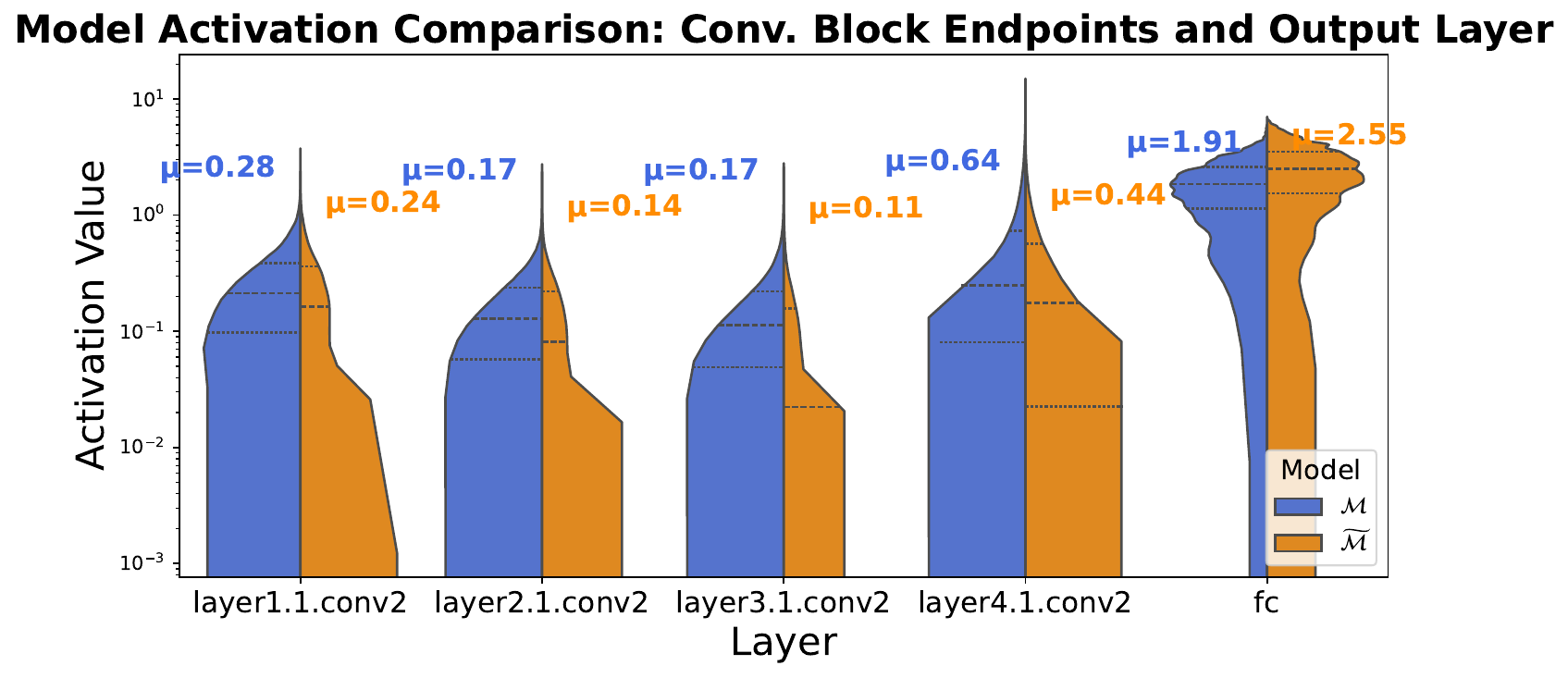}
    \caption{Activation distributions for path-selected neurons at bottleneck block endpoints and FC layer across models. Split violin plots show the probability density for neurons traversed during path selection. $\model$ (blue) and $\advmodel$ (orange) distributions are separated with quartile markers and annotated mean values.}
    \label{fig:violin}
\end{figure}

\begin{figure}[t]
	\centering
	\includegraphics[width=0.70\linewidth]{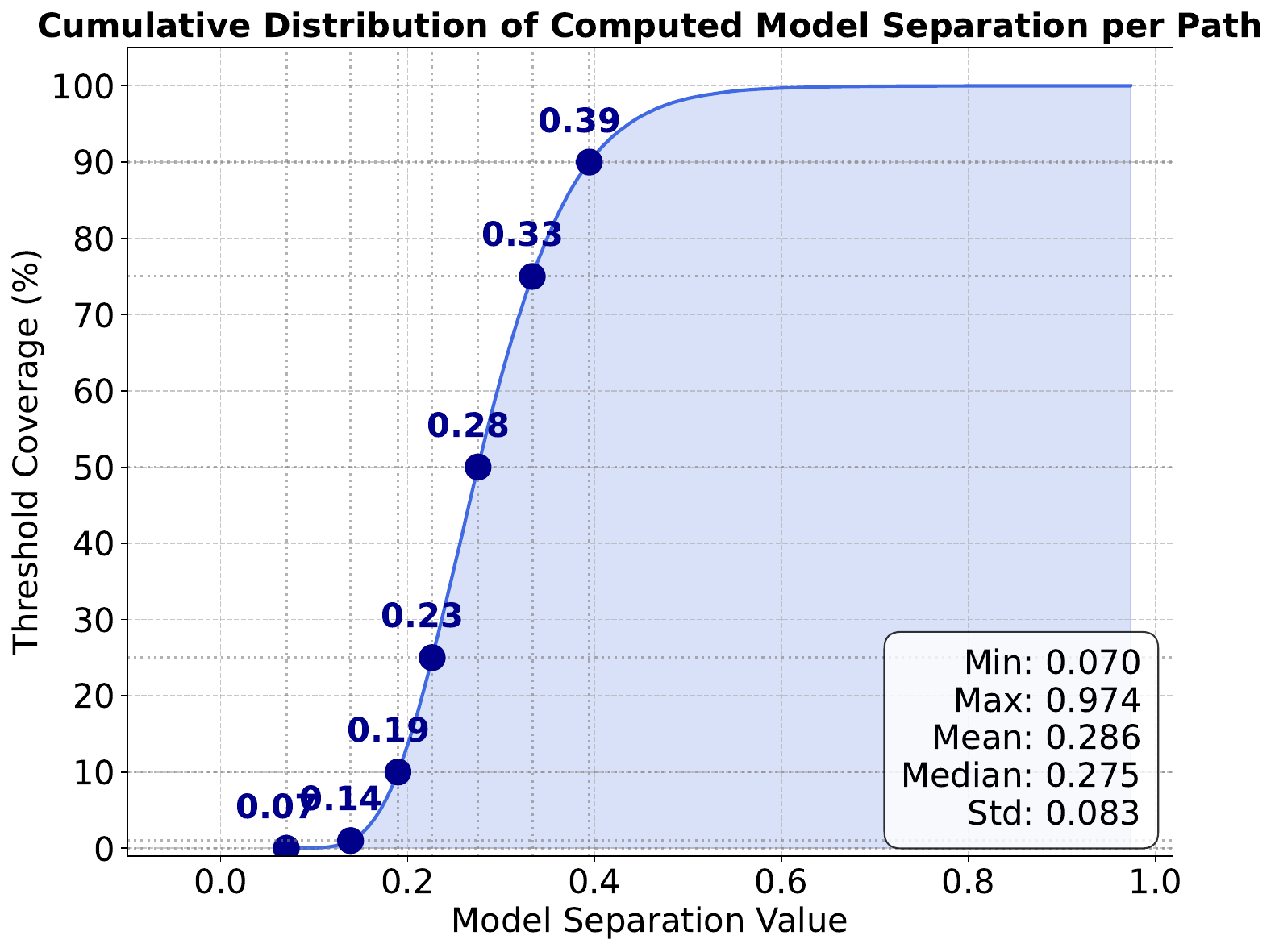}
	\caption{Cumulative distribution of model separation values (\cref{eqn:mformula}) for path-selected neurons, aggregated by mean per image-round path. The x-axis shows separation magnitudes; the y-axis shows the percentage of paths below each threshold. Blue dots indicate percentile thresholds.}
	\label{fig:meanneuron}
\end{figure}

\begin{figure}[t]
	\centering
	\includegraphics[width=0.65\linewidth]{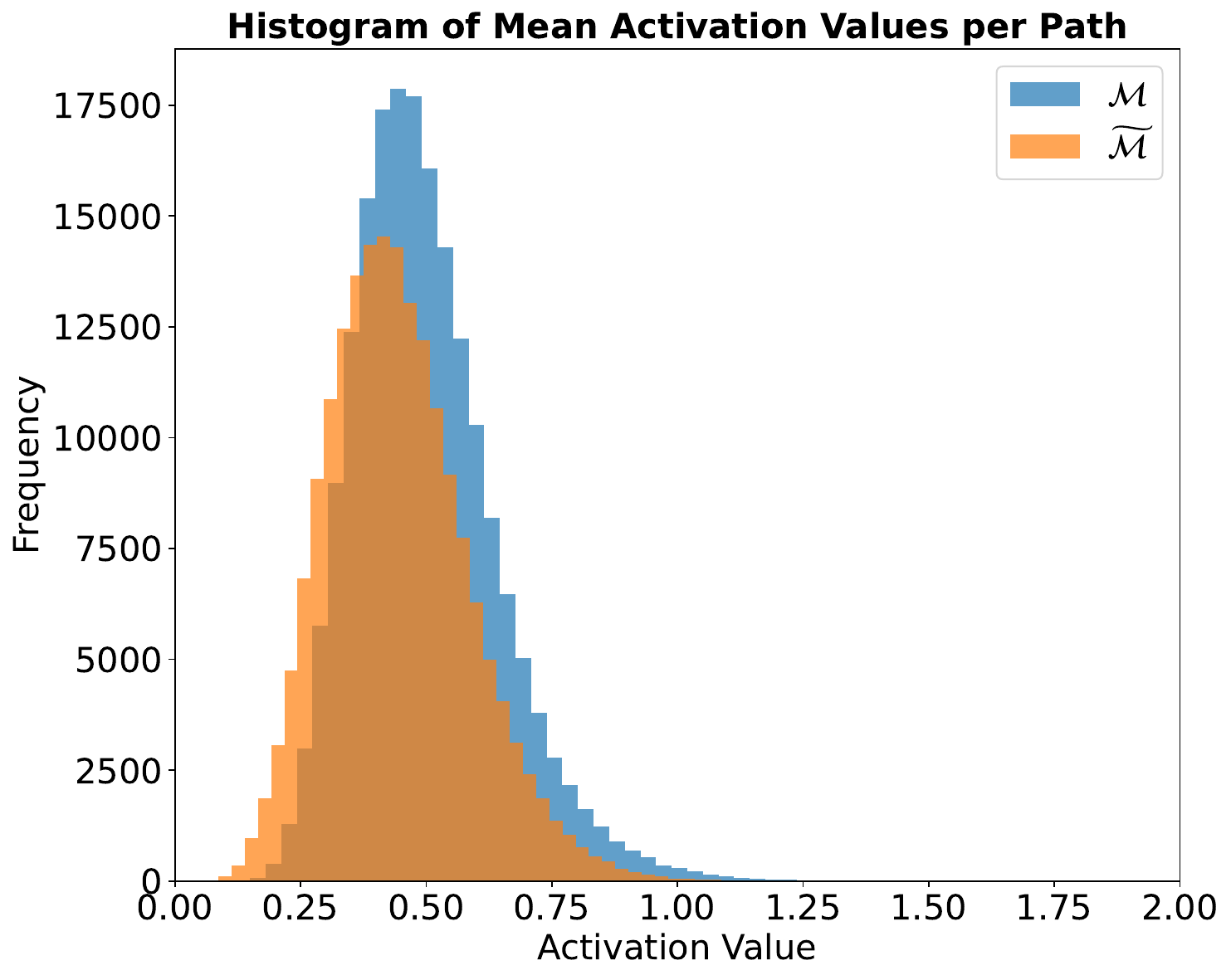}
	\caption{Mean activation values per path. Each point is the average activation across layers for one image-round pair. }
	\label{fig:histogram}
\end{figure}

We selected ResNet-18, a widely-used architecture for classification problems \cite{he2016deep} . When designing the experiments we wanted to use a lightweight yet popular model to ensure that our results are interpretable, transparent and general. We created two instances of this model: $\model$, which classifies images of dogs vs cats, and $\advmodel$ for dogs vs squirrels.\footnote{Training data: Dog, cat, and squirrel classes extracted from Animals-10 dataset from \label{fn:std}\url{https://www.kaggle.com/datasets/alessiocorrado99/animals10}} 

Since we seek to measure the trace separation between the two models, the shared category dog provides a way of linking the performance of the two models. Both $\model$ and $\advmodel$ achieved similar performance on the common task; they achieved 95\% and 98\% accuracy, respectively, on this task, with a 95.3 and 98.4$\%$ F1-score. The presence of the dog class  allowed us to assess how internal representations vary between models with the same architecture but contrasting training classes. 
We hypothesized that this training difference would show up as a detectable and measurable divergence in their internal activation patterns, even when both were processing images from the common dog class.

To answer the questions in \cref{sec:om-exps:design}, we evaluate, using \cref{eqn:mformula}, activation differences (1) across paths through the models and (2) across both models overall. Given the relatively small size of the model (ResNet-18 contains approximately 11M parameters), analyzing all activation values in addition to individual paths is feasible. In our analysis, we focus on two classes of layers: those for \emph{feature extraction} (i.e., Conv2d layers at various depths) and those for \emph{classification} (the final layers including the linear layer that produces the binary classification output). 
 
To ensure robust statistical analysis, we used activations from 3,522 test images for dogs. The images used were those that were positive instances in both models.

\subsubsection{Results}

We first analyzed the distribution of activation differences that occurred during inference on dog images\footnote{From \label{fn:sed}\url{https://www.kaggle.com/datasets/chetankv/dogs-cats-images}.}  when the activations along the random path from the unverified model ($\advmodel$) were substituted into the verified model ($\model$). 

\paragraph{{Activation differences for paths.}} As demonstrated by the violin plot in Figure~\ref{fig:violin}, the initial feature-extraction layers show almost no difference in activations because both models learned comparable low-level features for dog recognition. However, the distributions greatly expand as data moves to deeper, more abstract layers. Despite reaching the same result, the models still diverge. This is particularly noticeable in the last few layers of the classifier, lending support to the trace separation property.

We examined the cumulative distribution of mean model separation values across 176,100 test instances (3,522 $\times$ 50 random path tests for each), which is presented in Figure~\ref{fig:meanneuron}. The minimum separation value, computed using \cref{eqn:mformula}, is 0.070. Although the two models had an similar classification score on the same task, this statistical analysis revealed measurable internal differences in their input processing. 

This pattern is also confirmed by the histogram in Figure~\ref{fig:histogram} of all activation values across 176,100 paths. The pattern is subtle at the initial layers, but increasingly shows much different distributions between the activation in the deeper layers. This reveals that simply looking at the output can mask many characteristic differences in their traces.

\paragraph{{Activation differences across layers.}} We processed the activation values of path-selected neurons layer-by-layer using Jensen-Shannon  divergence (JS) \footnote{Jensen-Shannon divergence (JSD) measures the similarity between two probability distributions. It is a symmetric and smooth version of Kullback-Leibler divergence ($D_{\text{KL}}$), defined as $\text{JSD}(P \parallel Q) = \frac{1}{2} D_{\text{KL}}(P \parallel M) + \frac{1}{2} D_{\text{KL}}(Q \parallel M)$, where $M = \frac{1}{2}(P+Q)$.} . The JS between the activation distributions of both models for path-selected neurons was low in the initial convolutional layers (conv1: 0.0, layer1.1.conv2: 0.198), indicating moderate representational overlap. The divergence values increased through the middle layers, reaching maximum at layer3.1conv2 (0.285), then decreased in the deeper layers (layer4.1.conv2: 0.093, fc: 0.245). The peak divergence in the middle layers corresponds to differences in intermediate feature representations, while the final classifier layer (fc: 0.245) shows how each model created decision boundaries to separate classes. These results show that the two models do exhibit measurable differences at certain layers, which means a path through these layers will also detect.

\paragraph{{Additional experiments}} We discuss additional separation results on classifiers in \cref{sec:addl-sep}.

\subsection{Model Separation in LLMs}
\label{sec:llm-sep}

\subsubsection{Setup}
To answer the questions in \cref{sec:om-exps:design} in the setting of LLMs, we study Llama-2-7b-chat-hf ($\model$) and Llama-2-7b-hf ($\advmodel$), two widely-used open-source models for general text generation \cite{touvron2023llama2openfoundation}. Because the research community has used them extensively, our results on them are more readily interpretable. Both contain $7$ billion weights across $32$ transformer layers with hidden dimension of $4096$, processed using FP16 precision on the GPU. %Our experimental framework, implemented in PyTorch (version: 2.8.0+cu128) with Hugging Face Transformers (version: 4.56.1), performs a comprehensive comparison at both the weight and activation levels. 

As in the ResNet-18 experiments above, we wish to establish that $\model$ and $\advmodel$ are similar.
Our ResNet-18 experiments are oriented around demonstrating model separability clearly by comparing binary classifiers with a common class. But, with our LLM experiments, we wish to show a scenario closer to a real-world deployment for verifiable AI: separability between a fine-tuned model and its base version.

We extracted each LLM's weights and measured the weight differences using the Frobenius norm. The models were closely aligned: cross-layer parameter analysis revealed high consistency across the model's $291$ weight matrices, with a mean relative Frobenius norm of $0.052$.\footnote{Normalization layers showed the highest consistency ($0.010 \pm 0.003$). MLP and attention layers had similar magnitudes (0.063--0.066) with attention layers exhibiting greater variance ($\pm0.014$ vs $\pm0.003$), and the output layer showing expected higher deviation ($0.124$) due to its specialized function.}

%In order to conduct a weight analysis to ascertain that both models are closely aligned with one another, we extracted  including embedding layers, self-attention components (Q, K, V, O projections), Multi-layer Perceptron (MLP) (gate, up, down projections), layer normalization weights, and the output head. 

In this experiment, we analyze model separation between across selected paths.
To validate the formula (\cref{eqn:mformula}) the experimental setup utilized forward hooks to capture intermediate activations during inference, allowing us to analyze internal representations. 
Unlike our ResNet-18 experiments, we do not also analyze the activation differences across both LLMs as these models are prohibitively large to study. We argue that showing divergence in the path activations is sufficient to establish trace separation; indeed, if we can show that one path through the activation traces is different, then necessarily the traces are different. 

In each trial, we select neurons at each layer to form path through the network and measure the activation differences between the two models.  We also verified that selecting the neuron with minimum separation value (most similar) at each layer yields consistent results, confirming that the observed separation is robust across different path selection strategies.

\subsubsection{Result}
\begin{figure}[t]
    \centering
    \includegraphics[width=0.6\linewidth]{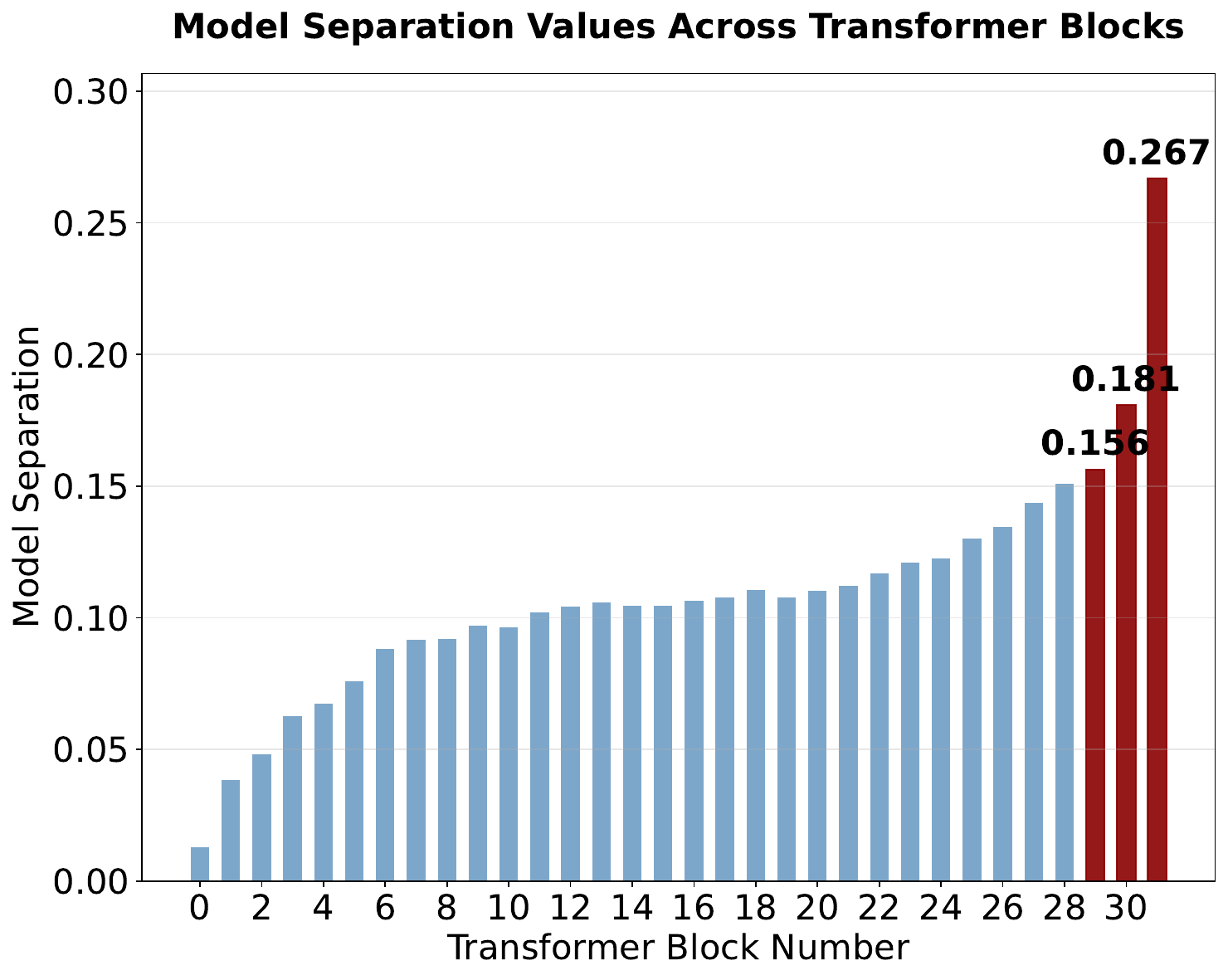}
    \caption{Model separation values across transformer blocks between Llama-2-7b-hf and Llama-2-7b-hf-chat models. The chart shows the mean separation value (Eq. \ref{eqn:mformula}) for each transformer block. Top 3 blocks  are in red. Dataset: 25 sequences.}
    %\caption{Activation difference analysis across transformer blocks showing differences between Llama-2-7b-hf and Llama-2-7b-hf-chat models for selecting the neuron exhibiting the minimum activation difference per layer, with activation patterns. Each bar shows the average absolute difference in activation values between two models for each transformer block. Analysis includes test cases across (15 sequences sample with 13 tokens on average over one round).}
    \label{fig:translayer}
\end{figure}
\begin{figure}[t]
    \centering
    \includegraphics[width=0.7\linewidth]{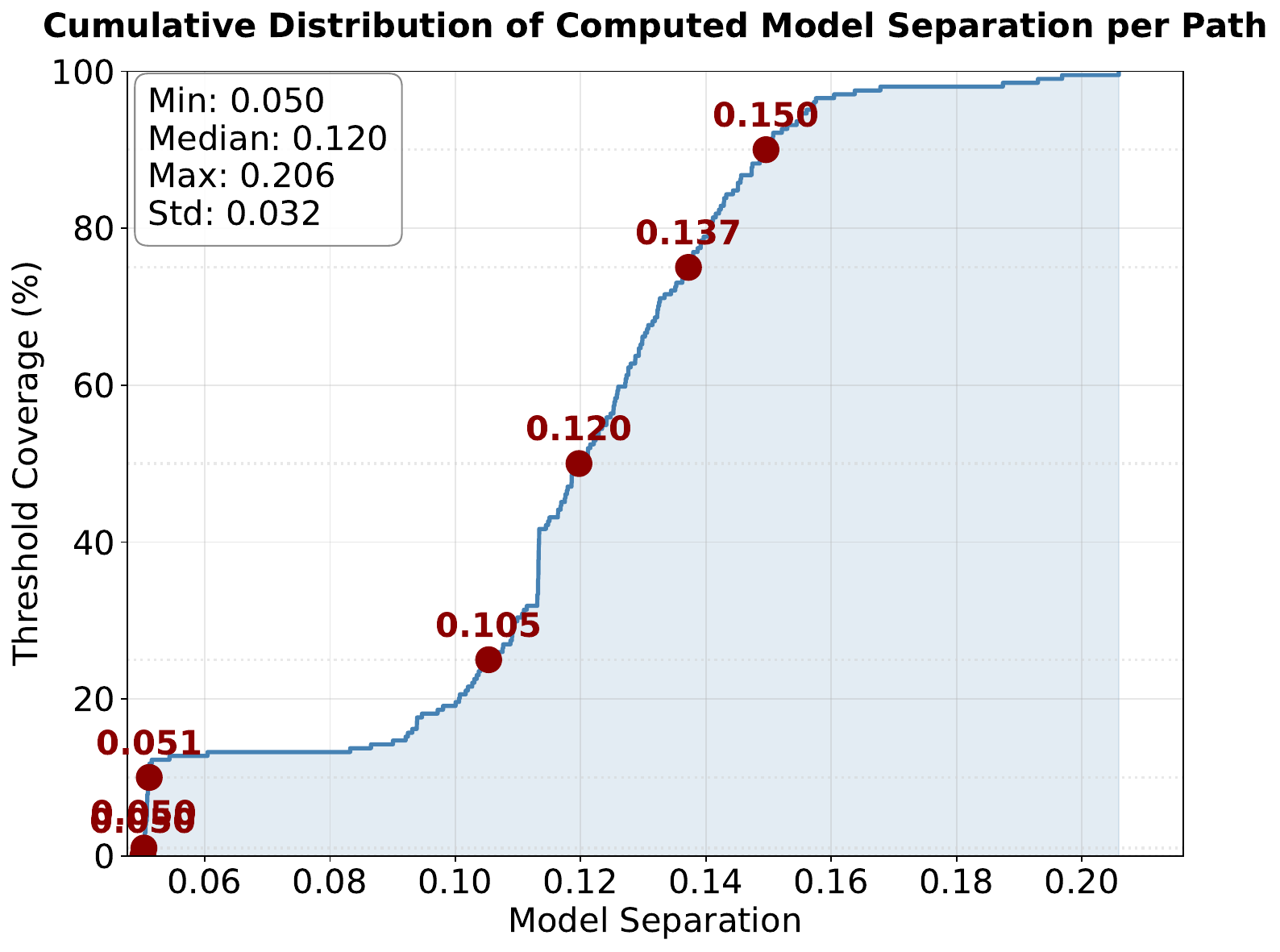}
    \caption{Cumulative distribution of model separation (\cref{eqn:mformula}) between Llama-2-7b-hf and Llama-2-7b-hf-chat, focusing on the last token (output) for each sequence (x: separation value, y: ratio of tests below that threshold, red dots: percentile thresholds). Neurons are randomly selected per layer, with mean values aggregated per path.}
    \label{fig:distLama}
\end{figure}
We recorded intermediate representations at each layer during inference on $25$ token sequences. This led to $290$ layer-wise comparisons per test case. The separation values showed consistent patterns that diverged markedly from the weight similarity, owing to nonlinear activation functions.

The model separation values revealed a clear progressive pattern across the model's 32 transformer blocks (\cref{fig:translayer}). Early blocks (0--8) showed minimal separation with mean values around 0.013--0.092, suggesting that basic feature extraction remains largely unchanged. Middle layers (9--24) exhibited a gradual increase in separation values, reaching around 0.096--0.123. Late layers (25--31) demonstrated sharp spikes in separation, particularly in blocks $25$ (0.130), $30$ (0.181), and $31$ (0.267). The final block ($31$) exhibited the highest mean separation value of approximately $0.267$, indicating that the final representations between the models are the most different. These results show that as earlier blocks capture low-level features, it is harder to distinguish between models, while later blocks process more complex features and show higher separation values, making it easier to distinguish between models that differ in training. This supports the notion of trace separation.
\cref{eqn:mformula} effectively differentiates between closely related models (\cref{fig:distLama}). The model separation analysis shows that the smallest separation is approximately 0.013 (Block 0) and the median separation is approximately 0.105. All of the separation values are greater than 0.01, which means that the differences are clearly distinguishable, hence the experiment design is sound. This shows that our method can reliably find variations in activation patterns across transformer layers, even in models that are very similar in structure. These results demonstrate that our path testing protocol was able to distinguish between two models with different evaluation traces, confirming that the observed separation is robust across different path selection strategies.

\section{Adversarial Attacks on Model Separation}
\label{sec:fs-exps}

We now experimentally investigate easy and efficient attacks which could be mounted against the protocol. To do this, an adversary using $\advmodel$ would need to emulate the valid model in such a way that the commitments to its activations which give the result $y_A$ would seem to be those produced from $\model$.

\subsection{Attack Design}
We designed attacks to determine how difficult it is for an adversary to forge a trace that passes $\mathsf{RandPathTest}$ against the committed model $\model$. To evaluate the efficacy of each attack, we computed the model separation values at each layer using \cref{eqn:mformula}. For the evaluation criteria during path testing, the mean of the set of separation values $\{s_i : 1 \leq i \leq L\}$ should be zero, where $s_i$ denotes the separation value at layer $i$ and $L$ is the total number of layers. High separation values indicate that the adversary failed to reconstruct the verified model's activations.

While the model separation experiments of \cref{sec:om-exps} show that submitting the honest trace of a substitute model $\advmodel$ is detected (other-model soundness), a more capable adversary might try to \emph{forge} a trace: run $\advmodel$ to obtain the output, then construct an arbitrary trace that passes $\mathsf{RandPathTest}$ against $\model$'s weights. This corresponds to the strong other-model soundness notion defined in \cref{sec:strong-om-snd}.

This threat model is practically relevant: an adversary who substitutes a cheaper, quantized, or dealigned model knows $\model$'s committed weights and can attempt to reverse-engineer activations consistent with them. We explore two natural reconstruction strategies:

\subsubsection{Gradient Descent}
\label{sec:fs-exps:ad:grad}
Gradient descent is used during training to adjust the weights of a model through an iterative process. As such, we wish to determine if an adversary can reconstruct the activations in this way, using gradient descent to optimize over \emph{activations} rather than weights. To do so, we first compute: $\outfn \leftarrow \model(\qry)$, then calculate the loss as 
$L= (\advmodel(\qry') - \outfn)^2$ and
$\qry'_i = \qry'_{i-1} - \alpha \frac{\partial L}{\partial \qry'_{i-1}}$, with $\alpha$ being a learning rate. We repeat this process until the gradient descent converges. 

Our test in~\cref{sec:A1} looks at the possibility that the adversary can perform a reconstruction attack by gradient descent on two different models: a simple ANN and the Llama2-7B LLM. 
We hypothesized that simpler models with fewer weights might exhibit a higher susceptibility to adversarial reconstruction, as their reduced complexity could help an attacker's ability to reverse-engineer inputs from observed activations, given knowledge of the data distribution. 

A special case of this attack involves an adversary who attempts to change the semantics of LLM output by swapping output tokens and then trying to reconstruct activations to falsify a proof that the output came from the original model. Due to space constraints, we discuss this attack in \cref{swap_attack}.

\subsubsection{Inverse Transform} We also explored the feasibility of an inverse transform attack using a simple ANN architecture. Due to space constraints, we discuss this attack in \cref{ann_inverse_transform_attack}.

\subsection{Gradient Descent Attack on Basic Models}
\label{sec:A1}

In this experiment, we use a simple, custom model to investigate whether it is possible to reconstruct internal activation structure of a neural network using only model’s weights and observed output via gradient descent.

\subsubsection{Setup}
\label{sec:A1:setup}
% \begin{figure}
%     \centering
%     \includegraphics[width=0.7\linewidth]{figures/TPR-SQ1.pdf}
%     \caption{Layer-wise threshold pass rate analysis for (a) minimum, (b) mean, and (c) maximum activation differences.}
%     %\caption{Layer-wise threshold pass rate analysis for activation differences. Heatmaps show percentage of samples below threshold for (A) minimum, (B) mean, and (C) maximum neuron differences. Each layer ($L_1$, $L_2$, $L_3$) evaluated independently. Color scale: 0\% (red) to 100\% (green) pass rates across threshold values $10^{-4}$-0.1.}
%     \label{fig:soundness_rate_heatmap}
% \end{figure}

The custom model we used in this experiment as our target was intentionally designed to make the adversarial attack as simple as possible: a 3-layer fully connected ANN using min-max scaling (so all values are in the range $[0,1]$), trained on a dataset containing 150 samples and 4 features for 3 species of irises\footnote{\url{https://scikit-learn.org/stable/datasets/toy_dataset.html\#iris-dataset}}. The input layer takes in four features, the two hidden layers have $64$ and $32$ neurons each and use ReLU activation functions, and the output layer has 3 neurons for classification.

 %To test whether adversary will able to recover the network using inverse transformation \tushar{not inverse but grad descent}, we intentionally selected a simple experimental setup that used the PyTorch (version: 2.8.0+cu128) framework and CUDA to support GPU acceleration. The test used an ANN that was trained on the iris dataset \tushar{Need a footnote}. This dataset has $150$ samples and $4$ features that show sepal and petal measurements for three different types of iris \tushar{The flower?}. We used Minmax scaling to make sure that the values were all in the range of $[0,1]$ to make adversarial attack easier. We built a fully connected ANN with only three layers on purpose \tushar{Rephrase: We built a fully connected ANN that was designed to make the adversarial attack as easy as possible: 3 layers, minmax scaling, etc.}.   
 
 The adversarial attack was configured to run $10,000$ epochs with a learning rate of $0.005$, using mean squared error (MSE) loss to minimize the difference between target and predicted activations. We conducted experiments with $125$ different input samples, each tested across $50$ rounds to ensure statistical validity, as the attack is non-deterministic. 

\subsubsection{Results}

Our results can be found in \cref{tab:soundness_thresholds}, where we report the minimum, mean, and maximum separation values across each of the three layers in our custom model, calculated using \cref{eqn:mformula}. The table contains pass rates, i.e., the percentage of adversarially constructed activations whose separation values fall under each threshold. Since the adversary's goal is to achieve zero separation, high pass rates at low thresholds indicate weak soundness (adversary succeeds), while low pass rates indicate strong soundness (verification system successfully detects adversarial inputs). These thresholds thus represent a tradeoff between completeness and soundness: the lower the threshold, the stricter the verification criteria and the better the soundness, though overly strict thresholds may reject legitimate inputs and impact completeness due to floating point error.

Based on the results, there is a steep drop in pass rates as the reconstruction target moves from hidden layers toward the output layer. When examining minimum separation values~-- the best-case scenario for the adversary~-- the earlier layers show that adversarial reconstruction succeeds: $L_1$ and $L_2$ both achieve $100\%$ pass rates at the $6.7 \times 10^{-5}$ threshold, meaning adversaries can easily generate inputs that fool verification for these layers. However, the output layer ($L_3$) exhibits significantly stronger resistance, with only a $5.1\%$ pass rate at the same threshold, showing how reconstructing output layer activations from adversarial inputs is more challenging.

Note, however, our attacker must pass the separation threshold at \emph{every} layer simultaneously, not just one of them. We can see the impact of this fact in \cref{fig:soundness_rate_gd}: The ``All Layers'' curve represents the percentage of attack attempts where all three layers simultaneously pass the threshold. Notably, this curve exactly matches the $L_1$ curve, revealing that $L_1$ acts as the bottleneck layer. This occurs because the adversarial optimization targeting the output layer creates a cascading effect where reconstruction difficulty increases for earlier layers, with $L_1$ consistently exhibiting the highest separation values. At the $10^{-3}$ threshold, while $L_3$ achieves a 78.6\% pass rate, the ``All Layers'' rate matches $L_1$ at 0.0\%. 

% This pattern is further demonstrated by the heatmap in Figure~\ref{fig:soundness_rate_heatmap}, which shows that the output layer ($L_3$) is actually more resistant to reconstruction than hidden layers.
%Based on the formula \tushar{Which formula?}, it requires all layers to simultaneously pass the threshold \tushar{Not clear from the formula above, if that's the one}, none of the samples succeed at the $10^{-4}$ threshold, showing the challenge of achieving consistent reconstruction across the whole network even if the architecture is simple. These findings show that significant differences in internal information processing can be masked by adversarial attacks that focus solely on output alignment, thus highlights the challenges of using output-based verification for model integrity and suggests that the uniqueness of internal activation patterns must be considered by robust interpretability methods.

\begin{table}[t] % placed here for space saving
	\centering
	\caption{\% of Samples Passing Threshold by Metric \& Layer}
	\label{tab:soundness_thresholds}
	\begin{tabular}{l|l|cccc}
		\hline
		&  & \multicolumn{4}{c}{Threshold} \\
		\cline{3-6}
		Metric & Layers & $10^{-6}$ & $6.7 \times 10^{-5}$ & $4.6 \times 10^{-3}$ & 0.308 \\
		\hline
		Minimum & $L_1$ & 100.0\% & 100.0\% & 100.0\% & 100.0\% \\
		& $L_2$ & 100.0\% & 100.0\% & 100.0\% & 100.0\% \\
		& $L_3$ & 0.1\% & 5.1\% & 99.9\% & 100.0\% \\
		\hline
		Mean & $L_1$ & 0.0\% & 0.0\% & 2.5\% & 100.0\% \\
		& $L_2$ & 0.0\% & 0.0\% & 27.5\% & 100.0\% \\
		& $L_3$ & 0.0\% & 0.0\% & 98.7\% & 100.0\% \\
		\hline
		Maximum & $L_1$ & 0.0\% & 0.0\% & 0.0\% & 90.0\% \\
		& $L_2$ & 0.0\% & 0.0\% & 0.4\% & 100.0\% \\
		& $L_3$ & 0.0\% & 0.0\% & 96.7\% & 100.0\% \\
		\hline
	\end{tabular}
	
\end{table}
\begin{figure}[t]
    \centering
    \includegraphics[width=0.65\linewidth]{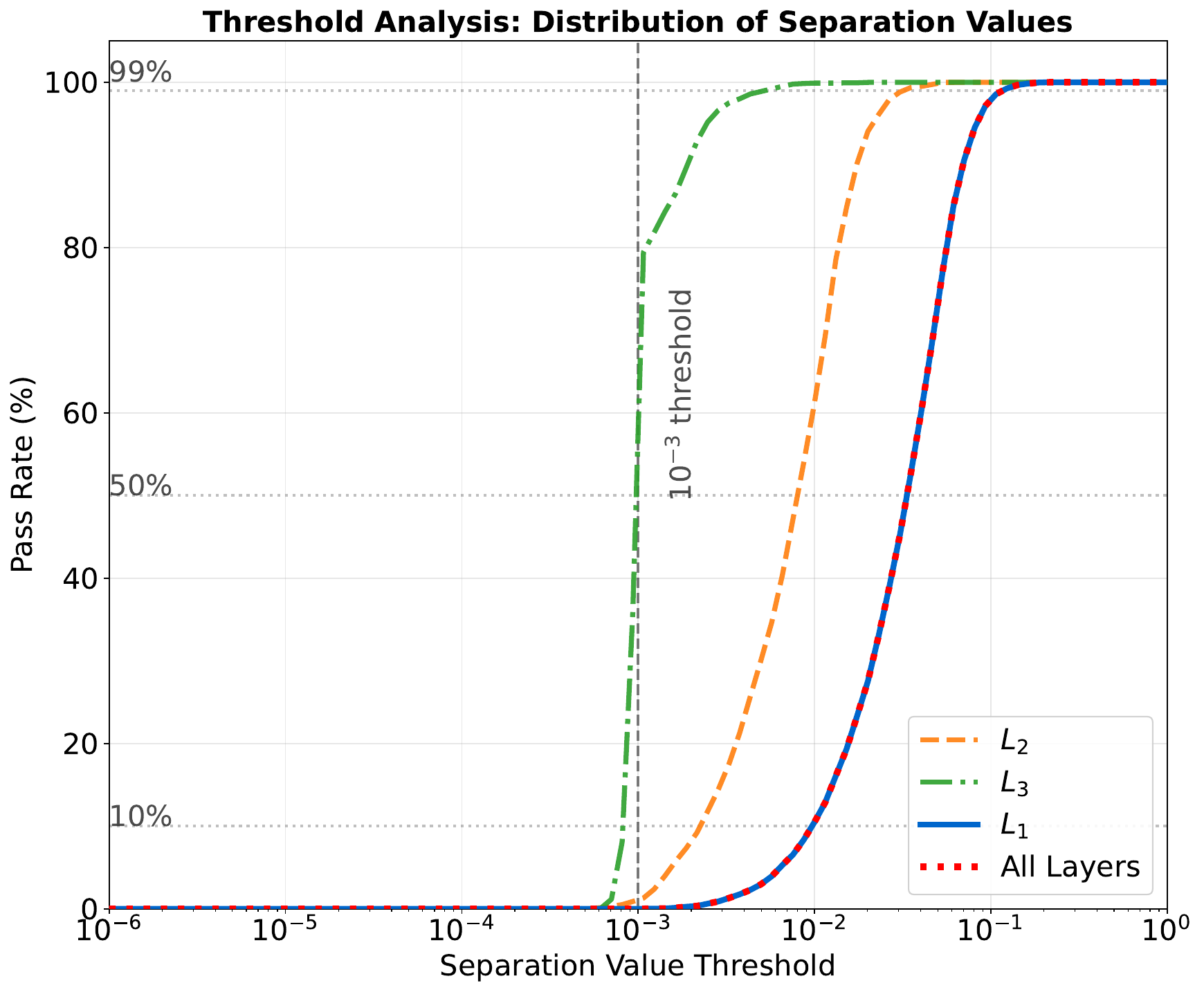}
    \caption{Pass rate analysis of mean separation values (log-scale). Individual layer curves $(L_1,L_2,L_3)$ and the ``All Layers'' requirement show $L_1$ to be the bottleneck for attack success.}
    \label{fig:soundness_rate_gd}
\end{figure}

\subsection{Gradient Descent Attack on LLMs}
\label{sec:A2}
Our second gradient-descent attack targets Llama-2-7B-chat-hf, treating the embedding layer as the input layer.

\subsubsection{Setup}
%Llama Gradient

\begin{figure}[t]
    \centering
    \includegraphics[width=0.65\linewidth]{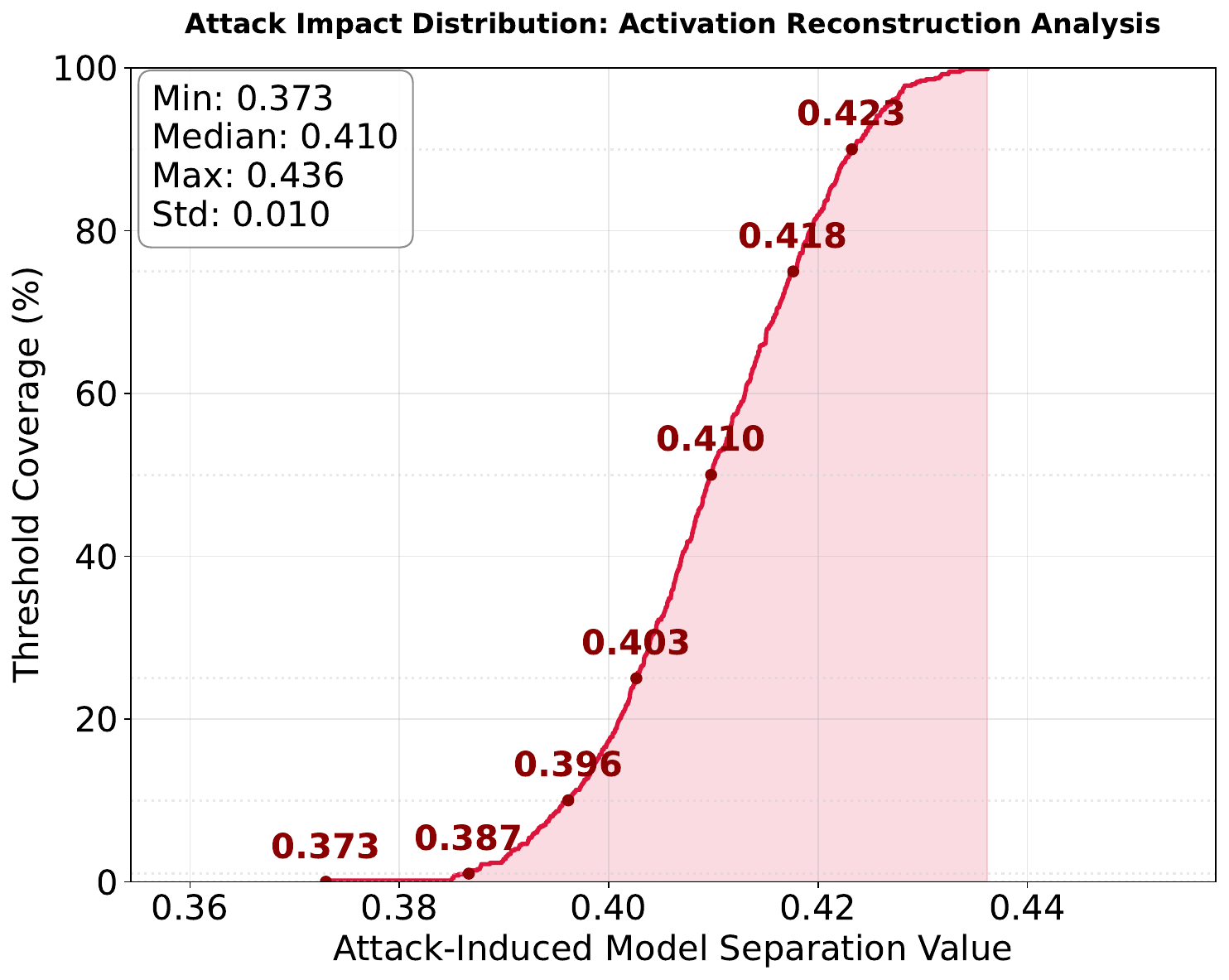}
    \caption{Model separation (Eq. \ref{eqn:mformula}) for Llama-2-7b-hf-chat at last token position, with neurons randomly selected per layer, and mean values aggregated per path. (x: separation value, y: ratio of tests below that threshold, red dots: percentile thresholds).}
    %\caption{Cumulative distribution of cross-model activation differences with minimum aggregation across layers. Shows soundness coverage versus activation difference thresholds, where each test case's activation paths are aggregated by taking the minimum difference across all layers. Key percentiles (0th, 1st, 10th, 25th, 50th, 75th, 90th) are marked with red dots and annotated values. Log-scale x-axis displays the range of activation differences, with statistics box showing distribution summary (Min, Median, Max, Std).}
    \label{fig:soundness_rate}
\end{figure}

%using PyTorch (version: 2.8.0+cu128) with CUDA acceleration on GPU hardware. %The model was loaded in float16 precision with automatic device mapping to optimize computational efficiency while maintaining sufficient numerical precision for gradient calculations. 
We implemented forward hooks across 290 layers spanning all 32 transformer blocks in Llama-2-7B-chat to capture intermediate values during both original inference and reconstruction attempts. The experimental samples consisted of 22 distinct text prompts covering various domains, each tokenized to a maximum sequence length of 10 tokens using the model's native tokenizer with padding. 

For each input, we performed $30$ independent reconstruction attempts using Adam optimizer with learning rate $0.01$ and L2 regularization ($\lambda = 0.001$), with  up to $10,000$ iterations per attempt or until convergence ($\textit{loss} \leq 1e-4$). 

For each reconstruction, we evaluated 5,000 randomly selected neuron paths to measure model separation, plus one additional path where every neuron was tested and the one yielding the lowest path separation value was selected. This resulted in approximately 3.5M path evaluations across all experiments (25 inputs $\times$ 30 reconstructions $\times$ 5,001 paths per reconstruction). This ensures statistical validity and demonstrates robust security against adversarial reconstruction even under relaxed time-complexity constraints.

\subsubsection{Results}

As shown in \cref{fig:soundness_rate}, we measured the model separation between reconstructed activations and the true activations of the valid network. 
We find that across all 3.5M path evaluations, no reconstruction achieved separation values below 0.37, with a minimum observed separation value of 0.373. The cumulative distribution shows that the mean separation value is 0.410, the standard deviation is 0.010, and a median of 0.410, demonstrating that adversarial reconstruction efforts consistently failed to approximate the original activations. 

Key percentiles further illustrate the robustness: the 1st percentile is 0.387, the 10th percentile is 0.396, the 25th percentile is 0.403, the 75th percentile is 0.418, and the 90th percentile reached 0.423. These results show that even when the adversary has the best possible conditions---including relaxed computational time and full model weights access---the model still has sufficient activation complexity to prevent successful activation reconstruction. 

These findings demonstrate that despite relaxing computational constraints and providing complete model information, the inherent non-convexity of the optimization landscape and high-dimensional nature of transformer activations provide robust protection, with 100\% of attempts failing to achieve meaningful reconstruction quality. Even the best-case reconstruction attempt maintained activation differences of $0.373$, indicating substantial deviation from the original activations

\section{Performance Evaluation}
\label{sec:perf}

This section evaluates the performance of our verification protocol and compares it against zkLLM~\cite{CCS:SunLiZha24}, a state-of-the-art scheme based on cryptographic proofs tailored for LLMs.

\subsection{Evaluation Design and Setup}
All experiments were conducted on a machine equipped with an AMD Ryzen 9 9900X CPU, 32GB of RAM, and an NVIDIA GeForce RTX 3090 GPU. Our evaluation focuses on the Llama-2-7B model. For zkLLM, we benchmarked the performance of their system using the publicly available implementation\footnote{Available at \url{https://github.com/jvhs0706/zkllm-ccs2024}.} For our protocol, we measure the time to commit to the intermediate activation tensors of the network and the proof size.
Specifically, our protocol generates a commitment to the trace by committing to the $192$ intermediate tensors in the execution trace, each with dimensions of $64 \times 4096$. To compute this efficiently, we use a hybrid commitment scheme that adapts well to the underlying architecture. This scheme uses a Merkle tree built on row-wise hashes (computed using SHA-256) but varies what it commits to:
\begin{itemize}[itemsep=0pt,topsep=0pt]
\item For matrix multiplication layers, where a non-learned input activation matrix is multiplied by a weight matrix, we commit only to the input activation matrix.
\item For element-wise transformation layers, we commit directly to the output tensor.
\end{itemize}
This hybrid commitment scheme is key to our protocol's high performance, as it does not serialize the entire computation. The proof contains the relevant row and the corresponding Merkle path for each layer in the network.

The zkLLM proof sizes were taken from their paper for the Llama-2-7B model. Although the hardware for the original zkLLM experiments differs, using our own timing benchmarks and their reported proof sizes provides a valuable baseline for understanding the relative performance of the two approaches.

\subsection{Results}
Our analysis reveals a significant performance advantage for our proposed protocol in terms of speed, with a trade-off in proof size. The performance of zkLLM on the Llama-2-7B model is as follows. The proving time for the entire inference process is 388.3 seconds. The proof size is 183 kB. The verification time is 2.36 seconds.

By switching from cryptographic assumptions to the separability assumptions demonstrated in \cref{sec:om-exps}, our protocol demonstrates a substantial improvement in efficiency. For proving time, the commitment and proof generation are remarkably fast. The time to commit to the entire trace is only 5.8 milliseconds. With the tensors in device memory, the proof generation for a single Merkle opening is on average 0.007 milliseconds. For the proof size, the base proof (the activation values of the neurons on the path) is 3.14 MB. This reduces to 1.6 MB using Brotli, a general purpose lossless compression algorithm. The Merkle openings (the cryptographic proof for the opened activations) add $\approx$100 KB to the total proof size, and this portion is not compressible. The total proof size is therefore $\approx$ 1.7 MB. Crucially, in our construction, the client does not possess the full model parameters locally but instead it holds a commitment to it. Therefore, the proof must also include the model weights and their Merkle openings. This effectively doubles the payload size. Consequently, the total proof size is approximately 3.4 MB. Our verification time is 12.44 milliseconds (two orders of magnitude lower).

\section{Conclusion and Future Work}
\label{sec:concl}
This work presents a novel and scalable verification framework for AI model inference that avoids the computational cost of producing a cryptographic proof for the full execution of a large neural network. A key conceptual contribution is a formal framework that allows protocol designers to leverage statistical properties of neural networks---specifically, trace separation between functionally dissimilar models---to argue the security of verifiable inference protocols. We instantiate this framework with a concrete protocol based on a ``random spot checking paradigm'': the prover commits to the execution trace and opens only a small number of entries along randomly sampled paths from output to input. To validate our design (and the assumptions underpinning it), we designed and ran extensive experiments showing that models with different input/output behavior can be distinguished based on a few queries on their traces (list of all activation values).

Our experiments also provide evidence that our protocol resists natural attack strategies beyond the other-model soundness setting our framework formally captures. We designed and ran experiments showing that for the specific strategies we tested---gradient-descent reconstruction, inverse transforms, and logit swapping---it is still difficult to produce a ``fake'' trace that validates an incorrect output. Achieving full soundness, where the adversary may produce arbitrary outputs and forge traces without restriction, remains an open problem. In particular, any single-path sampling strategy faces an inherent detection bound of $1/N$ (where $N$ is the maximum layer width), which may be limiting when the adversary can concentrate modifications in few nodes. Understanding how multi-path sampling and adaptive strategies can tighten this bound is an important direction for future work.

Additionally, we evaluated the proof size and time complexity of our verification protocol and showed that it greatly outperforms a state-of-the-art scheme for LLM verification in proving time, with moderately larger proofs.

Several promising avenues exist for future research. First, we welcome additional analysis of the protocol against more sophisticated and adaptive adversaries. For instance, a compelling adversarial strategy would be to train Generative Adversarial Networks to produce activation paths that are both internally consistent and statistically similar to honest traces. Additionally, exploring the security implications in applying the protocol to quantized models is important, as an adversary can exploit the rounding errors inherent in these lower precision architectures.

The concept of trace separation is widely agreed upon within the machine learning community and our experiments provide evidence of this can be useful for verifiable inference. That being said, more formal theoretical guarantees of trace separation are needed to prove security. While work in formalizing the inner workings of neural networks is ongoing~\cite{marchetti2025algebra}, the fact that no such theorem exists may be more of an indicator that such a result is non-trivial and will require future investigations.

Another key direction is expanding the applicability of the protocol to architectures with fundamentally different data flows, such as Graph Neural Networks (GNNs)~\cite{gnnsurvey}. Developing a secure non-interactive version of our scheme is an open problem: applying the common Fiat-Shamir heuristic to make the scheme non-interactive requires care  to avoid known issues in ``holographic'' protocols  like ours (based on accessing only a small part of the computation/inputs) and with non-neglible soundness error \cite{matteofs}. 

Finally, we plan to enhance the privacy of our protocol by using techniques from the cryptographic proofs literature (in a targeted manner) to avoid revealing raw activation values or potentially proprietary model parameters in the proofs. 
A na\"ive approach may be to simply apply a zero-knowledge proof on top of the verification scheme of our protocol, i.e., upon a path query by the verifier, the prover would answer with a zero-knowledge proof that the openings of the committed values satisfy the correct equations, without revealing the values. We expect this to still be faster than running the zero-knowledge proof directly on the entire network (as in zkLLM~\cite{sun2024zkllm}) since we are running it on a much smaller computation (i.e., the path). There are additional details to be considered, however, leaving an interesting space for future work.
\ifShortversion

\section{LLM usage considerations}

%\tushar{Section does not count towards the page limit.}

Since the design of our protocol for verifying the output of neural networks does not rely on any specific architecture, several of our experiments required the use of LLMs to provide reliable and interpretable results on our claim. Additionally, we do a comparative study of the size of the proof and the time complexity of opening and checking the proofs against a proposed protocol, zkLLM, a state-of-the-art proof scheme tailored for LLM, which are currently the most prevalent architectures of neural networks being deployed in the most general settings. In our experiments, we choose to use Llama 2 with 7 billion parameters, the smallest size released. This model has other important features which made it an important choice: It is open source and which makes it popular and well-understood by the research community which we felt strengthened the interpretability and transparency of our results; its hardware requirements were well within the scope of most retail model computers, making them an affordable option; and it has a lower environmental footprint in the ongoing inference stage. We note that our experiments were not run for more than 40 hours in one session and we did not do any retraining of the open-source model. Furthermore, no LLM commercial or otherwise was used for editorial purposes in this manuscript.  
\fi

%Part of the experiments in this work were conducted on LLMs. 

%need for llms -- popular use case for outsourced computation common deployed model type, useful to show applicability prior work on studying llms (zkllm), scientific contributions.

%In our experiments, we choose to use Llama 2 with 7 billion parameters, the smallest size released. We chose this model for several reasons:

%\begin{itemize}
%    \item open source
%    \item popular with research community
%    \item tractable on commodity hardware. 
%    \item lower environmental footprint of runtime. 
%\end{itemize}

%local hardware, community infrastructure -- using collaborative resources. 

%fewer queries, live analysis. did not re-train -- some fine-tuning (num hours?). pre-trained model, re-use prior work

%LLMs were \emph{not} used for editorial purposes in this manuscript. 

\section{Acknowledgments}

The authors would like to thank Miranda Christ, Nelly Fazio, Ed Felten, Sanjam Garg, Nicola Greco, Mariana Raykova, Adi Shamir and Raúl Vicente for useful discussions. In particular, it was a conversation with Nicola Greco that planted the seed for the approach described in this paper. 

This work was supported by Google under the Cyber NYC program. 
The views and conclusions contained herein are those of the authors and should not be interpreted as those of the sponsor.

\ifFullversion
\else
\bibliographystyle{ieeetr}
\fi
\bibliography{abbrev3,crypto,extra}
\appendix
\ifFullversion
\else
\appendices
\fi
\label{app}

\section{Motivating Scenarios for Other-Model Soundness}
\label{sec:motivation-other-model-snd}
As discussed in \cref{sec:intuition}, other-model soundness is a weaker security notion than full soundness: it only guarantees detection when the adversary runs a different model $\widetilde{\mathcal{M}}$ and commits to its real execution trace, rather than following an arbitrary cheating strategy. Yet, this restricted notion is practically useful. We argue that it captures several natural and pressing threat models, which we discuss through the following scenarios.

\paragraph{Economic incentives for serving outputs from different models.}
Hosting state-of-the-art large language models is expensive. This creates a natural economic incentive for providers to substitute the advertised model with a cheaper alternative (a smaller variant, a heavily quantized version, or simply a less capable model during peak demand~\cite{substitution2025}). Even when the intent is not malicious (e.g., load balancing during traffic spikes), undisclosed substitution breaks service agreements. It  can  also invalidate benchmarks that users rely upon when selecting a provider. Our protocol directly addresses this: if a provider commits to the weights of the advertised model $\mathcal{M}$ but serves inference from a different model $\widetilde{\mathcal{M}}$, verification will detect this with high probability.\footnote{This is true even if $\widetilde{\mathcal{M}}$ is a quantized or distilled version of $\mathcal{M}$, provided the two models produce sufficiently different output distributions on typical queries.}

\paragraph{Safety alignment.}
Other-model soundness can also address threats to safety alignment. Modern LLMs undergo extensive alignment procedures to ensure they refuse harmful requests~\cite{ouyang2022training}. However, recent work has shown that these safety measures are surprisingly brittle. Fine-tuning on as few as 100 examples can strip alignment from open-weight models~\cite{qi2023finetuning}, and the cost of doing so is minimal (under \$200 for a 70B-parameter model). Meanwhile, ``uncensored'' variants (models explicitly modified to remove safety constraints) are widely available on public repositories~\cite{uncensored2025}.

A provider might substitute a safety-aligned model $\mathcal{M}$ with a dealigned variant $\widetilde{\mathcal{M}}$ to reduce user friction from refusals, or to handle requests the aligned model would decline. Since dealigned models produce outputs from fundamentally different distributions, our verification approach can detect such substitutions. This provides users and regulators with assurance that they are interacting with the certified, safety-aligned model rather than an unsafe substitute.

\paragraph{Training data compliance.}
In some settings model provenance matters legally or contractually. Consider a provider whose model $\mathcal{M}$ has been audited to certify that it was trained only on properly licensed data. The provider may be tempted to serve a different model $\widetilde{\mathcal{M}}$ if this alternative performs better. The latter may be trained on additional data, including material the provider had no right to use. This is not a hypothetical concern: the question of training data provenance is increasingly relevant in light of ongoing litigation over copyrighted content in training corpora~\cite{zhao2024watermarking}. Our notion of other-model soundness can capture this scenario.

\section{Formal Background on Vector Commitments}
\label{apdx:vec}

Informally, a vector commitment scheme maps a vector of length $n$ of elements in some alphabet\footnote{By \textit{alphabet} we refer to an arbitrary finite set. For simplicity, the reader can think of an alphabet as the set of binary strings of fixed length. However, an alphabet can be any convenient way to describe an object (in cryptography is occasionally common to use finite fields as alphabets, for example).}  $\alphab$.
Formally, a vector commitment scheme  over alphabet $\alphab$ consists of the following algorithms:
\begin{itemize}
	\item \(\mathsf{GenParams}(1^\lambda, n)\): On input a security parameter \(\lambda\) and vector length \(n\),
	outputs public parameters \(\mathsf{pp}\).
	\item \(\mathsf{CommitVec}(\mathsf{pp}, v_1,\dots,v_n)\): Outputs a commitment \(\cm\) to the vector \(v  \subseteq \alphab^n \).
	\item \(\mathsf{Open}(\mathsf{pp}, v, i)\): Generates a proof \(\pi_i\) for the value \(v_i\) at position \(i\).
	\item \(\mathsf{Verify}(\mathsf{pp}, \cm, i, v_i, \pi_i)\): Returns 1 if \(\pi_i\) is a valid proof that
	\(v_i \in \alphab \) corresponds to the value in the committed vector at position \(i\), and 0 otherwise.
\end{itemize}

A secure VC must satisfy \emph{binding} (no efficient adversary can produce two different openings for the same index that both
verify against the same commitment, except with negligible probability), \emph{position binding} (no adversary can produce an accepting opening proof that claims a false value $\tilde{y} \neq v_i$ for position $i$ against commitment $\cm = \mathsf{CommitVec}(\pp, v)$ ).
Efficient vector commitments are expected to have short commitments (often constant size), short proofs
(polylogarithmic or constant size), and fast verification.

\section{Cryptographic Verifiable Inference Schemes}
\label{sec:ver-inf-def}
\label{sec:basics-models}

\subsection{Basic Formal Definitions  for Models}
In order to be as precise as possible in our construction, we will define a very general abstract syntax and semantics for \textit{family of models}. For the focus of this specific paper, the reader can think of families of models as \textit{neural networks}; however, our notion is intentionally general in order to be potentially applicable in other settings.

\begin{definition}[Family of models]
\label{def:model}
A family of models over an alphabet $\alphab$ consists of a tuple $\left( \nM, \nQ, \nTrc, \Eval, \idxsOut\ifFullversion, f_{out}\fi\right)$ where:
\begin{itemize}
	\item  $\nM, \nQ, \nTrc \in \NN$ are positive integers denoting the size of the model, query and trace respectively (see also next items). Whenever we refer to an object as model, query or (evaluation) trace, we are referring to an element in the sets $\alphab^\nM, \alphab^\nQ, \alphab^\nTrc$ respectively.
	\item  $\Eval : \alphab^\nM \times \alphab^\nQ \to \alphab^\nTrc$ is an efficiently  computable deterministic function that associates to a model and a query a trace string. Given a model $\model \in \alphab^\nM$ and a query $\qry \in \alphab^\nQ$, we refer to $\trace := \Eval \left( \model, \qry \right)$ as the trace corresponding to the evaluation of $\model$ on $\qry$ (or evaluation trace for short). See also  \cref{rem:examples-models-trace} for some intuitions.
	\item $\idxsOut \subseteq \{ 1, 2, \dots, \nTrc \}$ is a set of of positions containing the outputs in an evaluation trace (see also definition of $\outfn$ below).
	\ifFullversion
	\item $f_{out} : \Sigma^{|\idxsOut|}\to \bit^*$ is a function that processes output-related parts of the trace---with indices $\idxsOut$---and returns the output (see also definition of out $\outfn$ below and \cref{rem:examples-models-trace}). 
	\fi
\end{itemize}
Given a family of models, we define two additional notations that will be convenient to talk about outputs:
\begin{itemize}
\item \textbf{Output values in a trace string: } we define $\outfn\left( \cdot \right)$ as the function that maps a  trace $\trace$ (an element of $\alphab^\nTrc$)  to its outputs. 
It is defined as some fixed function of some subset of the trace (e.g., some specific layer of the network).
Explicitly, \ifShortversion$$\outfn(\trace) := \left( \trace_j \right)_{j\in \idxsOut}.$$\else$$\outfn(\mathsf{trc}) := f_\mathsf{out}\left( \mathsf{trc}_j\right)_{j \in \mathsf{Idxs}_\mathsf{Out}}.$$\fi
Notice that the verifier can compute this on their own by querying the appropriate elements of the trace, i.e., the positions in $\idxsOut$.
%Notice that $\trace$ is not required to correspond to the evaluation of any model (function $\outfn$ simply returns the specific values in a string that contain output-related values).
\item \textbf{Output of a model run on a query: } given a model $\model \in \alphab^\nM$ and a query $\qry \in \alphab^\nQ$, we denote by $\model(\qry)$ the actual output of the evaluation of $\model$ on query $\qry$. That is, $\model(\qry) := \outfn \left( \Eval \left( \model, \qry \right) \right)$.
\end{itemize}
\end{definition}

\begin{remark}[Traces are binding to outputs]
\label{rem:on-outputs-of-traces}
We will use the following straightforward observations to simplify the definition of \cref{def:idealized} (see also \cref{rem:asymmetry-defs}). Consider a model $\model \in \alphab^{\nM}$, a query $\qry \in \alphab^\nQ$ and an arbitrary trace string $\advtrace \in \alphab^\nTrc$  (potentially provided by an adversary). Let $\trace := \Eval \left( \model, \qry \right)$ be the genuine evaluation trace of $\model$ on $\qry$. Then, if the adversary provides $\advtrace \in \alphab^\nTrc$ (the alleged evaluation trace of $\model$ on $\qry$), then it is also implicitly providing (and committing to) the alleged outputs of $\model$ on $\qry$; these are simply $\outfn(\advtrace)$ and, by definition of $\outfn$, can thus be simply retrieved by accessing\footnote{This type of random access is exactly the one that the client has in \cref{def:idealized} and that we enforce cryptographically (through vector commitments) in our construction/compiler in \cref{sec:compiler}.}  the positions in $\advtrace$ with indices in $\idxsOut$.
If $\advtrace$ is such that $\outfn(\advtrace) \neq \outfn(\trace)$ then, obviously,  by definition, $\outfn(\advtrace) \neq \model(\qry)$.
\end{remark}

\begin{remark}[Some examples and intuitions]
\label{rem:examples-models-trace}
For intuition, consider the case in which the family of models consists of feed-forward neural networks.
Here, $\nQ$ corresponds to (an encoding of) the input layer: a query $\qry \in \alphab^{\nQ}$ simply encodes the input features given to the network.
Similarly, $\nM$ encodes all network parameters (e.g., weights and biases flattened into a single string in $\alphab^{\nM}$).
The value $\nTrc$ naturally corresponds to the total number of neuron activations across all layers for a single forward pass.

Under this interpretation, $\Eval(\model,\qry)$ computes the usual forward pass of the neural network on input $\qry$ and returns a trace string containing all intermediate activations in a fixed order.
Finally, $\idxsOut \subseteq \{1,\dots,\nTrc\}$ is taken to be the set of indices corresponding to the output layer of the network, so that $\outfn(\trace)$ recovers exactly the final-layer activations.
Consequently, $\model(\qry) = \outfn(\Eval(\model,\qry))$ is the standard notion of the network’s output on input $\qry$.

\ifFullversion
For the case of an LLM, we could have $\idxsOut$ correspond to the layer $\vec{z}$ of logit values and then define $f_{out}$ as:
$$f_{out}(\vec{z}) := \mathsf{softmax}(\vec{z}).$$
\fi
\end{remark}

\subsection{More On Our Design Choices}
\label{appx:why-dist}
\paragraph{Model and Query Distributions.}
Our formalism describes models and queries as being sampled from distributions over them.
We remark that this is encompasses a broad range of situations. For example, in our definition of other-model soundness (\cref{def:idealized}) it includes the case where there is a single honest model (the support of distribution $\Dmodel$ is a singleton, e.g., comprising the output of a model trained on a legal dataset) and we are concerned that the adversary might use for inference a model that has been trained on an illegal dataset (in this case the support of $\advDmodel$ will also be a singleton, including only the output of the training on the forbidden dataset\footnote{We stress that the forbidden dataset may be adversarially chosen (as a function of the legal one) and may  in principle overlap substantially with the legal dataset.}).
Probability distributions over model are also expressive enough. For example, given a set $\mathcal{S}$ of ``plausible'' datasets, we could define $\Dmodel^{\mathcal{S}}$ so that sampling from that distribution, i.e., $\model \sample \Dmodel^\mathcal{S}$, essentially encodes ``$\model$ is any model with a specific architecture trained through a specific learning algorithm\footnote{With  specified parameters, such as learning rate and number of steps.} and on any of the datasets in $\mathcal{S}$''.
Having distributions on models and queries also permits \textit{constraining} appropriately what are models and queries where we believe our assumptions required for security are true.
Finally, we point out that we implicitly assume that all models we sample have the same architecture. This is without loss of generality: it is harder for the verifier to distinguish whether the adversary is using a particular model or not for inference when these have the same architecture. Therefore, whatever claim of security holds for the ``same architecture'' case, will hold for the ``heterogeneous architecture'' case.

\paragraph{Distance of Outputs.}
A ``successful'' adversary is intuitively one that is able to convince the client that the output of a model on a specific query is $\tilde{y}$, when instead it is some $y \neq \tilde{y}$. Our techniques will intuitively leverage the fact that the \textit{more distant} an adversarial output is from the honest one, the \textit{easier} it will be to detect it. For this reason we adopt a more general notion of ``successful'' adversary, one where the adversary may be required to provide not only $\tilde{y}$ distinct from the honest output, but one that is \textit{sufficiently far} from it (above a certain threshold).

We point out that the security against ``far enough'' outputs is still of substantial applicability. An application may be robust with respect to ``small deviations'' from the honest output and thus not require security in that particular case.
At the same time, there are intuitive ways in which supporting large deviations can be a useful security notion: depending on the distance definition, large deviations may for example capture the case of the adversary claiming an output in a completely different class (and this may be the type of false claims the application protects against). As an example, consider classification loss functions; a loss function using a soft target distribution (e.g., from word embeddings) penalizes mistakes less in cases where the predictions are of semantically similar class, while cross-entropy loss with one-hot encoding treats any misclassification as a large error.

\section{Our Protocol in the Refereed Model}
\label{sec:ref}

\subsection{Background}

An interesting development in the verifiable computation is the refereed model of computation~\cite{canetti2013refereed}, which addresses the problem of verifying a computation $y=F(x)$ in a setting where two parties, Prover $P$ and Challenger $V$, make conflicting claims, $y_1=F(x)$ and $y_2=F(x)$, respectively. The model operates under the assumption that at least one party is honest. The objective is to enable an efficient, sublinear-time verification process by a third party, the Referee $R$, to determine which claim is valid.

The protocol leverages a challenge-response mechanism akin to a bisection on the computation's execution trace: $P$ provides an intermediate state S of the machine's execution trace. If the state $S$ is inconsistent with the prover's claim, a discrepancy exists. $V$, assumed to be honest, identifies the segment of the computation—either from the initial state to $S$ or from $S$ to the final state—that is inconsistent with the correct execution. $V$ then challenges $P$ to reveal an intermediate state within the erroneous segment. This iterative bisection process continues until a single, incorrect computational step is isolated.

The protocol terminates after $O(\log T)$ rounds, where $T$ is the number of steps in the computation of $F$. The final round requires the Referee to verify only a single, one-step computation between two adjacent states. The simplicity and high efficiency of this protocol have made it a cornerstone of optimistic rollup architectures (e.g., Arbitrum \cite{kalodner2018arbitrum}, Optimism \cite{armstrong2021ethereum}), which accept off-chain computations as valid by default and only initiate a fraud-proof process—governed by this refereed model—in response to a challenge.

The application of the refereed model to verify the outputs of large-scale AI systems, such as deep neural networks, was initially proposed in~\cite{opt-ai}. While that work correctly identifies the inherent suitability of the refereed computation model for this domain—given the sublinear verification complexity and the often-extensive computational traces of large AI models—it does not provide empirical results or an experimental validation of this hypothesis.

\subsection{Construction}

In the refereed model we have two parties $P_1,P_2$ who make competing assertions about the output of a model $\model$ on input $\qry$. That is, $P_i$ claims $\model(\qry) = y_i$ with $y_1 \neq y_2$. We assume that both $P_1,P_2$ know the full description of $\model$ and that the vector of weights defining $\model$ was honestly committed to via a vector commitment $C_{\model}$ available to all parties (both $P_1$ and $P_2$ know how to open $C_{\model}$ in order to provide proof of correctness for individual weights). As in our previous protocol, $C_{\model}$ serves as the ``ground truth'' about the model we are trying to verify.

Let $n$ be the total number of nodes in model $\model$, and assume that the nodes are linearly ordered in a canonical way via topological sort. Without loss of generality we assume that the output of the model is the activation of the last node and the total number of nodes $n$ is a power of 2 $n=2^m$. Let $H$ be a collision resistant hash function. The parties  engage on the following protocol:
\begin{itemize}
    \item $P_i$ produces a commitment to a claimed hashed trace $C_i={\sf VC}(\htrace_i)$. If $P_i$ is honest it computes
    \begin{itemize}
        \item $\trace_i=\Eval(\model,\qry)=[a_{i,0},\ldots,a_{i,n}]$ the vector of activation values;
        \item $A_{i,k}=[a_{i,k},H(a_{i,0},\ldots,a_{i,k-1})]$ the pair of activation value $a_{i,k}$ and the hash of all the previous activation values;
        \item $\htrace_i=[A_{i,0},\ldots,A_{i,n}]$
    \end{itemize} Note that $A_{1,n} \neq A_{2,n}$ since the parties make competing assertions. Also note that there exists an index $k_{\sf in}$ such that if $k \leq k_{\sf in}$, $a_{i,k}$ is an input node and therefore $A_{i,k}$ should match the input $\qry$.
    \item Set $k=0$ and $\ell = n$
    \item While $\ell > k+1$ do the following
        \begin{itemize}
            \item set $u=(k+\ell)/2$
            \item $P_i$ opens $A_{i,k}, A_{i,\ell}, A_{i,u}$ w.r.t. $C_i$
            \item if $k \leq k_{in}$ and $A_{i,k}$ does not match $\qry$ reject $P_i$'s claim and stop
            \item if $A_{1,u} \neq A_{2,u}$ then $\ell \gets u$ else $k \gets u$
        \end{itemize}
    \item If not stopped before, the loop ends with $a_{1,k}=a_{2,k}$ and $a_{i,\ell} \neq a_{2,\ell}$ with $\ell=k+1$. The party $P_i$ open $a_{i,j}$ and $w_{j,\ell}$ for $j \in G_{\ell}$, i.e., the activations and the weights of all the nodes that feed in node $\ell$. Accept the claim of the party $P_i$ for which $a_{i,\ell} = \sum_{j \in G_{\ell}} \phi(w_{j,\ell} a_{i,j})$.
\end{itemize}

\begin{theorem}
If at least one of the two parties is honest the protocol always terminate by accepting the correct claim.
\end{theorem}

\begin{proof}[Proof sketch]
The claim is due to the following facts:
\begin{itemize}
\item The honest prover will never be caught inside the while loop;
\item The protocol bisects through the internal nodes until it finds the \emph{first} node (in the canonical order) in which the committed traces disagree (this is due to the collision resistance of the function $H$)
\item Therefore in the last step only one of the two provers can survive the test (since $a_{1,j}=a_{2,j}$ for all the $j \in G_{\ell}$ since the topological sort guarantees that $j<\ell$.
\end{itemize}
\end{proof}

\section{Inverse Transform Attack} \label{ann_inverse_transform_attack}

\subsection{Attack Design}

Since the functions used in network computations are pseudo invertible, we utilize an inverse transform to explore the possibility that an adversary can adjust activations in a way that they will align with the claimed output. We used three different ways to invert matrices: The Moore-Penrose pseudo-inverse $\mathbf{W}^+ = (\mathbf{W}^T\mathbf{W})^{-1}\mathbf{W}^T$ for general least-squares solutions, the SVD-based inverse  $\mathbf{W}^{-1} = \mathbf{V}\mathbf{S}^{-1}\mathbf{U}^T$ for numerical stability with ill-conditioned matrices, and the regularized inverse $\mathbf{W}^{-1} = (\mathbf{W}^T\mathbf{W} + \lambda\mathbf{I})^{-1}\mathbf{W}^T$ with $\lambda=10^{-4}$ to deal with possible singularities.

With white-box access to model parameters, we apply matrix inversions one at a time, starting with the output layer and working our way back: First, we compute $\mathbf{a}^{(2)} = (\mathbf{W}^{(3)})^{-1}\mathbf{y}$, then $\mathbf{a}^{(1)} = (\mathbf{W}^{(2)})^{-1}\mathbf{a}^{2}$, and finally $\hat{\mathbf{x}} = (\mathbf{W}^{(1)})^{-1}\mathbf{a}^{1}$. On the ReLU with $h = \max(0, h)$ this attack was not successful; however for sigmoid-activated layers with $h \in [0, 1]$ it was. We study this attack on a simple ANN in \cref{ann_inverse_transform_attack} and note that this reconstruction attack approach for LLMs would be even more infeasible, due to its highly complex non-linearities (e.g., SiLU activations, skip-connection layers).

%We saw that it was not possible to recover activations even in a simpler ANN architecture with ReLU and so we did not run the same test on LLMs, which have nonlinearities such as SiLU and a more complex architecture which would make it harder to recover activations.

\subsection{Experiment Setup}

To evaluate the inverse transform attack, we utilize the same simple ANN architecture described in \cref{sec:A1:setup}. We tested 120 input samples over 3 different inversion algorithms and compared reconstructed inputs and intermediate activations against ground truth values using absolute difference metrics (minimum, maximum, and mean) at each layer. The accuracy of the reconstruction is mostly limited by the difficulty of inverting certain activation functions and numerical precision limitation, and not because of computational intensity.

\begin{table}
\centering
\caption{Percentage of Samples Passing Soundness Threshold by Method and Layer}
\label{tab:soundness_thresholds:2}
\begin{tabular}{l|l|cccc}
\hline
 &  & \multicolumn{4}{c}{Threshold} \\
\cline{3-6}
Method & Layers & $10^{-4}$ & 0.001 & 0.01 & 0.1 \\
\hline
Pseudo Inverse & $L_1$ & 0.8\% & 9.2\% & 60.8\% & 100.0\% \\
 & $L_2$ & 0.0\% & 1.7\% & 17.5\% & 73.3\% \\
 & $L_3$ & 100.0\% & 100.0\% & 100.0\% & 100.0\% \\
 & \textbf{All Layers} & \textbf{0.0\%} & \textbf{0.0\%} & \textbf{11.7\%} & \textbf{73.3\%} \\
\hline
SVD & $L_1$ & 0.8\% & 9.2\% & 60.8\% & 100.0\% \\
 & $L_2$ & 0.0\% & 1.7\% & 17.5\% & 73.3\% \\
 & $L_3$ & 100.0\% & 100.0\% & 100.0\% & 100.0\% \\
 & \textbf{All Layers} & \textbf{0.0\%} & \textbf{0.0\%} & \textbf{11.7\%} & \textbf{73.3\%} \\
\hline
Regularized & $L_1$ & 0.0\% & 5.0\% & 62.5\% & 100.0\% \\
 & $L_2$ & 0.0\% & 1.7\% & 17.5\% & 73.3\% \\
 & $L_3$ & 100.0\% & 100.0\% & 100.0\% & 100.0\% \\
 & \textbf{All Layers} & \textbf{0.0\%} & \textbf{0.0\%} & \textbf{12.5\%} & \textbf{73.3\%} \\
\hline
\end{tabular}
\end{table}

\subsection{Results}

Table \ref{tab:soundness_thresholds:2} shows that the attack was unsuccessful. The output layer ($L_3$) was able to perfectly reconstruct all $120$ test samples ($100\%$ passing even the strictest $10^{-4}$ threshold). However, the quality of the reconstruction got worse at deeper layers because of the cumulative effects of inverting activation functions. The most strict soundness criteria showed the attack's flaws: at the $10^{-4}$ threshold, no samples were able to reconstruct all layers at the same time for any inversion method. Only $0.8\%$ of samples passed at $L_1$ and $0\%$ at $L_1$. However, relaxing the threshold to weaker levels create a more successful adversary:  at the $ 0.1$ threshold, $73.3\%$ of samples achieved successful reconstruction across all layers simultaneously. In practice, note that, theoretically, our threshold could be $0$, but to account for potential completeness concerns, we can set a reasonable $10^{-4}$ threshold---and even then, the adversary is unsuccessful.

The success of the reconstruction is heavily reliant on both the characteristics of the activation function and the structure of the network. For example, while invertible activation functions like sigmoid theoretically allow for inverse computation, our experiments utilized ReLU activations that cause irreversible information loss at negative values, which accounts for the observed degradation in reconstruction.
Additionally, the dimensional relationships between layers---especially when a layer has more output neurons than input neurons---make systems that are overdetermined, which makes reconstruction more accurate even though these relationships are not linear. All three inversion methods yield remarkably consistent results, with only minor variations in the regularized approach at $L_1$. This confirms that reconstruction quality is fundamentally limited by information loss at activation boundaries, numerical precision, and model architecture, rather than by the choice of inversion technique.
The fact that reconstruction fidelity gets worse as depth increases---from perfect at the output layer to progressively worse---supports our analysis.

\section{Swap Attack on LLMs} \label{swap_attack}

\subsection{Attack Design}
\begin{comment}
Similar to gradient descent used during training to adjust activations, we investigate whether an adversary can manipulate activations to produce malicious outputs. Unlike the activation reconstruction attack (\cref{sec:fs-exps}), the swap attack employs a logit injection strategy where the adversary first swaps the logit values of the highest and lowest probability tokens to create a malicious target output, then optimizes intermediate activations to achieve this adversarial target while preserving the malicious logits as the desired output.
\end{comment}
The previous attacks test strong other-model soundness (\cref{sec:strong-om-snd}), where the adversary's output comes from a specific substitute model. We now consider a more aggressive attack that is closer to full soundness: the adversary takes the true model's output and directly perturbs it (by swapping the highest and lowest probability tokens), then tries to adjust intermediate activations to match. Unlike the previous attacks, no explicit substitute model is involved in this attack. The adversary constructs both the malicious output and the forged trace. Specifically, the swap attack employs a logit injection strategy in which the adversary first swaps the logit values of the highest- and lowest probability tokens to create a malicious target output and then optimizes intermediate activations to achieve this adversarial target while preserving the malicious logits as the desired output.

Specifically, we first compute the original model output: $\outfn \leftarrow \model(\qry)$, and create malicious target logits $\mathbf{l}_{\advmodel}$ by swapping the values at indices $\textit{argmax}(\outfn)$ and $\textit{argmin}(\outfn)$. We then initialize an optimizable activation tensor $a_{\advmodel}$ from the output of the first block's input normalization and inject it into the forward pass after the embedding layer. The adversary uses gradient descent to minimize:
$$L = (\outfn - \outfn_{\advmodel})^2 + \lambda \|a^{(2)}\|^2$$

where $a^{(2)}_{\advmodel}$ is updated as:

$$a^{(2)}_{\advmodel,i} = a^{(2)}_{\advmodel,{i-1}} - \alpha \frac{\partial L}{\partial a^{(2)}_{\advmodel,{i-1}}}$$

with $\alpha$ being the learning rate and $\lambda = 0.001$ for L2 regularization. During optimization, only the last token's activations are updated (gradients for other positions are zeroed), and we iterate until convergence or up to 10,000 steps. This approach tests whether manipulating post-embedding, pre-transformer activations can successfully induce malicious model behavior.
\subsection{Experiment Setup}
We used the same experimental setup as described in \cref{sec:A2}, with 290 forward hooks across all 32 transformer blocks capturing intermediate activations for 25 text prompts. However, unlike the activation reconstruction attack, the swap attack modifies the optimization objective and injection point.

For each input prompt, we first generated malicious target logits by swapping the highest and lowest probability output tokens. We then performed 10 independent reconstruction attempts, optimizing activations at the first transformer block's input normalization layer (rather than at the embedding layer) using the same Adam optimizer configuration (learning rate 0.01, L2 regularization $\lambda = 0.001$) with up to 5,000 iterations or until convergence ($\textit{loss} \leq 10^{-4}$).

Each reconstruction was evaluated on 5,001 paths (5,000 random + 1 optimized), yielding approximately 3.75M path evaluations (25 inputs $\times$ 30 reconstructions $\times$ 5,001 paths) to measure separation values between original and adversarially modified activations.

\subsection{Results}

\begin{figure}[t]
    \centering
    \includegraphics[width=0.7\linewidth]{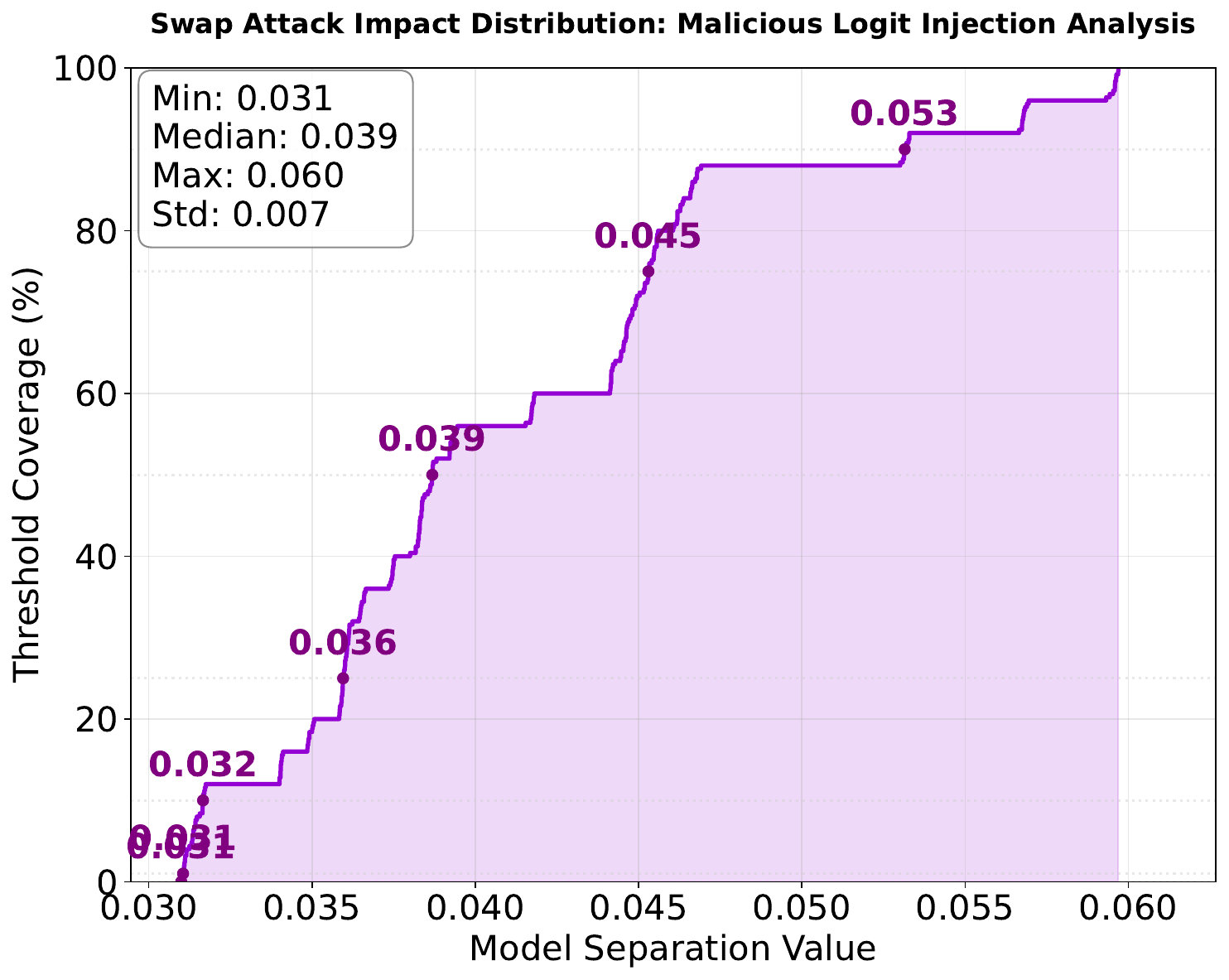}
    \caption{Model separation values (\cref{eqn:mformula}) for Llama-2-7b-hf-chat at last token position, with neurons randomly selected per layer and values aggregated by mean per path. Red dots mark key percentile points (x: separation value, y: percentage of tests below that threshold).}
    \label{fig:swap_attack_soundness}
\end{figure}

Our results, \cref{fig:swap_attack_soundness}, show the model separation values between adversarially-modified activations and the original activations under the logit swap attack.

We find that across evaluations, the swap attack achieved significantly lower separation values compared to direct activation reconstruction. The minimum observed separation value was 0.031, with a cumulative distribution showing a mean of 0.041, standard deviation of 0.010, and median of 0.039. Despite the adversary's ability to directly manipulate intermediate activations and target specific malicious outputs through logit swapping, the separation values remained bounded.

Key percentiles illustrate the attack's limited impact: the 1st percentile is 0.031, the 10th percentile is 0.032, the 25th percentile is 0.036, the 75th percentile is 0.045, and the 90th percentile reached 0.053. While these separation values are lower than those observed in direct activation reconstruction (\cref{sec:A2}), they still demonstrate that the model maintains detectable differences between original and adversarially-manipulated activation paths. Even when the adversary can optimize activations at intermediate layers with full model access and relaxed computational constraints, the high-dimensional nature of transformer activations prevents perfect replication of the original forward pass.

These findings show that while manipulating post-embedding activations is more effective than reconstructing inputs from scratch, 100\% of swap attack attempts still produced measurable activation deviations, with the best-case maintaining a separation value of 0.031~-- an order of magnitude smaller than activation reconstruction but still indicating detectable differences in the activation space.

\section{Additional Model Separation Experiments} 
\label{sec:addl-sep}
\begin{figure}[t]
	\centering
	\includegraphics[width=0.70\linewidth]{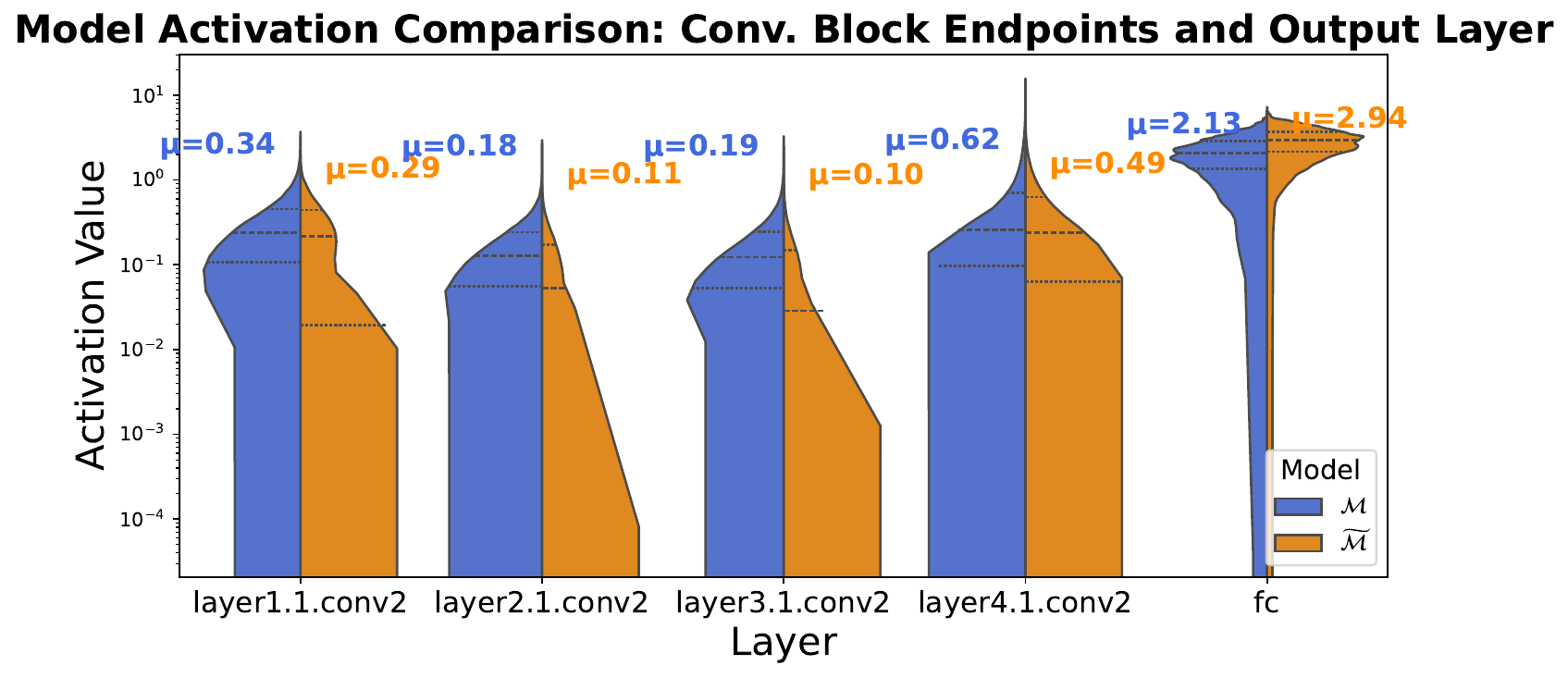}
		\caption{Activation distributions for path-selected neurons. Split violin plots show the probability density for neurons traversed during path selection across models trained on identical classification but different dog image subsets.}
	\label{fig:violin-v2}
\end{figure}

\begin{figure}[t]
	\centering
	\includegraphics[width=0.70\linewidth]{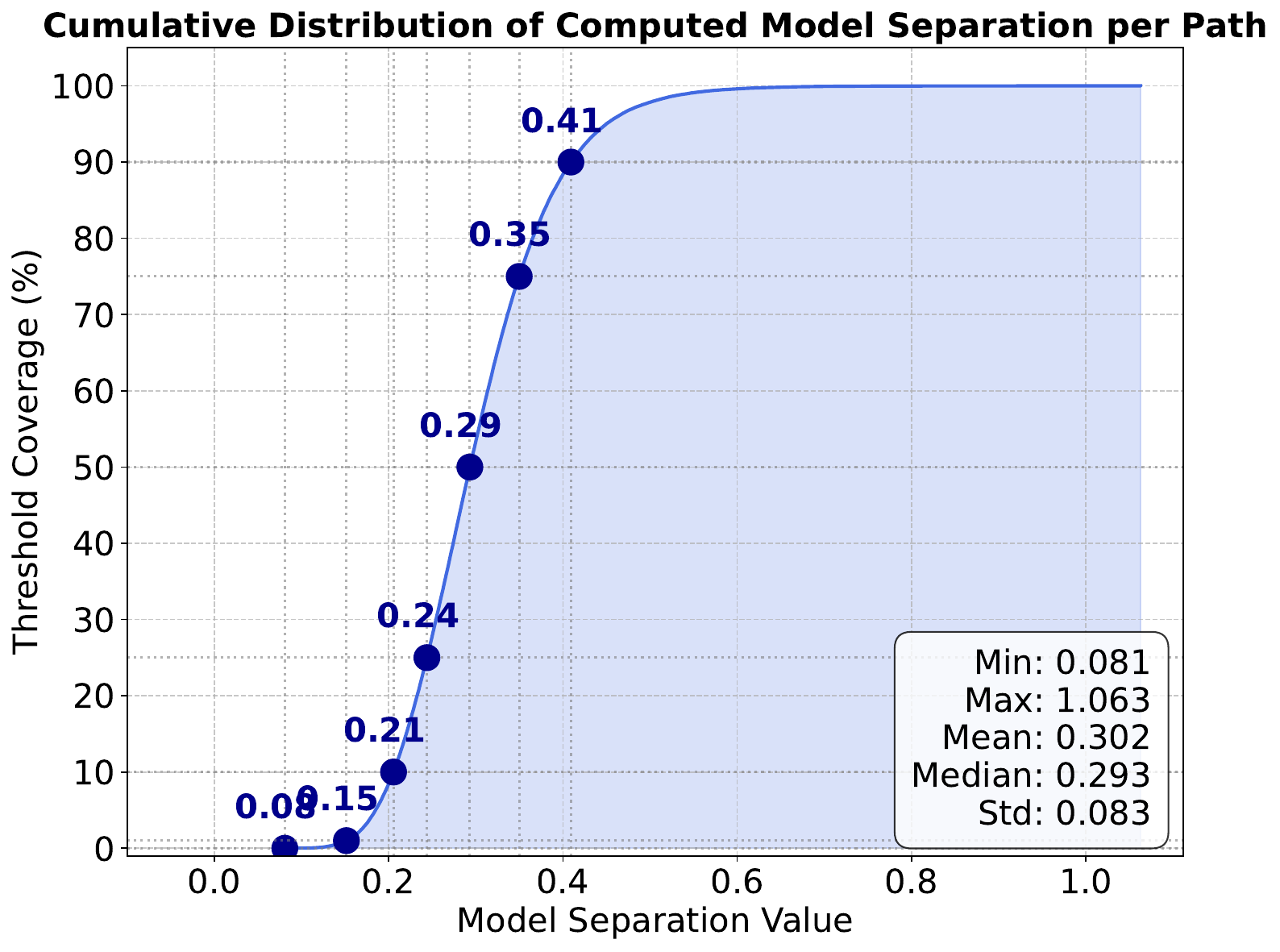}
	\caption{Cumulative distribution of model separation values. Shows separation magnitudes for path-selected neurons when both models perform dog vs cat classification with different dog training subsets.}
	\label{fig:meanneuron-v2}
\end{figure}

\begin{figure}[t]
	\centering
	\includegraphics[width=0.65\linewidth]{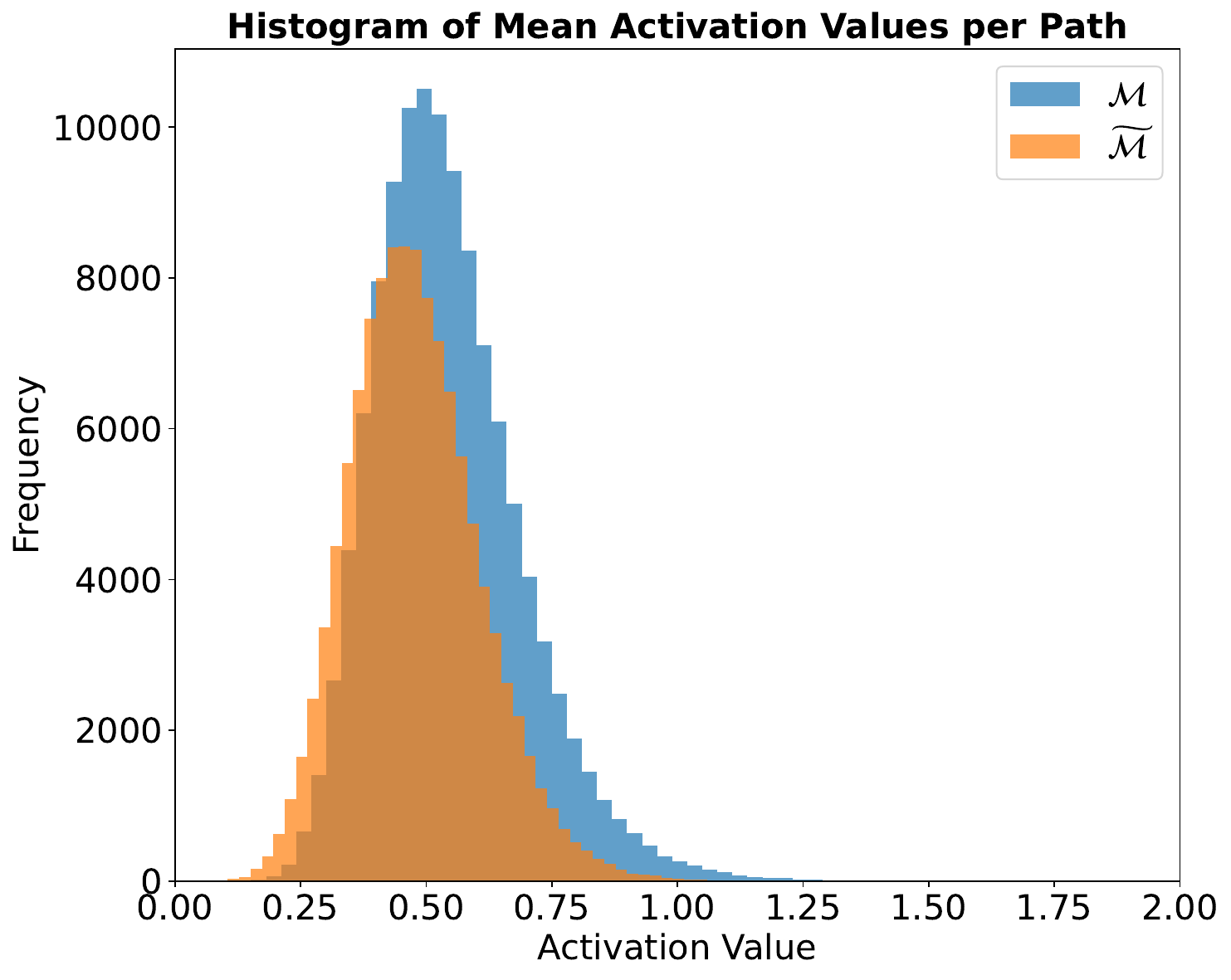}
	\caption{Mean activation values per path for dog vs cat classification with different dog training subsets. Each point represents the average activation across layers for one image-round pair.}
	\label{fig:histogram-v2}
\end{figure}
To ensure the robustness and generalization of our trace separation findings, we conducted two additional sets of experiments that study if model separation holds for classifiers trained on the same task and for classifiers trained at a higher granularity.

\subsection{Model Separation in Classifiers for the Same Task}
\label{sec:addl-sep:same-task}

Our first extension is to see if we could measure trace separation on models trained on the same classification task. We replicated the experimental setup in \cref{sec:resnet-sep} with two classifiers, but changed it as follows: we partitioned the dog dataset into disjoint subsets $\textrm{dog}_A$ and $\textrm{dog}_B$ (each with 1,500 images). $\model$ was given the entire cat dataset plus $\textrm{dog}_A$ as its fine-tuning set, and $\advmodel$ was given the same cat dataset but with $\textrm{dog}_B$ instead. So, both $\model$ and $\advmodel$ have the same classification task~-- distinguishing between cats and dogs~-- with a shared class but different subsets of data for the other task. 

We reserved a pool of 2,509 validation images for trace separation tests. The minimum separation value, computed using \cref{eqn:mformula}, is 0.081, with a mean of 0.302 and median of 0.293 across all paths. Notably, this is a higher degree of separation than shown in our initial experiment in \cref{sec:resnet-sep}.

The results demonstrated in Figure~\ref{fig:violin-v2}, the activation distributions for path-selected neurons, show similar patterns to our initial experiment in \cref{sec:resnet-sep}. The initial feature extraction layers (conv1) show virtually no difference (JS: 0), as both models learned comparable low-level features. However, the distributions expand substantially in deeper layers, with peak divergence occurring at layer3.1.conv2 (JS: 0.269). The final classifier layer (fc) maintains significant divergence (JS: 0.268), indicating that each model developed distinct decision boundaries despite solving identical tasks.

The cumulative distribution analysis (Figure~\ref{fig:meanneuron-v2}) reveals that 25\% of paths exhibit separation values below 0.244, while 75\% remain below 0.350, demonstrating consistent measurable differences across the vast majority of inference paths. The mean activation histogram (Figure~\ref{fig:histogram-v2}) confirms that these differences are clearly separable across the entire model. Thus, the results from this experiment demonstrate that even when models solve the same classification problem and share one class's training data, differences in the other class's training images induce measurable and consistent trace separation. Our results are in line with related results in understanding neural networks~\cite{zhang2022neural}. Indeed, our results show that trace separation fundamentally captures dataset-level variations, instead of merely distinguishing classification objectives.

\subsection{Model Separation in Classifiers for Granular Tasks}
\label{sec:addl-sep:granular}

\begin{figure}[t]
	\centering
	\includegraphics[width=0.70\linewidth]{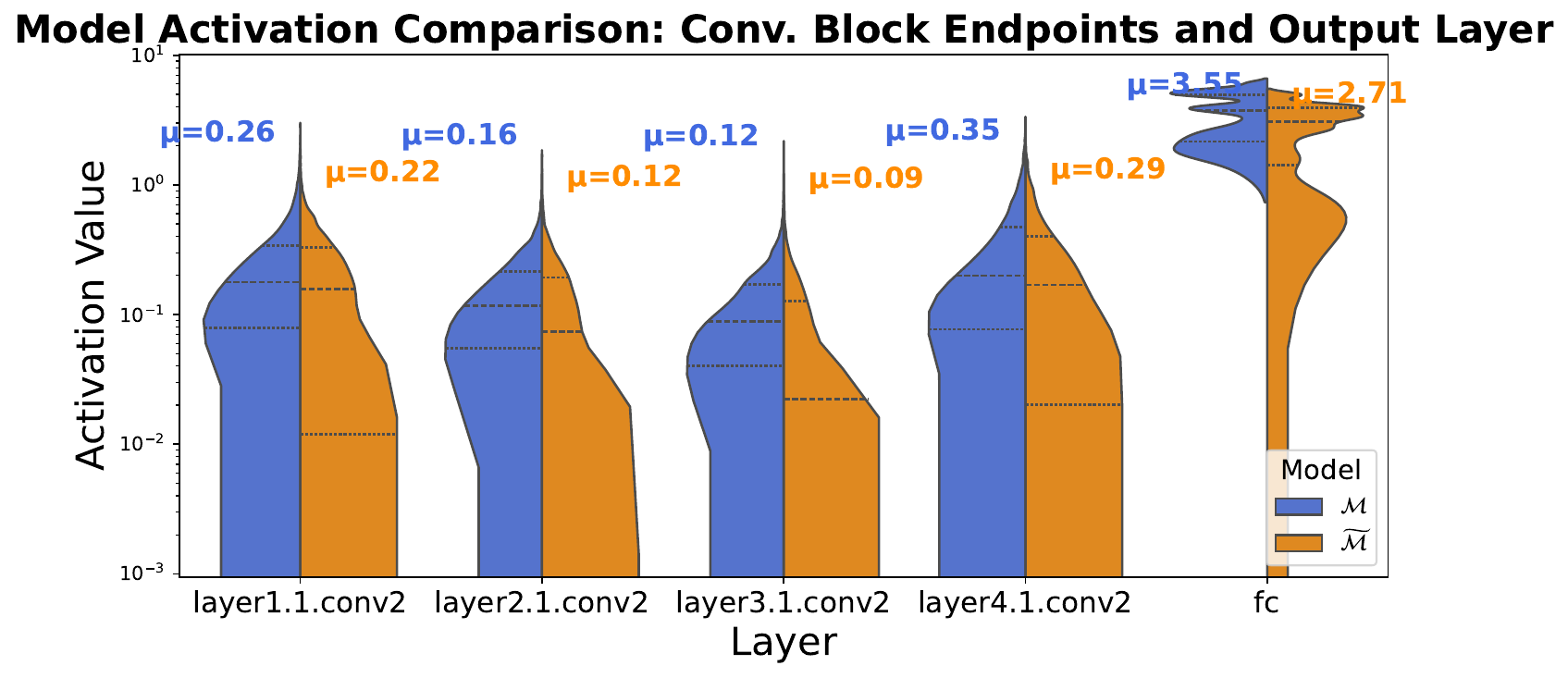}
	\caption{Activation distributions for path-selected neurons. Split violin plots show the probability density for neurons traversed during path selection across fine-grained dog breed classification models.}
	\label{fig:violin-v3}
\end{figure}

\begin{figure}[t]
	\centering
	\includegraphics[width=0.70\linewidth]{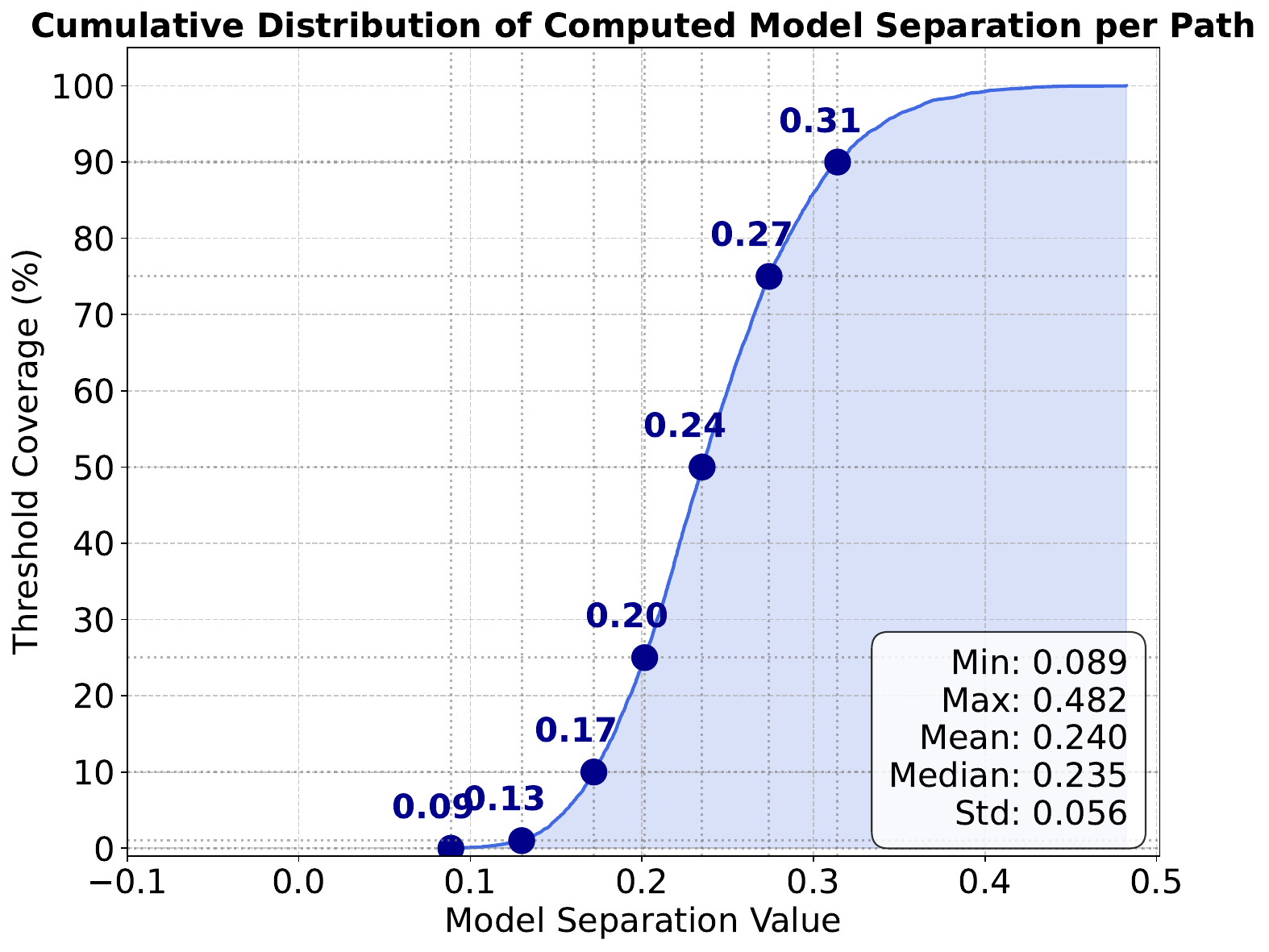}
	\caption{Cumulative distribution of model separation values. Shows separation magnitudes for path-selected neurons in fine-grained breed classification with Scottish Deerhound as the shared class.}
	\label{fig:meanneuron-v3}
\end{figure}

\begin{figure}[t]
	\centering
	\includegraphics[width=0.65\linewidth]{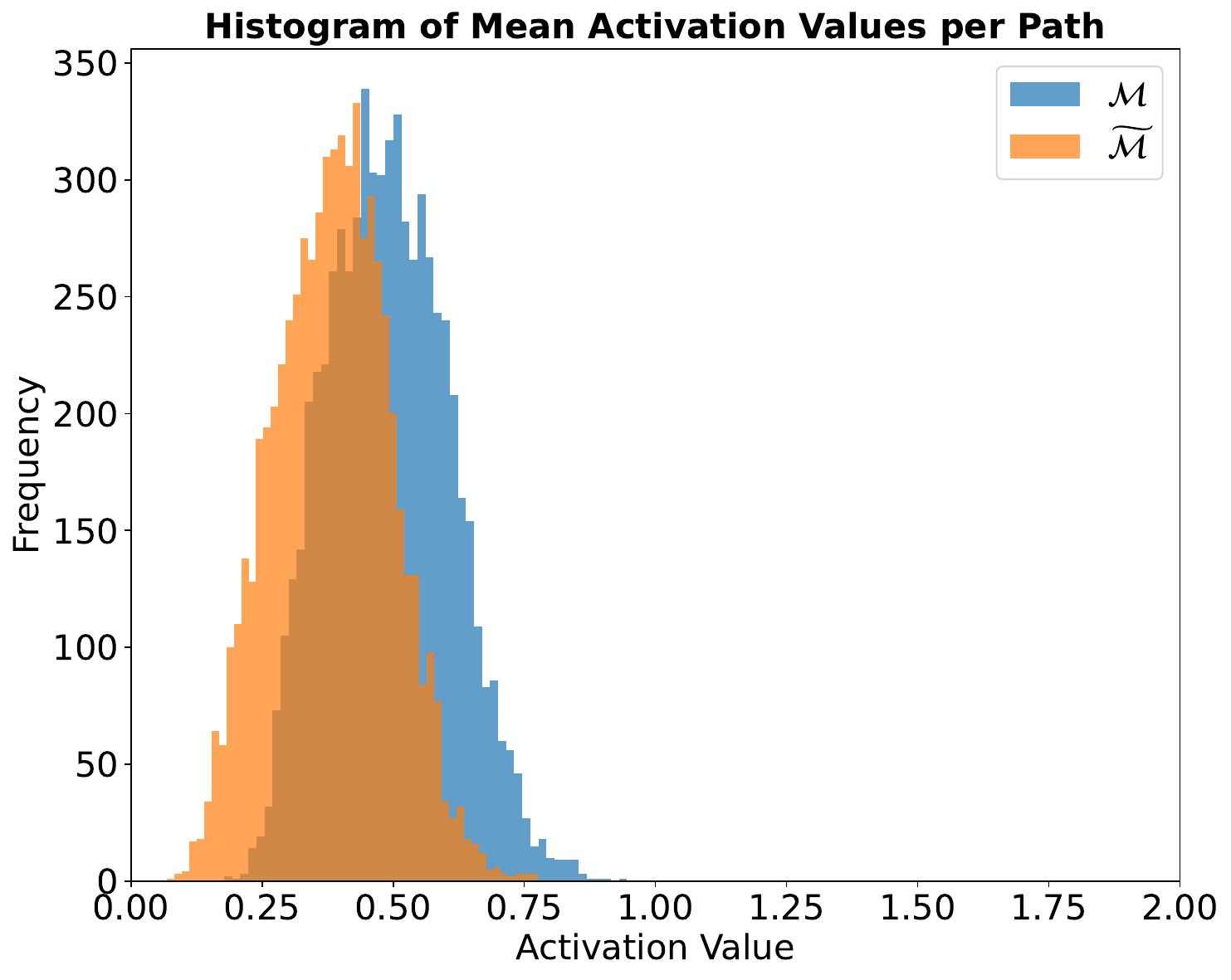}
	\caption{Mean activation values per path for fine-grained dog breed classification models. Each point represents the average activation across layers for one image-round pair in the breed classification experiment.}
	\label{fig:histogram-v3}
\end{figure}

Our second experiment investigates if we can see trace separation for models trained on finer-grained characteristics than our dog-cat-squirrel experimental setup. So, instead of animal classification, we conducted dog breed classification using images from the Stanford Dogs Dataset.\footnote{\url{http://vision.stanford.edu/aditya86/ImageNetDogs/}} Specifically, $\model$ distinguished between Entlebucher and Scottish Deerhound breeds, while $\advmodel$ classified Bernese Mountain Dog versus Scottish Deerhound breeds. The shared Scottish Deerhound class enables comparison analogous to our trace separation setup in \cref{sec:resnet-sep}. Despite the visual similarity between these dog breeds, we observed robust trace separation across 26 test images evaluated over 250 rounds each (6,500 paths total), with minimum separation of 0.089, mean of 0.240, and median of 0.235. 

As shown in Figure~\ref{fig:violin-v3}, the activation distributions reveal minimal divergence in early layers (conv1: 0.000) but substantial separation in deeper layers, with peak divergence at layer3.1.conv2 (JS 0.287). Remarkably, the final classifier exhibited highest JS (0.719), indicating that even fine-grained classification tasks induce strong divergence in decision boundaries. The cumulative distribution (Figure~\ref{fig:meanneuron-v3}) and mean activation analysis (Figure~\ref{fig:histogram-v3}) confirm consistent separation patterns across all inference paths~-- even greater than the ones in \cref{sec:resnet-sep}. The fine-grained tasks clearly amplify separation, likely due the models developing highly structured decision boundaries. Our results in this experiment establish trace separation as a general property detectable even on finer-grained classification tasks.

\section{Strong Other-Model Soundness}
\label{sec:strong-om-snd}

In this section we define \emph{strong other-model soundness}, a security notion that sits between other-model soundness and full soundness. Like other-model soundness, it restricts the adversary's \emph{output} to that of a specific alternative model $\advmodel$. However, unlike other-model soundness---where the adversary must commit to the honest evaluation trace of $\advmodel$---strong other-model soundness allows the adversary to construct the trace \emph{arbitrarily}, as in the full soundness setting.

This notion captures a natural threat scenario: an adversary who runs a different model to produce the output (see \cref{sec:motivation-other-model-snd}), but who may attempt to craft a fake trace that passes the verification test (rather than simply submitting the honest trace of $\advmodel$). It is strictly stronger than other-model soundness (since the adversary has more freedom in choosing the trace) and strictly weaker than full soundness (since the adversary's claimed output is constrained to $\advmodel(\qry)$ rather than being arbitrary).

\subsection{Idealized Definition}

We extend \cref{def:idealized} with the following additional soundness notion.

\paragraph{Strong other-model soundness.} Let $\sndeps: \NN \to \RR$ be a function and let $\left( \Dmodel, \Dqry, \Delta \right)$ be as in \cref{def:idealized} and $\advDmodel$ be a model distribution (potentially distinct from $\Dmodel$). We say that the scheme satisfies $(\delta,\sndeps)$-strong-other-model-soundness w.r.t.\ $\left( \Dmodel, \Dqry, \Delta \right)$ and $\advDmodel$ if for any adversary $\adv$ and for any security parameter $\lambda \in \NN$, we have:
\begin{align*}
& \multilinePr{
	& \Delta[\model(\qry),\advmodel(\qry)]>\delta \;  \wedge \\
	& \outfn(\advtrace) = \advmodel(\qry) \; \wedge \\
	& \testTrace^{\model, \advtrace}\left( 1^\lambda, \qry \right) = 1
}{
	& \model \sample \Dmodel\\
	& \advmodel \sample \advDmodel\\
	& \qry \sample \Dqry\\
	& \advtrace  \gets \adv\left( 1^\lambda, \model, \qry, \advmodel \right)
} \leq \sndeps(\lambda)
\end{align*}

Informally, strong other-model soundness states that, except with $\sndeps$ probability, no adversary can fool the test into accepting a trace whose output matches that of a different model $\advmodel$, even when the adversary is free to construct the trace in any way it chooses. The key difference from other-model soundness is in the last line of the sampling: the trace $\advtrace$ is produced by an adversary $\adv$ (with access to $\model$, $\qry$, and $\advmodel$) rather than being the honest evaluation $\Eval(\advmodel, \qry)$. The only constraint we impose for the adversary's output is that the trace it produces should be consistent with the output of the ``wrong'' model $\advmodel$.

\begin{remark}[Relation to other notions]
\label{rem:strong-om-relations}
The following implications hold:
\begin{itemize}
	\item Full soundness implies strong other-model soundness: any adversary in the strong other-model soundness game can be simulated by a full soundness adversary that ignores the constraint on the output. More precisely, a full soundness adversary can sample $\advmodel \sample \advDmodel$ on its own and run the strong-other-model-soundness adversary internally.
	\item Strong other-model soundness implies other-model soundness: the other-model soundness adversary is a special case where the trace is restricted to $\advtrace = \Eval(\advmodel, \qry)$.
\end{itemize}
\end{remark}

\subsection{Cryptographic Definition}

We similarly extend \cref{def:verifiable-proper}.

\paranoindent{Strong other-model soundness.} Let $\sndeps: \NN \to \RR$ be a function and let $\left( \Dmodel, \Dqry, \Delta \right)$ be as in \cref{def:verifiable-proper} and $\advDmodel$ be a model distribution (potentially distinct from $\Dmodel$). We say that the scheme satisfies $(\delta,\sndeps)$-strong-other-model-soundness w.r.t.\ $\left( \Dmodel, \Dqry, \Delta \right)$ and $\advDmodel$ if for any efficient adversary $\adv = \left( \adv_1, \adv_2 \right)$ and for any security parameter $\lambda \in \NN$, we have:
\begin{align*}
	& \multilinePr{
		& \Delta[\model(\qry),\advmodel(\qry)]>\delta \  \wedge \\
	&\ \verify(\pp, \cm, \qry, \tilde{y},\!\\
	& \ \qquad \left(  \pi_1, \pi_2 \right), \rho ) \!=\! 1\!
	}{
		& \model \sample \Dmodel\\
		&\pp \gets \genparams(1^\lambda)\\
		&  \cm \gets \modelcommit\left( \pp, \model \right)\\
		& \advmodel \sample \advDmodel\\
		& \qry \sample \Dqry\\
		& \tilde{y} := \advmodel\left( \qry \right)\\
		&  \left( \pi_1, \mathsf{state} \right)  \gets \adv_1\left( \pp, \model, \qry, \advmodel \right)\\
		& \rho \sample \chalSpace\\
		& \pi_2 \gets \adv_2(\mathsf{state}, \rho)
	} \leq \sndeps(\lambda)
\end{align*}

Note that, compared to the other-model soundness definition of \cref{def:verifiable-proper}, the proof messages $(\pi_1, \pi_2)$ are now produced by an adversary $\adv$ rather than by the honest $\prove$ algorithm run on $\advmodel$. However, the claimed output $\tilde{y}$ remains fixed to $\advmodel(\qry)$.
\section{Estimating thresholds and probabilities for other-model soundness}
\label{appx:methods}

In this section, we describe a method for estimating the threshold values used in Definitions 2 and 3. The procedure estimates the \textit{Trace Separation} parameters ($\delta_{sep}^{out}, \delta_{sep}^{trace}$) and validates the \textit{Test Soundness} ($\teps$). The approach we describe here works for both classification and LLM tasks.

Here, we present the simplified case where there is a single honest model $\mathcal{M}$ and a single adversarial model $\tilde{\mathcal{M}}$. In practice, the adversarial model is sampled from a distribution $\advDmodel$ (e.g., quantized models, pruned models, or specific adversarial fine-tunings)\footnote{See also discussion in \cref{sec:motivation-other-model-snd}.}. To ensure robust guarantees, the estimation procedure described below should be repeated on a diverse set of adversarial models $\{\tilde{\mathcal{M}}_1, \dots, \tilde{\mathcal{M}}_K\}$ sampled from $\advDmodel$. The final threshold parameters $(\delta_{sep}^{out}, \delta_{sep}^{trace})$ must be selected as the minimum observed separation values across all $K$ \footnote{K represents the number of distinct variations in the adversarial model. It's set to a small number (e.g., 5-10).} instances to ensure that the \texttt{RandPathTest} is robust against worst case models.

\subsection{Preliminaries and statistical layer selection}

We start the estimation procedure by sampling a large set of $N$ queries $\{qry_j\}_{j=1,\dots,N}$ and to generate the logit vectors and activation traces for both the honest model $\mathcal{M}$ and adversarial model $\tilde{\mathcal{M}}$.

\begin{enumerate}
    \item \textbf{Statistical Layer Filtering:} We first compute the Jensen-Shannon (JS) divergence between the activation distributions of $\mathcal{M}$ and $\tilde{\mathcal{M}}$ for each layer. We define the set of valid layers $\mathcal{L}_{valid}$ as those where $JS > 0.05$, excluding noisy input layers. The JS threshold is chosen based on empirical observations where significant layers consistently exhibited divergence values well above this floor (e.g., $>0.09$ in ResNet experiments, as detailed in \cref{sec:resnet-sep}).
    
    \item \textbf{Metric 1: Output Difference ($dOut$).} Euclidean distance between logit vectors:
    \begin{equation*}
        dOut(j) := ||\mathbf{y}_j - \tilde{\mathbf{y}}_j||_2
    \end{equation*}
    
    \item \textbf{Metric 2: Trace Difference ($dTrc$).} Mean absolute difference over $\mathcal{L}_{valid}$. Let $D_l$ be the dimension of layer $l$:
    \begin{equation*}
        dTrc(j) := \frac{1}{|\mathcal{L}_{valid}|} \sum_{l \in \mathcal{L}_{valid}} \frac{1}{D_l} ||trc_{j,l} - \tilde{trc}_{j,l}||_1
    \end{equation*}
\end{enumerate}

We end up with a set of $N$ tuples
$$\mathcal{D} = \{(dOut(j), dTrc(j))\}_{j=1,\dots,N}$$

\subsection{Generating candidate thresholds}

We first identify candidate pairs $\mathcal{C} \subseteq \mathcal{D}$ that satisfy the separation property with high confidence ($1-\epsilon_{sep}$). Algorithm ~\ref{alg:sep_estimation} takes as input the dataset we prepared from the queries and returns a subset of candidate threshold values for which the trace separation condition holds.

\begin{algorithm}
\caption{Generating Candidates given target $\epsilon_{sep}$}
\label{alg:sep_estimation}
\begin{algorithmic}[1]
\State \textbf{Input:} Dataset $ \mathcal{D} = \{(dOut(j), dTrc(j))\}_{j=1}^N$, Target $\epsilon_{sep}$ (e.g., 0.01), Noise Floor $\delta_{noise}$ ($10^{-4}$)
\State Sort dataset by $dOut(j)$.
\State candidates $\leftarrow []$
\For{each unique value $x$ in sorted $dOut$}
    \State Let $S_x$ be the subset of samples where $dOut(j) \ge x$
    \If{$|S_x|$ is too small} \textbf{break} \EndIf
    \State Let $q$ be the $\epsilon_{sep}$-quantile of $dTrc$ values in $S_x$
    \If{$q > \delta_{noise}$}
        \State Append pair $(x, q)$ to candidates
    \EndIf
\EndFor
\State \textbf{Output:} List of \texttt{candidates}
\end{algorithmic}
\end{algorithm}

\subsection{Estimating path test soundness}

We now describe a procedure that given a specific trace threshold $\delta$ estimates the probability that the path test fails to catch an incorrect trace that differs by at least $\delta$. Algorithm ~\ref{alg:tst_estimation} takes as input the dataset $\mathcal{D}$, a given trace separation value $\delta$, and returns an estimate of the probability of failure of the procedure $\sf RandPathTest$.

\begin{algorithm}
\caption{Estimating $\epsilon_{tst}$ for a fixed $\delta$}
\label{alg:tst_estimation}
\begin{algorithmic}[1]
\State \textbf{Input:} Dataset $\mathcal{D}$, Threshold $\delta$, Repetitions $T$ (e.g., 50)
\State $SumFailProb \leftarrow 0$, $CountValid \leftarrow 0$
\For{$j = 1, \dots, N$}
    \If{$dTrc(j) \ge \delta$} \Comment{Consider only traces with a difference $\ge \delta$}
        \State $CountValid \leftarrow CountValid + 1$
        \State $FailCount \leftarrow 0$
        \For{$k = 1, \dots, T$}
            \State Run \texttt{RandPathTest} on $(\mathcal{M}, \tilde{trc}_j)$
            \If{Test \textbf{Accepts} (False Negative)}
                \State $FailCount \leftarrow FailCount + 1$
            \EndIf
        \EndFor
        \State $ProbFail_j \leftarrow FailCount / T$ \Comment{Estimate for sample $j$}
        \State $SumFailProb \leftarrow SumFailProb + ProbFail_j$
    \EndIf
\EndFor
\State $\hat{\epsilon}_{tst} \leftarrow SumFailProb / CountValid$ \Comment{Overall estimate}
\State \textbf{Output:} $\hat{\epsilon}_{tst}$
\end{algorithmic}
\end{algorithm}

\subsection{Final parameter selection procedure}
We now connect \cref{alg:sep_estimation,alg:tst_estimation} to select the optimal parameters. We seek the tightest Output Difference ($\delta_{out}$) that enforces a Trace Difference ($\delta_{trace}$) large enough for our $\sf RandPathTest$ test to detect reliably (low $\epsilon_{tst}$).

\begin{enumerate}
    \item Fix the maximum acceptable test error $\epsilon_{target}$ (e.g., 0.05).
    \item Run \textbf{Algorithm 1} to get the list $\mathcal{C}$ of candidates.
    \item Iterate through $\mathcal{C}$:
    \begin{itemize}
        \item Let the current pair be $(\delta_{out}, \delta_{trace})$.
        \item Run \textbf{Algorithm 2} using $\delta = \delta_{trace}$.
        \item If the resulting $\hat{\epsilon}_{tst} \le \epsilon_{target}$, select this pair and \textbf{Stop}.
    \end{itemize}
    \item \textbf{Result:} The selected $(\delta_{out}, \delta_{trace})$ are the final parameters. The protocol guarantees a soundness error of approximately $(\epsilon_{sep} + \epsilon_{tst})$.
\end{enumerate}

\end{document}